\documentclass[5p,twocolumn]{elsarticle}
\usepackage{amsmath}
\usepackage{amssymb}
\usepackage{amsfonts}
\usepackage{mathtools}

\newcommand*{\id}{{\mathrm{id}}}
\newcommand*{\R}{\mathbb R}
\newcommand*{\CC}{{\mathbb{C}}}
\newcommand*{\pa}{{\partial}}
\newcommand*{\rd}{{\mathrm{d}}}
\newcommand*{\lb}{[\![}
\newcommand*{\rb}{]\!]}
\newcommand*{\Vb}{\mathbb V}
\newcommand*{\Jf}{\mathfrak J}

\newtheorem{theorem}{Theorem}
\newtheorem{lemma}[theorem]{Lemma}
\newtheorem{corollary}[theorem]{Corollary}
\newdefinition{definition}{Definition}
\newdefinition{prescription}{Prescription}
\newdefinition{remark}{Remark}
\newdefinition{example}{Example}
\newdefinition{procedure}{Procedure}
\newproof{proof}{Proof}

\allowdisplaybreaks

\journal{...}

\begin{document}
\begin{frontmatter}
\title{Direct linearisation of the non-commutative Kadomtsev--Petviashvili equations} 
\author[1]{Gordon Blower\fnref{fn2}}
\ead{G.Blower@lancaster.ac.uk}
\author[2]{Simon~J.A.~Malham\corref{cor1}\fnref{fn1,fn2}}    
\ead{S.J.A.Malham@hw.ac.uk} 
\cortext[cor1]{Corresponding author} 
\fntext[fn1]{SJAM was supported by an EPSRC Mathematical Sciences Small Grant EP/X018784/1}
\fntext[fn2]{The authors would like to dedicate this paper to Henry McKean whose mathematical work has inspired, and continues to inspire, us so much.}
\address[1]{Department of Mathematics and Statistics, Lancaster University, Lancaster LA1~4YF, UK}
\address[2]{Maxwell Institute for Mathematical Sciences, and School of Mathematical and Computer Sciences,   
  Heriot-Watt University, Edinburgh EH14~4AS, UK (5/12/2024)}

\begin{abstract}
  We prove that the non-commutative Kadomtsev--Petviashvili (KP) equation and a `lifted' modified Kadomtsev--Petviashvili (mKP) equation are directly linearisable, and thus integrable in this sense.
  There are several versions of the non-commutative mKP equations, including the two-dimensional generalisations of the non-commutative modified Korteweg--de Vries (mKdV) equation and
  its alternative form (amKdV). Herein we derive the `lifted' mKP equation, whose solutions are the natural two-dimensional extension of those for the non-commutative mKdV equation
  derived in Blower and Malham~\cite{BM}. We also present the log-potential form of the mKP equation, from which all of these non-commutative mKP equations can be derived.
  To achieve the integrability results, we construct the pre-P\"oppe algebra that underlies the KP and mKP equations.
  This is a non-commutative polynomial algebra over the real line generated by the solution (and its partial derivatives) to the linearised form of the KP and mKP equations.
  The algebra is endowed with a pre-P\"oppe product, based on the product rule for semi-additive operators pioneered by P\"oppe for the commutative KP equation.
  Integrability corresponds to establishing a particular polynomial expansion in the respective pre-P\"oppe algebra.
  We also present numerical simulations of soliton-like interactions for the non-commutative KP equation.
\end{abstract}
\end{frontmatter}

\section{Introduction}
We consider the non-commutative KP equation and a non-commutative mKP equation. 
We establish that both equations are directly linearisable, and are thus integrable in this sense.
Here by directly linearisable we mean that solutions to these nonlinear equations can be reconstructed from solutions to a linear system.
The first equation of the linear system consists of the linearised partial differential equation corresponding to the KP and mKP equations.
This is the same for both, and its solution is often nominated the scattering data.
The second equation is a linear Fredholm equation involving the scattering data, which corresponds to the Gelfand--Levitan--Marchenko (GLM) equation.
The GLM equation is slightly different in each case. To achieve this result, we construct a pre-P\"oppe algebra. 
This algebra applies to the corresponding whole hierachies of the KP and mKP equations. 
It is a quasi-Leibniz algebra over the real line, generated by compositions of left and right derivatives of kernels of semi-additive operators.
These semi-additive operators are constructed from the scattering data.
The algebra is endowed with the P\"oppe product, first introduced by P\"oppe~\cite{PKdV,PSG} for Hankel operators for the sine-Gordon and Korteweg--de Vries (KdV) equations.
P\"oppe~\cite{PKP} generalised this product to semi-additive operators for the scalar KP equation. 
Underlying this `pre-P\"oppe' product is the property that the kernels of certain combinations of semi-additive operators correspond to the finite rank product of a pair of kernels.
Establishing integrability of the KP or mKP equation in the pre-P\"oppe algebra then corresponds to establishing a specific polynomial expansion therein.  
Higher hierarchy members correspond to higher degree polynomials. 

In \emph{potential form}, the non-commutative Kadomtsev--Petviashvili equation (KP) for the field $g=g(x,y;t)$ is, 
\begin{equation}
  \mathrm{D}g_y=\tfrac13\bigl(g_t-g_{xxx}\bigr)+2g_x^2+2\pa_x^{-1}[g_y,g_x],\label{eq:KPsingform}
\end{equation}
where we set $\mathrm{D}\coloneqq\pa^{-1}_x\pa_y$, and $[f,g]\coloneqq fg-gf$ is the usual commutator.
This form of the non-commutative KP equation can be found, for example, in Gilson and Nimmo~\cite{GN}, Kuperschmidt~\cite{K},
Hamanaka and Toda~\cite{HamanakaToda}, Sakhnovich~\cite{Sakhnovich}, Toda~\cite{Toda} and Wang and Wadati~\cite{WW}.

There are several forms of the non-commutative modified Kadomtsev--Petviashvili equation (mKP).
The forms we present herein can all be derived from the log-potential mKP equation for the field $q=q(x,y;t)$, given by, 
\begin{align}
  q\mathrm{D}\bigl(q^{-1}q_y+(q^{-1}q_x)^2\bigr)=&\;\tfrac13\bigl(q_t-q_{xxx}\bigr)+q_xq^{-1}q_{xx}\nonumber\\
                                                 &\;+q_xq^{-1}q_y+q_yq^{-1}q_x.\label{eq:lpKP}
\end{align}  
Thus, for example, if we set $\hat{g}\coloneqq q_xq^{-1}$ and $\hat{f}\coloneqq q_yq^{-1}$, then $\hat{g}=\hat{g}(x,y;t)$
and $\hat{f}=\hat{f}(x,y;t)$ satisfy the modified Kadomtsev--Petviashvili equation in the following form, 
\begin{subequations}
\begin{align}
 \hat{f}_x=&\;\hat{g}_y+[\hat{g},\hat{f}],\label{eq:eqamKPassign}\\
 \hat{f}_y=&\;\tfrac13\bigl(\hat{g}_t-\hat{g}_{xxx}\bigr)+2\hat{g}\hat{g}_x\hat{g}+\{\hat{f},\hat{g}_x\}\nonumber\\
           &\;+[\hat{g},\hat{g}_{xx}]-[\hat{g}^2,\hat{f}],\label{eq:amKP}  
\end{align}
\end{subequations}
where $\{f,g\}\coloneqq fg+gf$ is the usual anti-commutator. This is the form of the non-commutative mKP equation
that can be found in, eg., Dimakis and M\"uller--Hoissen~\cite{DMH}, Gilson, Nimmo and Sooman~\cite{GNS} and Wang and Wadati~\cite{WW}.
If we set $\hat{f}\equiv0$, this form of the non-commutative mKP collapses to the alternative form of the non-commutative modified Korteweg--de Vries (amKdV) equation. 
The fields $g=g(x,y;t)$ and $f=f(x,y;t)$ given by $g\coloneqq q^{-1}q_x$ and $f\coloneqq q^{-1}q_y$ satisfy the following form of the non-commutative mKP equation,
\begin{subequations}
\begin{align}
  f_x=&\;g_y-[g,f],\label{eq:mKPassign}\\
  f_y=&\;\tfrac13\bigl(g_t-g_{xxx}\bigr)+\{g_x,f+g^2\}-\bigl[g,\mathrm{D}(f+g^2)\bigr].\label{eq:mKP}
\end{align}
\end{subequations}
If we set $f\equiv0$, this form of the non-commutative mKP collapses to the standard form of the non-commutative modified Korteweg--de Vries (mKdV) equation. 
That the fields $(\hat{g},\hat{f})$ and $(g,f)$, defined via the log-potential $q$, respectively satisfy the forms of the non-commutative mKP equation
\eqref{eq:eqamKPassign}--\eqref{eq:amKP} and \eqref{eq:mKPassign}--\eqref{eq:mKP}, can be checked by direct straightforward computation.
The relation between the mKP equations forms \eqref{eq:eqamKPassign}--\eqref{eq:amKP} and \eqref{eq:mKPassign}--\eqref{eq:mKP} via~\eqref{eq:KPsingform}
is analogous to that for the alternative non-commutative KdV equation and the standard non-commutative KdV equation via the mirror meta-KdV given,
for example, in Carillo and Schiebold~\cite{CS}.

However, there is another form of the non-commutative mKP equation that we consider herein, and for which we establish the direct linearisation. 
Defining the fields $g=g(x,y;t)$ and $f=f(x,y;t)$ by $g\coloneqq q^{-1}q_x$ and $f\coloneqq q^{-1}q_y$, as we did above, but this time fully
deriving the time-evolution equation for $f$, generates the system pair,
\begin{subequations}
\begin{align}
  \mathrm{D}\bigl(g_y-[g,f]\bigr)=&\;\tfrac13\bigl(g_t-g_{xxx}\bigr)+\{g_x,f+g^2\}\nonumber\\
                                  &\;-\bigl[g,\mathrm{D}(f+g^2)\bigr],\label{eq:pairmKPg}\\
  \mathrm{D}\pa_y(f+g^2)=&\;\tfrac13\bigl(f_t-f_{xxx}\bigr)+\{f_x,f+g^2\}\nonumber\\
                         &\;+\bigl\{g,\mathrm{D}(g_y-[g,f])\bigr\}-\bigl[g_x,g_y-[g,f]\bigr]\nonumber\\    
                         &\;-\bigl[f,\mathrm{D}(f+g^2)\bigr].\label{eq:pairmKPf}
\end{align}
\end{subequations}
Note that we also inserted \eqref{eq:mKPassign} into the lefthand side of \eqref{eq:mKP} to generate the lefthand side of \eqref{eq:pairmKPg}.
The evolution equation for $f$ is derived via a similar manner to that for $g$ in \eqref{eq:mKP}.  
We also observe, in these equations, that \eqref{eq:pairmKPg} is linear or cubic in $g$ and linear in $f$, while \eqref{eq:pairmKPf} is linear in $f$ and quadratic in $g$.
This fact is crucial to our solution ansatz, as we explain further below--see the paragraph succeeding our discussion on the P\"oppe product in Lemma~\ref{lemma:Poppeprodsemiadd}.
At this juncture, it is convenient to set $h\coloneqq\pa_x^{-1}\bigl(g_y-[g,f]\bigr)$ and re-write the non-commutative mKP pair~\eqref{eq:pairmKPg}--\eqref{eq:pairmKPf} in the following form.
\begin{definition}[Lifted mKP equations]\label{def:liftedmKP}
We define the (non-commutative) \emph{lifted mKP equations} for $h$, $g$ and $f$ to be the system of equations given by,
\begin{subequations}
\begin{align}
    h_x=&\;g_y-[g,f],\label{eq:triplemKPh}\\
    h_y=&\;\tfrac13\bigl(g_t-g_{xxx}\bigr)+\{g_x,f+g^2\}\nonumber\\
                          &\;-\bigl[g,\mathrm{D}(f+g^2)\bigr],\label{eq:triplemKPg}\\
   \mathrm{D}\pa_y(f+g^2)=&\;\tfrac13\bigl(f_t-f_{xxx}\bigr)+\{f_x,f+g^2\}+\bigl\{g,h_y\bigr\}\nonumber\\
                          &\;-\bigl[g_x,h_x\bigr]-\bigl[f,\mathrm{D}(f+g^2)\bigr].\label{eq:triplemKPf}
\end{align}
\end{subequations}
\end{definition}
Importantly, all three forms for the mKP equation above,
\eqref{eq:eqamKPassign}--\eqref{eq:amKP}, \eqref{eq:mKPassign}--\eqref{eq:mKP} and \eqref{eq:triplemKPh}--\eqref{eq:triplemKPf},
collapse to the standard mKP equation in the commutative setting.
This is immediate for the first two forms. In the case of the lifted mKP equations \eqref{eq:triplemKPh}--\eqref{eq:triplemKPf}, in the commutative case,
naturally the commutator terms vanish, and the form $(\mathrm{D}g,g,\mathrm{D}g)$, i.e.\/ where $h=f=\mathrm{D}g$,
solves the commutative version of the lifted mKP equations \eqref{eq:triplemKPh}--\eqref{eq:triplemKPf}.
This is because $\mathrm{D}\{g_x,f+g^2\}$ equals the nonlinear terms on the right in \eqref{eq:triplemKPf},
minus the nonlinear term $ \mathrm{D}\pa_y(g^2)$ from the lefthand side.

The direct linearisation of integrable nonlinear partial differential equations is a classical approach, see for example:
Ablowitz, Ramani and Segur~\cite{ARS}; Dyson~\cite{Dyson}; Miura~\cite{Miura}; Nijhoff, Quispel, Van Der Linden and Capel~\cite{NQVC};
Nijhoff~\cite{Nijhoff}; Nijhoff and Capel~\cite{NC}; P\"oppe~\cite{PSG,PKdV,PKP} and Santini, Ablowitz and Fokas~\cite{SAF}.
Also see Zakharov and Shabat~\cite{ZS,ZS2}.
There has been much recent interest in the application of the direct linearisation approach to KP equations,
see for example Fu~\cite{Fu} and Fu and Nijhoff~\cite{FN1,FN2,FN3}.
The approach to direct linearisation of the non-commutative KP and mKP equations we present herein has its origins in the work
of P\"oppe and Nijhoff, just mentioned. In particular, recently, in Blower and Malham~\cite{BM}, Doikou, Malham and Stylianidis~\cite{DMS},
Doikou, Malham, Stylianidis and Wiese~\cite{DMSWc} and Malham~\cite{MNLS,MKdV}, P\"oppe's approach to direct linearisation was extended,
given an abstract algebraic context and applied to non-commutative classical integrable systems such as the non-commutative, KdV, and
nonlinear Schr\"odinger (NLS) and modified KdV, hierarchies. 

The direct linearisation solution approach for the non-commutative KP~\eqref{eq:KPsingform} and lifted mKP equations~\eqref{eq:triplemKPh}-\eqref{eq:triplemKPf}
we adopt herein is based on P\"oppe's approach for the commutative KP equation~\cite{PKP}. Also see Marchenko~\cite{Marchenko}.
The two principal ingredients are as follows.
\begin{enumerate}
\item The linearised form of the KP equation which has the following form for $p=p(t)$:
\begin{equation}\label{eq:linearform}
  \mathrm{D}p_y=\tfrac13\bigl(p_t-p_{xxx}\bigr).
\end{equation}
\item A linear integral equation, the Gelfand--Levitan--Marchenko (GLM) equation, of the form,
\begin{equation}\label{eq:GLM}
  P=G(\id-Q),
\end{equation}
  for the solution operator $G$, or equivalently, its kernel.
  Here $P$ is the scattering operator associated with the kernel function solution $p$ from \eqref{eq:linearform}.
\end{enumerate}
In this formulation, for the non-commutative KP equation, we set $Q\coloneqq P$. For the lifted mKP equations, we set $Q\coloneqq P^2$.
The form of the kernel $p$ associated with $P$ is also important. For the non-commutative KdV and NLS/mKdV hierarchies with no $y$-dependence,
we assume that $p$ has the form $p=p(z+\zeta+x;t)$.
Here $z,\zeta\in(-\infty,0]$ parametrise the operator kernel $p$ of $P$ and we assume that $x\in\mathbb R$ is an additive parameter.
The scattering operator $P$ in this case is thus a Hankel operator. However, for the KP equation, P\"oppe~\cite{PKP} assumed that $p$ has the form,
\begin{equation}
p=p(z+x,\zeta+x;y,t).\label{eq:semiadditiveform}
\end{equation}
This is a \emph{semi-additive} form with parameters $x,y\in\mathbb R$ and $z,\zeta\in(-\infty,0]$ parametrising the operator kernel $p$ of $P$.
Note the $y$-dependence in this semi-additive form has been left free so far. We now tie this down. The commutator relations in the Lax formulation
of the non-commutative KP equations impose the following two constraints on $p$:
\begin{equation}\label{eq:constraints}
p_x=p_z+p_\zeta\quad\text{and}\quad p_y=p_{zz}-p_{\zeta\zeta}.
\end{equation}
See Dimakis and M\"uller--Hoissen~\cite{DMH}, Gilson, Nimmo and Sooman~\cite{GNS}, Kuperschmidt~\cite{K}, P\"oppe~\cite{PKP} and Wang and Wadati~\cite{WW}
for further details on these constraints. We have already accommodated the first constraint with our semi-additive assumption for the form of $p$.
The second constraint above determines the $y$-dependence of $p$. Using~\eqref{eq:constraints}, another form for the linearised equation~\eqref{eq:linearform} is,
\begin{equation}\label{eq:secondlinearform}
   p_t=4\,\bigl(p_{zzz}+p_{\zeta\zeta\zeta}\bigr).
\end{equation}

To state our main results, we introduce the following notation which we use throughout.
For a Hilbert--Schmidt valued operator $G$ with square-integrable kernel $g$,
we denote its kernel by $\lb G\rb$, in other words,
\begin{equation*}
g=\lb G\rb.
\end{equation*}
In particular, we denote the solution to the GLM equation~\eqref{eq:GLM}, with the scattering data $p$ of the form~\eqref{eq:semiadditiveform},
by $\lb G\rb$ or $g=g(z,\zeta;x,y,t)$, depending on the context.
Recall for the non-commutative KP equation~\eqref{eq:KPsingform} we set $Q\coloneqq P$.
Our first main result is as follows. We prove this in Section~\ref{sec:KP}.
\begin{theorem}[KP]\label{thm:KP}
  Suppose $G$ is the Hilbert--Schmidt valued operator solution to the linear system of equations~\eqref{eq:linearform}--\eqref{eq:GLM},
  with $Q\coloneqq P$, together with the constraints~\eqref{eq:constraints}.
  Then $g=g(0,0;x,y,t)$ given by $g\coloneqq\lb G\rb$ satisfies the non-commutative KP equation~\eqref{eq:KPsingform}.
\end{theorem}
This establishes that the KP equation~\eqref{eq:KPsingform} is directly linearisable.  

Recall for the lifted mKP equations~\eqref{eq:triplemKPh}-\eqref{eq:triplemKPf} we set $Q\coloneqq P^2$.
If we set $V\coloneqq(\id-P)^{-1}$ and $V^\dag\coloneqq(\id+P)^{-1}$, then the GLM equation~\eqref{eq:GLM}, now given by $P=G(\id-P^2)$,
can be expressed in the form,
\begin{equation}\label{eq:mKPsolnform}
G=VPV^\dag.
\end{equation}
Derivatives of $\lb G\rb$ with respect to $x$ and/or $y$ can be expressed as linear combinations of monomials of the form, 
\begin{equation}\label{eq:mKPopmonomials}
\lb VP_{a_1,b_1}VP_{a_2,b_2}V\cdots VP_{a_k,b_k}V\rb,
\end{equation}
including monomials where one or more of the $V$'s shown are replaced by $V^\dag$.
Here we used the notation,
\begin{equation}
P_{a,b}\coloneqq\pa_z^a\pa_\zeta^b P.\label{eq:derivnotation}
\end{equation}
For any monomial $W=\lb VP_{a_1,b_1}VP_{a_2,b_2}V\cdots VP_{a_k,b_k}V\rb$ we use the notation $W^\dag$ 
to denote the corresponding monomial in which each `$P$' is replaced by `$-P$', each $V$ is replaced by $V^\dag$ and
each $V^\dag$ is replaced by $V$. For any operator $W$ with such a monomial form, we set,
\begin{equation}\label{eq:skew-symmdef}
[W]\coloneqq \lb W\rb-\lb W^\dag\rb\quad\text{and}\quad \{W\}\coloneqq \lb W\rb+\lb W^\dag\rb,
\end{equation}
which we extend linearly. For convenience, we also set,
\begin{equation}
P_{[a,b]}\coloneqq P_{a,b}-P_{b,a}. \label{eq:Pbracketnotation}
\end{equation}
Note, by partial fractions, we have $[V]\equiv[G]\equiv2\lb G\rb$. Our second main result is as follows. We prove this in Section~\ref{sec:mKP}.
\begin{theorem}[Lifted mKP equations]\label{thm:mKP}
  Suppose $G$ is the Hilbert--Schmidt valued operator solution to the linear system of equations~\eqref{eq:linearform}--\eqref{eq:GLM},
  with $Q\coloneqq P^2$, together with the constraints~\eqref{eq:constraints}.
  Then $h=h(0,0;x,y,t)$, $g=g(0,0;x,y,t)$ and $f=f(0,0;x,y,t)$, respectively given by $h\coloneqq[VP_{[1,0]}V^\dag]$, $g\coloneqq[V]$ and $f\coloneqq-\{VP_{[1,0]}V^\dag\}$,
  satisfy the lifted mKP equations~\eqref{eq:triplemKPh}-\eqref{eq:triplemKPf}.
\end{theorem}
This establishes that the lifted mKP equations are directly linearisable.

The crucial property that underlies Theorems~\ref{thm:KP} and \ref{thm:mKP} is the pre-P\"oppe product. Indeed, we state it here. 
\begin{lemma}[Pre-P\"oppe product]\label{lemma:Poppeprodsemiadd}
  Assume that $P$ and $\hat P$ are semi-additive Hilbert--Schmidt operators with parameters $x,y\in\R$, and continuously differentiable kernels.
  Further assume that $G$ and $\hat G$ are Hilbert--Schmidt operators with continuous kernels. Then we have the following pre-P\"oppe product rule:
  \begin{multline*}
    \lb G(P_{0,1}\hat P\!+\!P\hat P_{1,0})\hat G\rb(z,\zeta;x,y)\\=\lb GP\rb(z,0;x,y)\,\lb \hat P \hat G\rb(0,\zeta;x,y).
  \end{multline*}
\end{lemma}
This result is straightforwardly established using the Fundamental Theorem of Calculus.
It expresses the fact that particular linear combinations of operator compositions involving derivatives of semi-additive operators are equivalent
to a finite rank product of (the kernels of) specific factors, as shown. In particular, this means that in principle, specific linear combinations of 
derivatives of the kernel $g$ of $G$ in say \eqref{eq:mKPsolnform} in the lifted mKP case, might be expressible as finite rank products of lower order derivatives of $g$.
Note, the time derivative of $G$ is replaced by spatial derivatives via the linear equation~\eqref{eq:secondlinearform} for $p$.
As we shall see, the existence of such polynomial expansions corresponds to the integrability of the non-commutative KP and lifted mKP equations;
see Sections~\ref{sec:KP} and \ref{sec:mKP}.

Let us now motivate and explain the form of the lifted mKP equations~\eqref{eq:triplemKPh}-\eqref{eq:triplemKPf}.
In Blower and Malham~\cite{BM}, the authors sought and proved that the class of functions of the form $g\coloneqq[G]\equiv[V]$ solved the whole non-commutative NLS/mKdV hierarchy.
In that case the GLM equation is the same as \eqref{eq:GLM} with $Q=P^2$, while $P$ is a Hankel operator that satisfies the linearised form of the appropriate hierarchy member.
It is therefore natural to seek solutions of this form to the non-commutative mKP equation, with $P$ a semi-additive operator as described above.
It becomes immediately apparent that if we seek a class of solutions to the non-commutative KP equation of the form $g\coloneqq[G]\equiv[V]$, since $[V]$
is a skew form, then \eqref{eq:mKPassign} cannot be directly satisfied. The product of a skew form and a symmetric form is a skew form,
while the product of a skew form and a skew form is a symmetric form. See Remarks~\ref{rmk:skewsymmproducts} and \ref{rmk:explainform}.
In equation \eqref{eq:mKPassign}, we must pick $f$ to have a symmetric form so that the product in the term $[g,f]$ generates the skew form $g_y$,
but then we meet a contradiction as $f_x$ is not a skew form.
The substitution that transforms the non-commutative KP equations \eqref{eq:mKPassign}--\eqref{eq:mKP} to
\eqref{eq:pairmKPg}--\eqref{eq:pairmKPf} and then equivalently the lifted mKP equations \eqref{eq:triplemKPh}-\eqref{eq:triplemKPf} then
fixes the degrees of all the terms therein allowing $f$ to be a symmetric form and $h$ to be a skew form.
From this perspective, we can think of the lifted mKP equations as one natural generalisation of the non-commutative mKdV equation.
However, there is even further justification in terms of a \emph{Miura transformation}.
There is a natural Miura transformation of the form $g_x+g^2+f=2g_x^{\mathrm{KP}}$, relating the components $g$ abnd $f$ of the lifted mKP equations
and the solution $g^{\mathrm{KP}}$ of the non-commutative KP equation~\eqref{eq:KPsingform}. If $g^{\mathrm{KP}}$ has the form in Theorem~\ref{thm:KP}
and $g$ and $f$ have the forms in Theorem~\ref{thm:mKP}, then this relation holds exactly. Further, this Miura transformation collapses 
to the classical one relating the non-commutative mKdV and KdV equations when $f\equiv0$. See Section~\ref{sec:mKP} for more details.

The KP and mKP equations have an incredibly rich history and structure, with many connections to parallel branches of mathematics and theoretical physics.
The KP equation was originally derived as a natural generalisation of the KdV equation to two dimensions by Kadomtsev and Petviashvili~\cite{KadomtsevPetviashvili},
with applications to long waves in shallow water.
It has further applications to nonlinear optics (Pelinovsky, Stepanyants and Kivshar~\cite{PSK}) as well as to ferromagnetism and Bose--Einstein condensates.
Following the characterisation by Dyson~\cite{Dyson} of KdV solutions in terms of Fredholm determinants,
it was shown in the nineteen eighties how solutions to the KP equation could be formulated in terms of Fredholm Grassmannians and their associated determinantal bundle,
see Sato~\cite{SatoI,SatoII}, Miwa, Jimbo and Date~\cite{MJD}, Hirota~\cite{Hirota}; also see Segal and Wilson~\cite{SW} for analogous results for the KdV hierarchy.
Connections between the KP equation, Jacobians of algebraic curves and theta functions were also made at that time, see Mulase~\cite{Mu94} and Mumford~\cite{Mumford}.
Also see Ball and Vinnikov~\cite{BV}, Cheng, Li and Tian~\cite{CLT}, Ercolani and McKean~\cite{EM} and McKean~\cite{Mc87}. 
Global well-posedness for periodic square-integrable solutions given such data was established by Bourgain~\cite{Bourgain},
as well as for smooth solutions for smoother data. For more on this see Section~\ref{sec:discussion}.
There are also well-known rational solutions, see for example, Kasman~\cite{Kasman}, Pelinovsky~\cite{Pelinovsky} and P\"oppe~\cite{PKP}.
The classification of all possible soliton interactions, paramaterised via a finite-dimensional Grassmannians, was completed by Kodama~\cite{Kodama}.
Recent results have also exended the long list of connections between the KP equation and other fields.
For example, Bertola, Grava and Orsatti~\cite{BGO} establish a connection between the $\tau$-functions of the KP and nonlinear Schr\"odinger hierarchies. 
For another striking example, connections between the KP equation and the KPZ equation for the random growth off a one-dimensional substrate,
have also been recently established by Quastel and Remenik~\cite{QR}.
Also see Zhang~\cite{Zhang} for the further connection to the Baik--Rains distribution. 
These results are connected to Tracy-Widom distributions and solutions to the KP equation, see Tracy and Widom \cite{TW94,TW96,TW03}.
Also see McKean~\cite{McKean11}.
More general results in Quastel and Remenik~\cite{QR} include the matrix KP equation.
There are also supersymmetric extensions of the KP and mKP equations,
see Brunelli and Das~\cite{BrunelliDas}, Delduc, Gallot and Sorin~\cite{DGS}, Jia, Chen and Lou~\cite{JCL},
Manin and Radul~\cite{ManinRadul}, Nishino~\cite{Nishino} and Prykarpatski, Kycia and Dilnyi~\cite{PKD}. 
In 1990 Witten~\cite{Witten} conjectured a connection between partition functions in string theory and the $\tau$-function of the KdV hierarchy,
which was proved by Kontsevich~\cite{Kontsevich}. Also see Cafasso and Wu~\cite{CafassoWu}.
Paniak~\cite{Paniak} studied the extension of this connection to that between the non-commutative KP equation and D-branes.
See in particular, Hamanaka~\cite{Hamanaka06,Hamanaka} and Harvey~\cite{Harvey}.
In this context, due to quantisation of the phase space or a non-commutative geometry interpretation, non-commutativity refers to when
the independent coordinates, such as $x$, are non-commutative.
However, via the `Moyal product', the derivatives with respect to the independent coordinates $x$ and $y$ can be interpreted classically, but the product
between the dependent field and its derivatives with respect to $x$ and $y$ becomes non-commutative.
This is the interpretation taken by many authors who consider the non-commutative KP and mKP equations, which we adopt herein.
See Gilson and Nimmo~\cite{GN}, Gilson, Nimmo and Sooman~\cite{GNS}
Hamanaka~\cite{Hamanaka}, Hamanaka and Toda~\cite{HamanakaToda}, Koikawa~\cite{Koikawa}, Toda~\cite{Toda} and Wang and Wadati~\cite{WW}.
The latter `non-commutativity' is often referred to as the `nonabelian' context, see Nijhoff~\cite{Nijhoff-Lagrangian3form}.
Non-commutative integrable systems have received a lot of recent attention in terms of general applications.
For example, in non-commutative nonlinear optics, applications also include boomerons, trappons and simulton solutions; 
see Calogero and Degasperis~\cite{CD}, Degasperis and Lombardo~\cite{DL}. 
For the non-commutative KP equation, multi-soliton and $\tau$-function solutions can be characterised in terms of quasi-determinants,
see for example Etingof, Gelfand and Retakh~\cite{EGR97,EGR98}, Gilson and Nimmo~\cite{GN}, Gilson \textit{et al.} \cite{GNS}, Hamanaka~\cite{Hamanaka} and Sooman~\cite{So}.
We discuss this in more detail in Section~\ref{sec:discussion}, where we also introduce a quasi-trace solution formula for the non-commutative KP equation. 

To summarise, what is new in this paper, is that in the non-commutative case we:
\begin{enumerate}
\item[(1)] Present the pre-P\"oppe algebra associated with the KP and lifted mKP equations; 
\item[(2)] Prove the KP and lifted mKP equations are integrable via direct linearisation;
\item[(3)] Establish the Miura transformation linking the KP and lifted mKP equations;
\item[(4)] Establish a quasi-trace solution form for the KP equation;
\item[(5)] Generate approximate solutions to the KP equation by numerically solving the GLM equation.
\end{enumerate}

Our paper is structured as follows. 
In Section~\ref{sec:Poppealg} we develop the pre-P\"oppe algebra underlying both the non-commutative KP equation and lifted mKP equations.
We establish the direct linearisation of the non-commutative KP equation and lifted mKP equations respectively in Sections~\ref{sec:KP} and \ref{sec:mKP}. 
In Section~\ref{sec:discussion} we show that the solution to the non-commutative KP equation can be expressed as a quasi-trace,
and discuss extensions and generalisations of the results herein. We derive the soliton solution to the commutative KP equation in~\ref{sec:solitons}.
Then in \ref{sec:numerics} we show how, for the non-commutative KP equation, the GLM equation can be numerically solved to reveal the
time evolution of solutions, and compare these simulations against some pseudo-spectral schemes.
Finally, in \ref{sec:finishproof2ndThm}, we provide the final part to the proof of Theorem~\ref{thm:mKP}.

\section{Pre-P\"oppe algebra}\label{sec:Poppealg}
Herein we introduce the pre-P\"oppe algebra. We first introduce its simpler form, sufficient to encapsulate the KP equation, and then, its more general form that encapsulates the lifted mKP equations.
We begin by giving a more rigorous context to some of the material presented in the introduction.

Let $\mathbb V\coloneqq L^2((-\infty,0];\mathbb C^m)$ denote the space of square-integrable $\CC^m$-valued functions on $(-\infty,0]$,
and $\mathfrak J_2=\mathfrak J_2(\Vb)$ denote the space of Hilbert--Schmidt operators on $\mathbb V$.
Any Hilbert--Schmidt valued operator $G$ on $\Vb$ generates a unique kernel function $g\in L^2((-\infty,0]^{\times2};\CC^{m\times m})$
such that for all $\phi\in\Vb$ and $z\in(-\infty,0]$,
\begin{equation*}
\bigl(G\phi\bigr)(z)=\int_{-\infty}^0 g(z,\zeta)\,\phi(\zeta)\,\rd\zeta.
\end{equation*}
Conversely, any $g\in L^2((-\infty,0]^{\times2};\CC^{m\times m})$ generates a Hilbert--Schmidt valued operator $G$ on $\Vb$,
with $\|G\|_{\mathfrak J_2(\Vb)}=\|g\|_{L^2((-\infty,0]^{\times2};\CC^{m\times m})}$. See Simon~\cite{Simon}.
\begin{definition}[Left and right derivatives]
  Consider a Hilbert--Schmidt valued operator $G$ on $\Vb$ with a continuously differentiable kernel $g=g(z,\zeta)$.
  Assume that the partial derivatives $\pa_z g$ and $\pa_\zeta g$ of $g$ are square-integrable on $(-\infty,0]^2$.
  Then we define the left and right derivatives of $G$, respectively, $\pa_\ell G$ and $\pa_rG$, to be the Hilbert--Schmidt valued operators
  corresponding to the respective kernels $\pa_z g$ and $\pa_\zeta g$.
\end{definition}
\begin{remark}
Note for any pair of suitable Hilbert--Schmidt operators $G$, $\hat G$, we have,
\begin{equation*}
\pa_\ell(G\hat G)=(\pa_\ell G)\hat G\quad\mathrm{and}\quad \pa_r(G\hat G)=G(\pa_r\hat G).
\end{equation*}
\end{remark}
\begin{definition}[Semi-additive operator]\label{def:semiadd}
  A Hilbert--Schmidt operator $P$ on $\Vb$ with corresponding square-integrable kernel $p$ is \emph{semi-additive with parameters} $x,y\in\R$
  if its operation for any square-integrable function $\phi\in\Vb$, with $z\in(-\infty,0]$, is,
  \begin{equation*}
     \bigl(P\phi\bigr)(z;x,y)=\int_{-\infty}^0 p(z+x,\zeta+x;y)\,\phi(\zeta)\,\rd\zeta.
  \end{equation*}
\end{definition}
Note that for any semi-additive operator such as that above, we have $\pa_xP=(\pa_\ell+\pa_{\mathrm{r}})P$.  
As mentioned in the introduction, crucial to our results herein is the pre-P\"oppe product given in Lemma~\ref{lemma:Poppeprodsemiadd}.
\begin{remark}\label{rmk:multiplefactors}
Some linear combinations of monomials generate multi-factor products which are necessarily of the form $\lb\cdot\rb(z,0;x,y)\lb\cdot\rb(0,0;x,y)\cdots\lb\cdot\rb(0,0;x,y)\lb\cdot\rb(0,\zeta;x,y)$.
Hereafter, we represent such products by $\lb\cdot\rb\lb\cdot\rb\cdots\lb\cdot\rb\lb\cdot\rb$.
\end{remark}
\begin{remark}
We note that, in fact, Theorems~\ref{thm:KP} and \ref{thm:mKP} both hold for all $(z,\zeta)\in(-\infty,0]^{\times2}$,
with the proviso that kernel products are interpreted as outlined in Lemma~\ref{lemma:Poppeprodsemiadd}, and more generally, Remark~\ref{rmk:multiplefactors}.
\end{remark}
The simpler pre-P\"oppe algebra is generated by $\lb V\rb$ and its left and right partial derivatives, where $V\coloneqq (\id-P)^{-1}$.
Here, we assume $P$ is semi-additive with parameters $x,y\in\R$. We suppress its $y$-dependence herein.
The reason for this will become clear in Section~\ref{sec:KP}. By partial fractions for operators, we observe that,
\begin{equation}\label{eq:partialfractions}
V\equiv\id+PV\equiv\id+VP.
\end{equation}
Suppose $\mathbb N^\ast$ denotes the set of words $a_1a_2\cdots a_n$ we can construct from $a_1,a_2,\ldots,a_n\in\mathbb N$.
The set $\mathbb N^\ast$ is called the free monoid. Let $\mathcal C(n)$ denote the set of all compositions of $n$,
and $\mathcal C$ denote the set of all compositions.
The following map is a natural element in all of the P\"oppe algebras constructed thus far; see Malham~\cite{MKdV}. 
\begin{definition}[Signature character map]\label{def:signature}
  Suppose $a_1a_2\cdots a_n\in\mathbb N^\ast$. The \emph{signature character map} $\chi\colon\mathbb N^\ast\to\mathbb Q$
  is defined by ($C_\ell^n$ as the Leibniz coefficient $n$ choose $\ell$),
  \begin{equation*}
   \chi\colon a_1a_2\cdots a_n\mapsto \prod_{k=1}^nC_{a_k}^{a_k+\cdots+a_n}.
  \end{equation*}
\end{definition}
\begin{remark}
  It is natural to extend $\chi$ to act homomorphically on any tensor product of compositions so that we have,
  $\chi(w_1\otimes w_2\otimes\cdots\otimes w_\ell)=\chi(w_1)\chi(w_2)\cdots\chi(w_\ell)$ for any words $w_1,w_2,\ldots,w_\ell\in\mathbb N^\ast$.
  See Malham~\cite{MKdV} and Blower and Malham~\cite{BM}. We do not need to utilise this herein though.
\end{remark}

Hereafter we use the notation, $P_{\hat a}\coloneqq\pa_x^aP$.
We assume $P$ is continuously differentiable to any order we require, and its partial derivatives at that order are Hilbert--Schmidt valued.
By successively differentiating $V$ with respect to $x$, and using the partial fraction formulae~\eqref{eq:partialfractions}, and then applying the
kernel bracket operator `$\lb\,\cdot\,\rb$', we observe,
\begin{equation}\label{eq:expansion}
\pa_x^n\lb V\rb=\sum \chi(a_1\cdots a_k)\cdot\lb VP_{\hat{a}_1}V\cdots VP_{\hat{a}_k}V\rb,
\end{equation}
where the sum is over all compositions $a_1a_2\cdots a_k\in\mathcal C(n)$.
We now lean on the fact that $P$ is semi-additive. We observe that $P_{\hat a}=(\pa_\ell+\pa_{\mathrm{r}})^aP$ for any $a\in\mathbb N$.
This is just the chain rule. Thus we observe that $P_{\hat a}$ is represented by, 
\begin{equation}\label{eq:Leibnizexpansionorig}
P_{\hat a }=\sum_{k=0}^a C_{k}^a\,\pa_\ell^k\pa_{\mathrm{r}}^{a-k}P.
\end{equation}
Note, the terms in the sum are indexed by pairs $(k,a-k)$ of left and right derivatives, respectively.
Recall our notation \eqref{eq:derivnotation} in the introduction.
In the monomials on the right in~\eqref{eq:expansion}, the individual `$P_{\hat a}$' factors are Leibniz sums of pairs of left and right derivatives of the form $P_{a,b}$.
We can expand each of $P_{\hat{a}_1}$ to $P_{\hat{a}_k}$ in linear combinations of expressions of the form
$VP_{\hat{a}_1}VP_{\hat{a}_2}V\cdots VP_{\hat{a}_k}V$. Hence the base monomials in~\eqref{eq:expansion} are of the form,
\begin{equation}\label{eq:opmonomials}
\lb VP_{a_1,b_1}VP_{a_2,b_2}V\cdots VP_{a_k,b_k}V\rb.
\end{equation}
With this in hand, the following formulation of the pre-P\"oppe product formula in Lemma~\ref{lemma:Poppeprodsemiadd} is crucial to our algebra construction. 
\begin{lemma}\label{lemma:Poppeprodmonomials}
  For arbitrary Hilbert--Schmidt operators $F$ and $\hat{F}$ and a semi-additive Hilbert--Schmidt operator $P$
  with parameters $x,y$ and a smooth kernel, we have, for any $a,b,c,d\in\mathbb N$,
  \begin{align*}
  \lb FP_{a,b} V\rb\,\lb VP_{c,d}\hat{F}\rb=&\;\lb FP_{a,b+1} VP_{c,d}\hat{F}\rb\\
                                              &\;+\lb FP_{a,b} VP_{c+1,d}\hat{F}\rb\\
                                              &\;+\lb FP_{a,b} VP_{\hat{1}} VP_{c,d}\hat{F}\rb.
  \end{align*}
\end{lemma}
This result is straightforwardly established by systematically combining the partial fraction formulae~\eqref{eq:partialfractions} with P\"oppe's product rule
in Lemma~\ref{lemma:Poppeprodsemiadd}. 
We can thus construct the real matrix algebra of monomials of the form~\eqref{eq:opmonomials}, with the product given by the pre-P\"oppe product in Lemma~\ref{lemma:Poppeprodmonomials}.
This algebra is sufficient for our purpose of establishing the direct linearisation of the KP equation~\eqref{eq:KPsingform}.

To establish direction linearisation for the lifted mKP equations~\eqref{eq:triplemKPh}-\eqref{eq:triplemKPf}, we consider the algebra
generated by $\lb V\rb$ and $\lb V^\dag\rb$, where $V^\dag\coloneqq (\id+P)^{-1}$, and their left and right derivatives.
Actually as we shall see, it is slightly more general than this.
As above, we assume $P$ is semi-additive with parameters $x,y\in\R$, and we suppress its $y$-dependence.
By partial fractions, in addition to \eqref{eq:partialfractions}, we also have,
\begin{equation}\label{eq:partialfractionsdag}
V^\dag\equiv\id-PV^\dag\equiv\id-V^\dag P.
\end{equation}
All of the properties derived above naturally apply here, except that now we must also incorporate $\lb V^\dag\rb$.
For example, there is an equivalent expansion to \eqref{eq:expansion} for $\pa_x^n\lb V^\dag\rb$ which is naturally sign-sensitive.
We generalise the class of base polynomials of the form \eqref{eq:opmonomials} to those of the form \eqref{eq:mKPopmonomials}
including monomials where one or more of the $V$'s shown are replaced by $V^\dag$.
With these new monomial forms in hand, the form of the pre-P\"oppe product in Lemma~\ref{lemma:Poppeprodmonomials} can be extended as follows.
Recall our notation for $P_{[a,b]}$ in \eqref{eq:Pbracketnotation} in the introduction.
\begin{lemma}\label{lemma:Poppeprodmod}
  For arbitrary Hilbert--Schmidt operators $F$ and $\hat{F}$ and a semi-additive Hilbert--Schmidt operator $P$
  with parameters $x,y$ and a smooth kernel, for any $a,b,c,d\in\mathbb N$ we have, 
  \begin{align*}
  \lb FP_{a,b} V\rb\,\lb V P_{c,d}\hat{F}\rb=&\;\lb FP_{a,b+1} V P_{c,d}\hat{F}\rb\\
                                               &\;+\lb FP_{a,b} V P_{c+1,d}\hat{F}\rb\\
                                               &\;+\lb FP_{a,b} V P_{\hat{1}} V P_{c,d}\hat{F}\rb,\\    
  \lb FP_{a,b} V\rb\,\lb V^\dag P_{c,d}\hat{F}\rb=&\;\lb FP_{a,b+1} V^\dag P_{c,d}\hat{F}\rb\\
                                                   &\;+\lb FP_{a,b} V P_{c+1,d}\hat{F}\rb\\
                                                   &\;-\lb FP_{a,b} VP_{[1,0]}V^\dag P_{c,d}\hat{F}\rb,\\
  \lb FP_{a,b} V^\dag\rb\,\lb V P_{c,d}\hat{F}\rb=&\;\lb FP_{a,b+1} V P_{c,d}\hat{F}\rb\\
                                                   &\;+\lb FP_{a,b} V^\dag P_{c+1,d}\hat{F}\rb\\
                                                   &\;+\lb FP_{a,b} V^\dag P_{[1,0]}V P_{c,d}\hat{F}\rb,\\
  \lb FP_{a,b} V^\dag\rb\,\lb V^\dag P_{c,d}\hat{F}\rb=&\;\lb FP_{a,b+1} V^\dag P_{c,d}\hat{F}\rb\\
                                                        &\;+\lb FP_{a,b} V^\dag P_{c+1,d}\hat{F}\rb\\
                                                        &\;-\lb FP_{a,b} V^\dag P_{\hat{1}} V^\dag P_{c,d}\hat{F}\rb.
  \end{align*}
\end{lemma}
The first result in Lemma~\ref{lemma:Poppeprodmod} is just a re-written version of 
Lemma~\ref{lemma:Poppeprodmonomials}. The other three statements rely on slight modifications
of the proof of Lemma~\ref{lemma:Poppeprodmonomials}, utilising \eqref{eq:partialfractionsdag}.

Recall from the introduction, that for any monomial form $\tilde VP_{a_1,b_1}\tilde V\cdots\tilde VP_{a_k,b_k}\tilde V$, where the $\tilde V$'s
could be either $V$ or $V^\dag$, the monomial form $(\tilde VP_{a_1,b_1}\tilde V\cdots\tilde VP_{a_k,b_k}\tilde V)^\dag$, 
denotes the monomial form $(-1)^k\tilde V^\dag P_{a_1,b_1}\tilde V^\dag \cdots\tilde V^\dag P_{a_k,b_k}\tilde V^\dag$,
where if $\tilde V=V$ then $\tilde V^\dag=V^\dag$ and if $\tilde V=V^\dag$ then $\tilde V^\dag=V$.
Recall from the introduction, for any Hilbert--Schmidt operator $W$ of the monomial form $\tilde VP_{a_1,b_1}\tilde V\cdots\tilde VP_{a_k,b_k}\tilde V$ just outlined,
our notation~\eqref{eq:skew-symmdef} for the skew and symmetric forms $[W]$ and $\{W\}$.
The P\"oppe products in Lemma~\ref{lemma:Poppeprodmod} can also be represented as follows.
\begin{lemma}\label{lemma:Poppeprodforms}
  For arbitrary Hilbert--Schmidt operators $F$ and $\hat{F}$ and a semi-additive Hilbert--Schmidt operator $P$
  with parameters $x,y$ and a smooth kernel, for any $a,b,c,d\in\mathbb N$ we have, 
  \begin{align*}
  [FP_{a,b}V]\,[V P_{c,d}\hat{F}]=&\;\bigl\{FP_{a,b+1}[V P_{c,d}\hat{F}]\bigr\}\\
     &\;+\bigl\{FP_{a,b} V[P_{c+1,d}\hat{F}]\bigr\}\\
     &\;+\{FP_{a,b} V P_{\hat{1}} V P_{c,d}\hat{F}\}\\
     &\;-\{FP_{a,b} V P_{[1,0]} V^\dag P_{c,d}\hat{F}^\dag\},\\
  [FP_{a,b}V]\,\{V P_{c,d}\hat{F}\}=&\;\bigl[FP_{a,b+1}\{V P_{c,d}\hat{F}\}\bigr]\\
    &\;+\bigl[FP_{a,b} V\{P_{c+1,d}\hat{F}\}\bigr]\\
    &\;+[FP_{a,b} V P_{\hat{1}} V P_{c,d}\hat{F}]\\
    &\;+[FP_{a,b} V P_{[1,0]} V^\dag P_{c,d}\hat{F}^\dag],\\
  \{FP_{a,b}V\}\,[V P_{c,d}\hat{F}]=&\;\bigl[FP_{a,b+1}[V P_{c,d}\hat{F}]\bigr]\\
    &\;+\bigl[FP_{a,b} V[P_{c+1,d}\hat{F}]\bigr]\\
    &\;+[FP_{a,b} V P_{\hat{1}} V P_{c,d}\hat{F}]\\
    &\;-[FP_{a,b} V P_{[1,0]} V^\dag P_{c,d}\hat{F}^\dag],\\
  \{FP_{a,b}V\}\,\{V P_{c,d}\hat{F}\}=&\;\bigl\{FP_{a,b+1}\{V P_{c,d}\hat{F}\}\bigr\}\\
    &\;+\bigl\{FP_{a,b} V\{P_{c+1,d}\hat{F}\}\bigr\}\\
    &\;+\{FP_{a,b} V P_{\hat{1}} V P_{c,d}\hat{F}\}\\
    &\;+\{FP_{a,b} V P_{[1,0]} V^\dag P_{c,d}\hat{F}^\dag\}.
  \end{align*}
\end{lemma}
The results stated in Lemma~\ref{lemma:Poppeprodforms} follow directly from Lemma~\ref{lemma:Poppeprodmod}.
\begin{remark}\label{rmk:skewsymmproducts}
We remark that the product of two skew forms or two symmetric
forms generates a symmetric form, while the product of a skew and a symmetric form, or vice-versa, generates a skew form.
We observe this just above.
\end{remark}
\begin{remark}\label{rmk:Vcorrection}
  In the first and third results in Lemma~\ref{lemma:Poppeprodforms}, we can replace the second factor $[V P_{(c,d)}\hat{F}]$ by $[V]$
  with the corresponding natural consequences for the terms on the right, except that the final term in both cases changes sign.
  The change of sign occurs because the final term already accounts for the fact that the `$\dag$' operation has already been applied to $P_{c,d}$.
\end{remark}

Thus, the pre-P\"oppe algebra we consider herein is that generated by monomials of the form $\lb W\rb$ where
$W=\tilde VP_{a_1,b_1}\tilde V\cdots\tilde VP_{a_k,b_k}\tilde V$, and the $\tilde V$'s could be either $V$ or $V^\dag$.
The pre-P\"oppe product is that in Lemma~\ref{lemma:Poppeprodmod}. We do not require this algebra to be unital.
There is a subalgebra spanned by monomials of the form $\lb W\rb$ where $W=VP_{a_1,b_1}V\cdots VP_{a_k,b_k} V$ and the
pre-P\"oppe product is the first case in Lemma~\ref{lemma:Poppeprodmod}, or equivalently, that in Lemma~\ref{lemma:Poppeprodmonomials}.
The subalgebra is sufficient for us to establish our first main result Theorem~\ref{thm:KP} for the non-commutative KP equation, which we consider next.

\begin{remark}\label{rmk:algext}
  The rationale for nominating this algebra the `pre-P\"oppe algebra' is as follows.
  The P\"oppe algebra was developed in Malham~\cite{MKdV} and Blower and Malham~\cite{BM} respectively for the non-commutative KdV and NLS/mKdV hierarchies.
  It was based on the P\"oppe product devised by P\"oppe~\cite{PSG,PKdV} for the sine-Gordon and KdV equations. 
  The underlying product therein utilises that the scattering data generates a Hankel operator.  
  For the KP equation, P\"oppe~\cite{PKP} generalised this product to the case when the scattering data is of semi-additive form.
  The Hankel operator case can be considered as a subalgebra of the semi-additive case, indeed, see Remark~\ref{rmk:KPreductions} at the end of Section~\ref{sec:KP}.
  Since the Hankel case has already been nominated the P\"oppe algebra, we have nominated the semi-additive case as the pre-P\"oppe algebra.  
  Additionally, there are extensions and further aspects of the algebra that bolster the `pre-P\"oppe' perspective we have taken.
  We do not require them herein, but they may be useful for further investigations. See item (5) in Section~\ref{sec:discussion}.
\end{remark}

\section{Non-commutative KP equation}\label{sec:KP}
We reconstruct the solution $g=g(x,y;t)$ to the non-commutative KP equation~\eqref{eq:KPsingform} from the kernel $\lb G\rb$ of solution $G$ to
the pair of linear equations~\eqref{eq:linearform} and \eqref{eq:GLM}, as outlined in the introduction. We also require that
the solution $p=p(t)$ to the linear equation~\eqref{eq:linearform} has the semi-additive form~\eqref{eq:semiadditiveform}.
Thus the first constraint in~\eqref{eq:constraints} is naturally satisfied. We additionally impose the second constraint in~\eqref{eq:constraints}
on $p=p(t)$, namely $p_y=p_{zz}-p_{\zeta\zeta}$. Further recall from the introduction that our semi-additive form assumption for $p$
and this latter constraint on $p_y$ imply that $p$ satisfies~\eqref{eq:secondlinearform}. Hence to summarise, we assume
$P$ is a semi-additive operator with kernel $p=p(z+x,\zeta+x;y,t)$ which satisfies,
\begin{equation}\label{eq:twolinearforms}
   P_t=4\,\bigl(P_{zzz}+P_{\zeta\zeta\zeta}\bigr)\quad\text{and}\quad P_y=P_{zz}-P_{\zeta\zeta}.
\end{equation}
We assume that such a Hilbert-Schmidt valued operator exists for a short time at least, i.e.\/ for some $T>0$ and suitable initial data,
there exists a smooth square integrable matrix-valued function $p=p(z+x,\zeta+x;y,t)$ satisfying~\eqref{eq:twolinearforms} for $t\in[0,T]$.
The second ingredient in our prescription is to find a Hilbert-Schmidt valued operator solution $G$ to the linear Fredholm equation~\eqref{eq:GLM}.
The following result is proved in Doikou \textit{et al.} \cite[Lemma~1]{DMS}. 
\begin{lemma}[Existence and uniqueness]\label{lemma:GLMEU}
  Assume that $Q_0\in\mathfrak J_2$ and for some $T>0$ we know, $Q\in C^\infty\bigl([0,T];\Jf_2\bigr)$ with $Q(0)=Q_0$ and
  $P\in C^\infty\bigl([0,T];\Jf_N\bigr)$, where $N$ is $1$ or $2$. Further assume, $\mathrm{det}_2(\id-Q_0)\neq0$.
  Then there exists a $T^\prime>0$ with $T^\prime\leqslant T$ such that for $t\in[0,T^\prime]$ we have $\mathrm{det}_2(\id-Q(t))\neq0$ and 
  there exists a unique solution $G\in C^\infty\bigl([0,T^\prime];\Jf_N\bigr)$ to the linear equation $P=G(\id-Q)$.
\end{lemma}

We now establish our first main result, Theorem~~\ref{thm:KP}, for the non-commutative KP equation~\eqref{eq:KPsingform}.
We express the linear Fredholm equation~\eqref{eq:GLM} with $Q\coloneqq P$ in the form,
\begin{equation*}
    G=PV,
\end{equation*}
where $V\coloneqq(\id-P)^{-1}$. We know that such a solution operator $G$ exists for a short time at least, from Lemma~\ref{lemma:GLMEU}.
Using the partial fractions formulae~\eqref{eq:partialfractions}, we observe $\pa_t G=\pa_t V=V(\pa_t P)V$ and $\pa_x G=\pa_x V=V(\pa_x P)V$,
and generally for $n\geqslant1$, $\pa_x^n G=\pa_x^n V$ which, after applying the kernel bracket operator, can be expanded in the form~\eqref{eq:expansion}.
Our goal is to show that $g=\lb G\rb$ satisfies the non-commutative KP equation~\eqref{eq:KPsingform}. We observe that in \eqref{eq:KPsingform}, $g$ only
appears in derivative form, and thus it sufficies, using the partial fractions formulae~\eqref{eq:partialfractions}, to show that $g=\lb V\rb$ satisfies \eqref{eq:KPsingform}.
Then, using~\eqref{eq:twolinearforms}, we observe that, 
\begin{equation*}
\pa_t\lb V\rb=4\,\bigl(\lb V(P_{3,0}+P_{0,3})V\rb.
\end{equation*}
From expansion~\eqref{eq:expansion} with $n=3$, we observe that,
\begin{align*}
  \pa_x^3\lb V\rb=&\;\lb VP_{\hat 3}V\rb+3\,\bigl(\lb VP_{\hat 2}VP_{\hat 1}V\rb+\lb VP_{\hat 1}VP_{\hat 2}V\rb\bigr)\\
  &\;+6\,\lb VP_{\hat 1}VP_{\hat 1}VP_{\hat 1}V\rb\\
  =&\;\lb V(P_{3,0}+3P_{2,1}+3P_{1,2}+P_{0,3})V\rb\\
  &\;+3\,\lb V(P_{2,0}+2P_{1,1}+P_{0,2})VP_{\hat 1}V\rb\\
  &\;+3\,\lb VP_{\hat 1}V(P_{2,0}+2P_{1,1}+P_{0,2})V\rb\\
  &\;+6\,\lb VP_{\hat 1}VP_{\hat 1}VP_{\hat 1}V\rb,
\end{align*}
and so,
\begin{align*}
  \tfrac13\,\bigl(\pa_t\lb V\rb-\pa_x^3\lb V\rb\bigr)  
  =&\;\lb V(P_{3,0}-P_{2,1}-P_{1,2}+P_{0,3})V\rb\\
  &\;-\,\lb V(P_{2,0}+2P_{1,1}+P_{0,2})VP_{\hat 1}V\rb\\
  &\;-\,\lb VP_{\hat 1}V(P_{2,0}+2P_{1,1}+P_{0,2})V\rb\\
  &\;-2\,\lb VP_{\hat 1}VP_{\hat 1}VP_{\hat 1}V\rb.
\end{align*}

Let us now turn our attention to the second nonlinear term in~\eqref{eq:KPsingform}.
Using the pre-P\"oppe product in Lemma~\ref{lemma:Poppeprodmonomials} we observe, since $\pa_y P=P_{[2,0]}$ from \eqref{eq:twolinearforms},
that the products $g_xg_y$ and $g_yg_x$ respectively have the form,
\begin{align*}
  \lb VP_{\hat 1}V\rb\,\lb VP_{[2,0]}V\rb=&\;\lb V(P_{1,1}+P_{0,2})VP_{[2,0]}V\rb\\ 
                                        &\;+\lb VP_{\hat 1}V(P_{3,0}-P_{1,2})V\rb\\ 
                                        &\;+\lb VP_{\hat 1}VP_{\hat 1}VP_{[2,0]}V\rb\\ 
  \lb VP_{[2,0]}V\rb\,\lb VP_{\hat 1}V\rb=&\;\lb V(P_{2,1}-P_{0,3})VP_{\hat 1}V\rb\\ 
                                        &\;+\lb VP_{[2,0]}V(P_{2,0}+P_{1,1})V\rb\\ 
                                        &\;+\lb VP_{[2,0]}VP_{\hat 1}VP_{\hat 1}V\rb. 
\end{align*}
We can now establish the following identity. 
\begin{lemma}\label{lemma:keystep}
  We have the following identity for $[g_x,g_y]$:
  \begin{align*}
    2\,\Bigl[\lb VP_{\hat 1}V&\rb,\lb VP_{[2,0]}V\rb\Bigr]+\lb V(P_{4,0}-2P_{2,2}+P_{0,4})V\rb\\
           =&\;\pa_x\Bigl(\lb V(P_{3,0}-P_{2,1}-P_{1,2}+P_{0,3})V\rb\\
            &\;+\lb VP_{\hat 1}VP_{[2,0]}V\rb-\lb VP_{[2,0]}VP_{\hat 1}V\rb\Bigr).
   \end{align*}
\end{lemma}
\begin{proof}
  Directly computing the `$\pa_x$' derivative of the expression on the righthand side, we observe,
  \begin{align*}
    \pa_x\bigl(\lb V(P_{3,0}-&P_{2,1}-P_{1,2}+P_{0,3})V\rb\\
     +\lb VP_{\hat 1}&VP_{[2,0]}V\rb-\lb VP_{[2,0]}VP_{\hat 1}V\rb\bigr)\\
    =&\;\lb V(P_{4,0}-2P_{2,2}+P_{0,4})V\rb\\
    &\;+\lb VP_{\hat 1}V(P_{3,0}-P_{2,1}-P_{1,2}+P_{0,3})V\rb\\
    &\;+\lb V(P_{3,0}-P_{2,1}-P_{1,2}+P_{0,3})VP_{\hat 1}V\rb\\
    &\;+2\,\lb VP_{\hat 1}VP_{\hat 1}VP_{[2,0]}V\rb\\ 
    &\;+\lb V(P_{2,0}+2P_{1,1}+P_{0,2})VP_{[2,0]}V\rb\\ 
    &\;+\lb VP_{\hat 1}V(P_{3,0}+P_{2,1}-P_{1,2}-P_{0,3})V\rb\\
    &\;-\lb V(P_{3,0}+P_{2,1}-P_{1,2}-P_{0,3})VP_{\hat 1}V\rb\\
    &\;-2\,\lb VP_{[2,0]}VP_{\hat 1}VP_{\hat 1}V\rb\\
    &\;-\lb VP_{[2,0]}V(P_{2,0}+2P_{1,1}+P_{0,2})V\rb\\
   =&\;\lb V(P_{4,0}-2P_{2,2}+P_{0,4})V\rb\\
    &\;+2\,\lb VP_{\hat 1}V(P_{3,0}-P_{1,2})V\rb\\
    &\;+2\,\lb VP_{\hat 1}VP_{\hat 1}VP_{[2,0]}V\rb\\ 
    &\;+2\,\lb V(P_{1,1}+P_{0,2})VP_{[2,0]}V\rb\\   
    &\;-2\,\lb V(P_{2,1}-P_{0,3})VP_{\hat 1}V\rb\\
    &\;-2\,\lb VP_{[2,0]}VP_{\hat 1}VP_{\hat 1}V\rb\\
    &\;-2\,\lb VP_{[2,0]}V(P_{2,0}+P_{1,1})V\rb.
  \end{align*}
  Note that in the last step, we added and subtracted the term $\lb VP_{[2,0]}VP_{[2,0]}V\rb$. 
  Comparing this with our expressions for $g_xg_y$ and $g_yg_x$ just preceding the lemma, gives the required result. \qed
\end{proof}
\begin{proof}[of Theorem~~\ref{thm:KP}] This follows directly from Lemma~\ref{lemma:keystep}.
Using the pre-P\"oppe product from Lemma~\ref{lemma:Poppeprodmonomials}, we observe that $g_x^2$ is given by,
\begin{align*}
  \lb VP_{\hat 1}V\rb\,\lb VP_{\hat 1}V\rb=&\;\lb V(P_{1,1}+P_{0,2})VP_{\hat 1}V\rb\\
                                           &\;+\lb VP_{\hat 1}V(P_{2,0}+P_{1,1})V\rb\\
                                           &\;+\lb VP_{\hat 1}VP_{\hat 1}VP_{\hat 1}V\rb.
\end{align*}
Gathering these results, we observe, $\frac13(g_t-g_{xxx})+2\,g_x^2$, equals,
\begin{align*}
  \tfrac13\,\bigl(\pa_t\lb V\rb-\pa_x^3\lb V&\rb\bigr)+2\,\lb VP_{\hat 1}V\rb\,\lb VP_{\hat 1}V\rb\\
  =&\;\lb V(P_{3,0}-P_{2,1}-P_{1,2}+P_{0,3})V\rb\\
  &\;-\,\lb VP_{[2,0]}VP_{\hat 1}V\rb+\,\lb VP_{\hat 1}VP_{[2,0]}V\rb,
\end{align*}
which equals the expression on the righthand side of Lemma~\ref{lemma:keystep},
once we have applied $\pa_x^{-1}$ to both sides (this requires the kernels herein vanish in the far field).
From the result of the lemma, that in turn equals the expression, with $\pa_x^{-1}$ applied,
on the left, which equals $2\,[g_x,g_y]+\mathrm{D}g_y$.
We have thus established the result of Theorem~\ref{thm:KP}.\qed
\end{proof}

\begin{remark}\label{rmk:KPreductions}
  We make the following observations.
  Consider the Hilbert--Schmidt-valued scattering operator $P$, with semi-additive form $p=p(z+x,\zeta+x;y,t)$.
  Here $z$ and $\zeta$ are the primary kernel variables, and $x,y\in\mathbb R$ are considered to be additional parameters.
  Time $t\in\mathbb R$ is also an additional parameter dependence, though herein we typically restrict it to $t\geqslant0$.
  We can canonically generate any semi-additive kernel, $p=p(z+x,\zeta+x;y,t)$, from any kernel $p=p(z,\zeta;\cdot,\cdot)$.
  Such semi-additive kernels thus represent a family of operators.
  If we assume such a semi-additive kernel has no $y$-dependence, then the constraints~\eqref{eq:constraints} collapse to
  $p_x=p_z+p_\zeta$ and $p_{zz}=p_{\zeta\zeta}$, so necessarily and sufficiently, $p$ has the form $p=p(z+\zeta+2x;t)$.
  Modulo the scalar factor `$2$', this is the Hankel form for the scattering operator we assumed for the cases
  of the non-commutative KdV hierarchy in Malham~\cite{MKdV} and for the non-commutative NLS/mKdV hierarchy in Blower and Malham \cite{BM}.
  Another characterisation here is to impose the symmetry constraint, $p(x+z,\zeta;\cdot,\cdot)=p(z,\zeta+x;\cdot,\cdot)$.
  This is equivalent to specialising the kernel $p$ to be the kernel of a Hankel operator of the form $p=p(z+\zeta+x;\cdot,\cdot)$. 
  Lastly, we also remark that if we assume no $t$-dependence in $p$, then the non-commutative KP equation~\eqref{eq:KPsingform} reduces to a non-commutative version of the Boussinesq equation.
\end{remark}

\section{Non-commutative lifted mKP equation}\label{sec:mKP}
We now reconstruct the solution $(h,g,f)$, with $h=h(x,y;t)$, $g=g(x,y;t)$ and $f=f(x,y;t)$, to the lifted mKP equations~\eqref{eq:triplemKPh}-\eqref{eq:triplemKPf},
from the kernel $\lb G\rb$ solving the pair of linear equations~\eqref{eq:linearform} and \eqref{eq:GLM}, given in the introduction.
We make exactly the same assumptions on the operator $P$ being semi-additive of the form shown and satisfying~\eqref{eq:twolinearforms},
as we did for the case of the non-commutative KP equation at the beginning of Section~\ref{sec:KP}.
The result of Lemma~\ref{lemma:GLMEU} applies here as well.
The main difference here is that we now set $Q\coloneqq P^2$. The GLM equation~\eqref{eq:GLM} now takes the form $P=G(\id-P^2)$,
which we can rewrite in the form~\eqref{eq:mKPsolnform}, i.e.\/ $G=VPV^\dag$, where $V\coloneqq(\id-P)^{-1}$ and $V^\dag\coloneqq(\id+P)^{-1}$.
Using the partial fractions formulae~\eqref{eq:partialfractions} and \eqref{eq:partialfractionsdag}, it is straightforward to show that,
\begin{equation}\label{eq:mKPskewform}
[V]\equiv[G]\equiv2\,\lb G\rb.
\end{equation}
This solution form for the lifted mKP equations is consistent with that for the non-commutative KdV equation,
to which the lifted mKP equations collapse when $\pa_y=0$ and $f=0$, considered by Blower and Malham~\cite{BM}.
We thus assume the solution form for $g=g(x,y;t)$ is $g=[V]$.
The solution form for $f=f(x,y;t)$ is motivated by the Miura transformation linking, the solution $g^{\mathrm{KP}}=g^{\mathrm{KP}}(x,y;t)$
to the non-commutative KP equation~\eqref{eq:KPsingform}, to the solution components $g$ and $f$ in the lifted mKP equations~\eqref{eq:triplemKPh}-\eqref{eq:triplemKPf}.
This Miura transformation has the form,
\begin{equation}\label{eq:Miura}
g_x+g^2+f=2g_x^{\mathrm{KP}}.
\end{equation}
\begin{remark}\label{rmk:Miuraproof}
A similar Miura transformation linking, the non-commutative KP equation~\eqref{eq:KPsingform}, to the non-commutative mKP equation form \eqref{eq:eqamKPassign}--\eqref{eq:amKP},
can be found in Gilson \textit{et al.}\/ \cite[eq.~6.11]{GNS}.
The Miura relation~\eqref{eq:Miura} fixes the third field from any two of the three fields shown therein.
For example, if $g$ and $f$ are the solution components satisfying the lifted mKP equations \eqref{eq:triplemKPg} and \eqref{eq:triplemKPf},
then evolution equation for $g^{\mathrm{KP}}$ is determined.
The Miura relation~\eqref{eq:Miura} is straightforwardly established as follows.
First, we note that $2g^{\mathrm{KP}}$ satisfies the same equation as the KP equation~\eqref{eq:KPsingform}, except with the factor `$2$' missing from the two nonlinear terms.
Second we compute the time derivative of~\eqref{eq:Miura}, giving $g_{xt}+\{g,g_t\}+f_t=2g_{xt}^{\mathrm{KP}}$.
Third, we substitute for $2g_{xt}^{\mathrm{KP}}$ from the rescaled version of the KP equation~\eqref{eq:KPsingform} for $2g^{\mathrm{KP}}$ just mentioned,
and replace all instances of $2g^{\mathrm{KP}}$ and its derivatives by $g_x+g^2+f$.
Fourth, we replace $g_t$ and $f_t$ in $g_{xt}+\{g,g_t\}+f_t$ with their forms in the lifted mKP equations \eqref{eq:triplemKPg} and \eqref{eq:triplemKPf}.
It is straightforward to check that the outcomes of the third and fourth steps just mentioned, match.
\end{remark}
Using the results of Lemma~\ref{lemma:Poppeprodmod}, or more directly via Lemma~\ref{lemma:Poppeprodforms} and Remark~\ref{rmk:Vcorrection}, it is straightforward to show that,
\begin{equation}\label{eq:gsquared}
[V]^2=\{VP_{\hat 1}V\}+\{VP_{[1,0]}V^\dag\}.
\end{equation}
Since $g_x=[VP_{\hat 1}V]$ and $g_x^{\mathrm{KP}}=\lb VP_{\hat 1}V\rb$, then~\eqref{eq:Miura} implies,
\begin{equation*}
f=2\lb VP_{\hat 1}V\rb-[VP_{\hat 1}V]-\{VP_{\hat 1}V\}-\{VP_{[1,0]}V^\dag\},
\end{equation*}
from which we deduce $f=-\{VP_{[1,0]}V^\dag\}$. We deduce that correspondingly, $h=[VP_{[1,0]}V^\dag]$, as follows.
If we set $P_{a,b}\coloneqq P_{[1,0]}$ in identity~\eqref{eq:mainanti1} just below, we observe,
\begin{equation}\label{eq:crucial}
 \mathrm{ad}_{[V]}\{VP_{[1,0]}V^\dag\}=\pa_x[VP_{[1,0]}V^\dag]-[V P_{[2,0]}V].
\end{equation}
Since $g_y=[V P_{[2,0]}V]$ and $f=-\{VP_{[1,0]}V^\dag\}$, we deduce that this statement is equivalent to
the first lifted mKP equation \eqref{eq:triplemKPh} if $h=[VP_{[1,0]}V^\dag]$.
This explains the solution forms chosen in Theorem~\ref{thm:mKP}.
\begin{remark}\label{rmk:explainform}
  This solution ansatz for $(h,g,f)$ explains the form of the lifted mKP equations~\eqref{eq:triplemKPh}-\eqref{eq:triplemKPf}.
  Recall from Remark~\ref{rmk:skewsymmproducts}, the product of two skew forms generates a symmetric form, and the product of a skew and symmetric form, or vice-versa, is a skew form.
  Equation~\eqref{eq:triplemKPg} is linear in $f$ and linear or cubic in $g$. Given our ansatz above, all the terms therein are thus skew forms.
  Equation~\eqref{eq:triplemKPf} is linear in $f$ and quadratic in $g$ or just quadratic in $f$. Hence all the terms therein are symmetric forms.
  The equivalent rationale for equation~\eqref{eq:triplemKPh} is given just above.
\end{remark}
Before proceeding to prove Theorem~\ref{thm:mKP}, we establish
the following lemma which lists a set of useful commutator identities.
We use the standard notation $\mathrm{ad}_ab\coloneqq ab-ba$, and the
notation $\pa^\perp P_{a,b}\coloneqq P_{a+1,b}-P_{a,b+1}$.
\begin{lemma}\label{lemma:commutatoridentities}
  For a semi-additive Hilbert--Schmidt operator $P$ on $\Vb$ with parameters $x,y\in\mathbb R$
  and a smooth kernel, for any $a,b\in\mathbb N\cup\{0\}$ we have, 
  \begin{align}
    \mathrm{ad}_{[V]}\{VP_{a,b}V^\dag\}=&\;\pa_x[VP_{a,b}V^\dag]-[V(\pa_x P_{a,b})V]\nonumber\\
                                      &\;+[VP_{[1,0]}V^\dag P_{a,b}V]\nonumber\\
                                      &\;-[VP_{a,b}V^\dag P_{[1,0]}V]\label{eq:mainanti1}\\
    \mathrm{ad}_{[V]}\{VP_{a,b}V\}=&\;[V(\pa^\perp P_{a,b})V]-[V(\pa^\perp P_{a,b})V^\dag]\nonumber\\
                                  &\;+[VP_{\hat 1}V P_{a,b}V]-[VP_{a,b}V P_{\hat 1}V]\nonumber\\
                                  &\;+[VP_{[1,0]}V^\dag P_{a,b}V^\dag]\nonumber\\
                                  &\;-[VP_{a,b}V P_{[1,0]}V^\dag]\label{eq:mainanti2}\\               
    \mathrm{ad}_{[V]}[VP_{a,b}V^\dag]=&\;\pa_x\{VP_{a,b}V^\dag\}+\{V(\pa^\perp P_{a,b})V\}\nonumber\\
                                    &\;-\{VP_{[1,0]}V^\dag P_{a,b}V\}\nonumber\\
                                   &\;-\{VP_{a,b}V^\dag P_{[1,0]}V\}\label{eq:mainanti3}\\
    \mathrm{ad}_{[V]}[VP_{a,b}V]=&\;\{V(\pa_xP_{a,b})V^\dag\}+\{V(\pa^\perp P_{a,b})V\}\nonumber\\
                               &\;+\{VP_{\hat 1}V P_{a,b}V\}-\{VP_{a,b}V P_{\hat 1}V\}\nonumber\\
                               &\;-\{VP_{[1,0]}V^\dag P_{a,b}V^\dag\}\nonumber\\
                               &\;-\{VP_{a,b}V P_{[1,0]}V^\dag\}. \label{eq:mainanti4}  
  \end{align}
\end{lemma}
The identities in Lemma~\ref{lemma:commutatoridentities} are straightforwardly established using the results of Lemma~\ref{lemma:Poppeprodforms}.

\begin{proof}[of Theorem~\ref{thm:mKP}]
We already know from~\eqref{eq:crucial} that $h\coloneqq[VP_{[1,0]}V^\dag]$, $g\coloneqq[V]$ and $f\coloneqq-\{VP_{[1,0]}V^\dag\}$ satisfy~\eqref{eq:triplemKPh}.
Here we show that they satisfy~\eqref{eq:triplemKPg}.
That they also satisfy~\eqref{eq:triplemKPf} is left to~\ref{sec:finishproof2ndThm}, as the computation is straightforward, but somewhat laborious. 
First, we consider the term $\frac13(g_t-g_{xxx})$. This proceeds in an entirely similar manner to our computation of $\frac13\,(\pa_t\lb V\rb-\pa_x^3\lb V\rb)$
in the case of the KP equation in Section~\ref{sec:KP}.
Here the skew bracket $[\,\cdot\,]$ just replaces the bracket $\lb\,\cdot\,\rb$, so that we have,
\begin{align*}
  \tfrac13\bigl(\pa_t[V]-\pa_x^3[V]\bigr)
     =&\;[V(P_{3,0}-P_{2,1}-P_{1,2}+P_{0,3})V]\\
      &\;-\,[V(P_{2,0}+2P_{1,1}+P_{0,2})VP_{\hat 1}V]\\
      &\;-\,[VP_{\hat 1}V(P_{2,0}+2P_{1,1}+P_{0,2})V]\\
      &\;-2\,[VP_{\hat 1}VP_{\hat 1}VP_{\hat 1}V].
\end{align*}
Second, we compute the term $\{g_x,f+g^2\}$. Combining the identity~\eqref{eq:gsquared} with our ansatz for $f$ we observe that
$f+g^2=\{VP_{\hat 1}V\}$. Hence to determine $\{g_x,f+g^2\}$, using the results of Lemma~\ref{lemma:Poppeprodforms}, we compute,
\begin{align*}
   [VP_{\hat 1}V]\{VP_{\hat 1}V\}=&\;\bigl[V(P_{1,1}+P_{0,2})\{VP_{\hat 1}V\}\bigr]\\
                                &\;+\bigl[VP_{\hat 1}V\{(P_{2,0}+P_{1,1})V\}\bigr]\\
                                &\;+[VP_{\hat 1}VP_{\hat 1}VP_{\hat 1}V]\\
                                &\;+[VP_{\hat 1}VP_{[1,0]}V^\dag P_{\hat 1}V^\dag],\\
   \{VP_{\hat 1}V\}[VP_{\hat 1}V]=&\;\bigl[V(P_{1,1}+P_{0,2})[VP_{\hat 1}V]\bigr]\\
                                &\;+\bigl[VP_{\hat 1}V[(P_{2,0}+P_{1,1})V]\bigr]\\
                                &\;+[VP_{\hat 1}VP_{\hat 1}VP_{\hat 1}V]\\
                                &\;-[VP_{\hat 1}VP_{[1,0]}V^\dag P_{\hat 1}V^\dag].
\end{align*}
If we combine these two results we find,
\begin{align*}
  \{g_x,f+g^2\}=&\; 2\,[V(P_{1,1}+P_{0,2})VP_{\hat 1}V]\\
                &\;+2\,[VP_{\hat 1}V(P_{2,0}+P_{1,1})V]\\
                &\;+2\,[VP_{\hat 1}VP_{\hat 1}VP_{\hat 1}V].
\end{align*}
Third, to compute the term $\bigl[g,\mathrm{D}(f+g^2)\bigr]$, we observe that $f+g^2=\{VP_{\hat 1}V\}=\pa_x\{V\}$ and thus $\mathrm{D}(f+g^2)=\pa_y\{V\}=\{VP_{[2,0]}V\}$.
Hence $\bigl[g,\mathrm{D}(f+g^2)\bigr]=\mathrm{ad}_{[V]}\{VP_{[2,0]}V\}$ and we can employ identity~\eqref{eq:mainanti2} in Lemma~\ref{lemma:commutatoridentities},
with $P_{a,b}\coloneqq P_{[2,0]}$, and use that linear combinations of $P_{a,b}$ in identity~\eqref{eq:mainanti2} naturally carry through.
Hence we observe that,
\begin{align*}
  \mathrm{ad}_{[V]}&\{VP_{[2,0]}V\}\\
                 =&\;\bigl[V\{(P_{3,0}-P_{1,2})\}\bigr]+[VP_{\hat 1}VP_{[2,0]}V]\\
                  &\;+[VP_{[1,0]}V^\dag P_{[2,0]}V^\dag]-\Bigl(\bigl[V(P_{2,1}-P_{0,3})[V]\bigr]\\
                  &\;+[VP_{[2,0]}VP_{\hat 1}V]+[VP_{[2,0]}VP_{[1,0]}V^\dag]\Bigr)\\
                 =&\;[V(P_{3,0}-P_{2,1}-P_{1,2}+P_{0,3})V]+[VP_{\hat 1}VP_{[2,0]}V]\\
                  &\;-[VP_{[2,0]}VP_{\hat 1}V]-\pa_y[VP_{[1,0]}V^\dag].
\end{align*}
If we combine the three computations above we observe,
\begin{align*}
  \tfrac13(g_t-g_{xxx})+ \{g_x,f+g^2\}-\bigl[g,&\mathrm{D}(f+g^2)\bigr]\\
  &=\pa_y[VP_{[1,0]}V^\dag].
\end{align*}
Since $h=[VP_{[1,0]}V^\dag]$, we observe that we have established that this $h$, $g\coloneqq[V]$ and $f\coloneqq-\{VP_{[1,0]}V^\dag\}$ satisfy equation~\eqref{eq:triplemKPg}.
If we combine this result, with our computation in~\ref{sec:finishproof2ndThm}, our proof of Theorem~\ref{thm:mKP} is complete.\qed
\end{proof}

\begin{remark}[Modified KP solution forms]\label{rmk:ncsolution}
We consider the following cases:
  
(a) \emph{Lifted mKP equations:} For any semi-additive solution kernel $p$ satisfying the linear constraints~\eqref{eq:twolinearforms},
the triple $(h,g,f)$ with $h\coloneqq[VP_{[1,0]}V^\dag]$, $g\coloneqq[V]$ and $f\coloneqq-\{VP_{[1,0]}V^\dag\}$ represents a solution to the lifted mKP equations~\eqref{eq:triplemKPh}--\eqref{eq:triplemKPf}.
These forms for $g$ and $f$ are not a solution to the non-commutative mKP equation \eqref{eq:mKPassign}--\eqref{eq:mKP}, nor, as expected, to the form \eqref{eq:eqamKPassign}--\eqref{eq:amKP}.
This is because these forms for $g$ and $f$ do not satisfy the constraint \eqref{eq:mKPassign}, i.e.\/ $f_x=g_y-[g,f]$.
In our solution form above, $g$ is skew-symmetric and $f$ is symmetric.
Thus in the constraint, $g_y$ and $[g,f]$ are skew-symmetric, while the lefthand side, $f_x$, is not.
Indeed we formulated the non-commutative mKP pair~\eqref{eq:pairmKPg}--\eqref{eq:pairmKPf} and lifted mKP equations~\eqref{eq:triplemKPh}--\eqref{eq:triplemKPf}, in order to alleviate this issue;

(b) \emph{Commutative case:} In this context, we observe that in the lifted mKP equations, the first equation \eqref{eq:triplemKPh} generates
the relation $h=\mathrm{D}g$, and thus equations \eqref{eq:triplemKPg} and \eqref{eq:triplemKPf} respectively collapse to,
\begin{align}
   \mathrm{D}g_y=&\;\tfrac13\bigl(g_t-g_{xxx}\bigr)+2g_x(f+g^2)\label{eq:scalarmKPg}\\
   \mathrm{D}\pa_y(f+g^2)=&\;\tfrac13\bigl(f_t-f_{xxx}\bigr)+2f_x(f+g^2)+2g\mathrm{D}g_y.\label{eq:scalarmKPf}
\end{align}
The forms $g\coloneqq[V]$ and $f\coloneqq-\{VP_{[1,0]}V^\dag\}$ still represent a solution to the commutative equations \eqref{eq:scalarmKPg}--\eqref{eq:scalarmKPf};

(c) \emph{Separate subclass of commutative solutions:} If $f\coloneqq\mathrm{D}g$, then it is straightforward to show that equation~\eqref{eq:scalarmKPf} above
is the result of the operator `$\mathrm{D}$' applied to the equation \eqref{eq:scalarmKPg}, in other words the former is simply consistent with the latter.
Note $f$ is equivalent to $h$ in this instance, and we are thus considering a separate subclass of solutions to those considered hitherto as $f$ and $h$ have the same form.
We observe with $f\coloneqq\mathrm{D}g$, equation \eqref{eq:scalarmKPg} becomes, $\mathrm{D}g_y=\tfrac13\bigl(g_t-g_{xxx}\bigr)+2g_x(\mathrm{D}g+g^2)$,
which is the potential form of the \emph{commutative mKP equation};

(d) \emph{Non-commutative mKdV equation:} If we assume the fields $(h,g,f)$ in the lifted mKP equations~\eqref{eq:triplemKPh}--\eqref{eq:triplemKPf} are $y$-independent,
and $f\equiv 0$, then the lifted mKP equations collapse to the non-commutative potential mKdV equation, with solution $g\coloneqq[V]$; see Blower and Malham~\cite{BM}.
\end{remark}

\section{Discussion}\label{sec:discussion}
Let us first establish the aforementioned quasi-trace solution form for the solution to the non-commutative KP equation~\eqref{eq:KPsingform}.
The following result is the semi-additive extension of the corresponding result for Hankel operators in~\cite[Lemma~4]{DMSWb}.
We suppress any $y$-dependence in the following as it plays a passive role. We define the linear quasi-trace operator as follows.
\begin{definition}[Quasi-trace]\label{def:quasitrace}
Suppose $F$ is a trace-class operator on $\Vb$ with matrix-valued kernel $f$. Then we define the quasi-trace `$\mathrm{qtr}$' of $F$ to be,
\begin{equation*}
\mathrm{qtr}\,F\coloneqq\int_{-\infty}^0f(z,z)\,\mathrm{d}z,
\end{equation*}  
provided the integral of each component of $f$ is finite.
\end{definition}
The usual trace would also involve the matrix trace of $f$ in the integrand. We use the notation `$\lb F\rb_{0,0}$' to
denote the continuous kernel $f$ of any Hilbert--Schmidt valued operator $F$, evaluated at $z=\zeta=0$, i.e.\/ $\lb F\rb_{0,0}\equiv f(0,0;x)$.
\begin{lemma}[Quasi-trace identity]\label{lemma:quasitraceid}
Suppose that $H$ and $\hat{H}$ are semi-additive Hilbert--Schmidt valued operators, and $F$ is any Hilbert--Schmidt valued operator, on $\Vb$ with parameter $x\in\R$.
Suppose the kernels of $H$ and $\hat{H}$ are continuously differentiable, while the kernel of $F$ is continuous, on $(-\infty,0]^{\times2}$.
Then we have the following quasi-trace identity,
\begin{equation*}
  \lb HF\hat{H}\rb_{0,0}\equiv\mathrm{qtr}\,\bigl((\pa_{\ell}+\pa_{r})(HF\hat{H})\bigr).
\end{equation*}
\end{lemma}
\begin{proof}
  The righthand side is equal to the integral over $(\xi,\eta,z)\in(-\infty,0]^{\times3}$ of,
  \begin{align*}
    \pa_\ell &h(z+x,\xi+x)\,f(\xi,\eta;x)\,\hat{h}(\eta+x,z+x)\\
    &+h(z+x,\xi+x)\,f(\xi,\eta;x)\,\pa_r\hat{h}(\eta+x,z+x).
  \end{align*}
  We can replace the partial derivatives $\pa_\ell$ and $\pa_r$ by $\pa_z$.
  We can then use the partial integration formula with respect to $z$ for the first term.
  The boundary term generated equals the integral of $h(x,\xi+x)\,f(\xi,\eta;x)\,\hat{h}(\eta+x,x)$ over $(\xi,\eta)\in(-\infty,0]^{\times2}$,
  while the integral term generated cancels with the second term above. This gives the result. 
\qed
\end{proof}
We then have the following immediate corollary.
\begin{corollary}[Quasi-trace form]\label{cor:quasitrace}
Suppose the semi-additive Hilbert--Schmidt operator $P$ with parameters $x,y\in\R$, satisfies the linear system of equations~\eqref{eq:linearform},
and $G=PV$ is the Hilbert--Schmidt operator solution to the linear integral equation \eqref{eq:GLM}.
Then the matrix-valued function $g=g(x,y;t)$ given by the kernel $g\coloneqq\lb G\rb_{0,0}$,
from Theorem~\ref{thm:KP} satisfies the non-commutative KP equation~\eqref{eq:KPsingform},
and has the quasi-trace form,
\begin{equation}\label{eq:quasitracesolnform}
  g(x,y;t)\equiv\mathrm{qtr}\,\bigl((\pa_\ell+\pa_r)G\bigr).
\end{equation}
\end{corollary}
\begin{proof}
Using Lemma~\ref{lemma:quasitraceid} with $H=\hat{H}=P$ and $F=V$, and the identities~\eqref{eq:partialfractions}, we have,
\begin{align*}
&&\lb PVP\rb_{0,0}=&\;\mathrm{qtr}\,\bigl((\pa_\ell)VP+PV(\pa_rP)\bigr)\\
\Leftrightarrow&& \lb PV\rb_{0,0}-\lb P\rb_{0,0}=&\;\mathrm{qtr}\,\bigl((\pa_\ell)V+V(\pa_rP)-\pa_xP\bigr)\,
\end{align*}  
where recall that for a semi-additive operators we have, $\pa_xP=(\pa_\ell+\pa_r)P$.
It is straighforward to show that for a semi-additive operator $\mathrm{qtr}\,(\pa_xP)=\lb P\rb_{0,0}$.
Cancelling these terms on both sides of the result just above and using that $G=PV$ gives the result.
\qed
\end{proof}  
There are well-known quasi-determinant solutions to the non-commutative KP and mKP equations,
see Etingof, Gelfand and Retakh~\cite{EGR97,EGR98}, Gilson and Nimmo~\cite{GN}, Gilson \textit{et al.} \cite{GNS}, Hamanaka~\cite{Hamanaka} and Sooman~\cite{So}.
See item (1) in our discussion just below.

There are many further generalisations of the results herein we intend to pursue, as follows. 

(1) \emph{Solitons, quasi-determinants and $\tau$-functions:} Connecting our quasi-trace formula~\eqref{eq:quasitracesolnform} for the non-commutative KP equation
to the quasi-determinant solution forms in Etingof, Gelfand and Retakh~\cite{EGR97,EGR98} and Gilson and Nimmo~\cite{GN}, is a natural immediate next step.
The latter represent $\tau$-functions with specific forms corresponding to multi-soliton solutions. And then a natural question is whether the quasi-trace formula~\eqref{eq:quasitracesolnform}
can be extended to lifted mKP solutions, and, can any such extended quasi-trace formula be related to quasi-determinant solution formulae that may exist for the lifted mKP solutions.
As already mentioned, Gilson \textit{et al.} \cite{GNS} established quasi-determinant formulae for the non-commutative modified KP equations~\eqref{eq:eqamKPassign}--\eqref{eq:amKP}.
Further, are there efficient numerical proceodures for computing Fredholm quasi-determinents, much like those developed for ordinary Fredholm determinants by Bornemann~\cite{Bornemann}? 

(2) \emph{Fredholm--Grassmannian flow:} As is well-known, the solution flow to the KP equation can be interpreted as flow on a Fredholm Grassmann manifold;
see Sato~\cite{SatoI,SatoII} and Segal and Wilson~\cite{SW}. And Kodama~\cite{Kodama} has classified all possible KP multi-soliton interactions and shown that
they can be parametrised within the context of finite dimensional Grassmannians. The form of the GLM equation~\eqref{eq:GLM} in particular predicates that
the solution flow occurs on a given coordinate patch of a Fredholm Grassmannian;
see Doikou, Malham and Stylianidis~\cite{DMS}, Doikou \textit{et al.\/} \cite{DMSWa,DMSWb,DMSWc}, Malham~\cite{MKdV} and Blower and Malham~\cite{BM}.
Also see Beck \textit{et al.\/} \cite{BDMSa,BDMSb}.
It would be of considerable interest to: (a) Explore the interpretation of the non-commutative KP and lifted mKP flows as Fredholm Grassmannian flows.
For example, flows with singularities (see, for example, Pelinovsky~\cite{Pelinovsky} or P\"oppe~\cite{PKP})
correspond to a poor choice of coordinate patch, locally, when projecting from the underlying linear flow
onto the Fredholm Grassmannian via the GLM equation. This might be the case for the infinite `trenches' mentioned in \ref{sec:numerics};
(b) Extend Kodama's classification to the non-commutative KP and lifted mKP flows, and in particular show that
the multi-solitons parametrised as such, are finite dimensional subflows of the underlying Fredholm Grassmannian flow. Also see Chuanzhong, Mironov and Orlov~\cite{CMO}; 
(c) Utilise the Fredholm Grassmannian flow to improve and extend known solution regularity results for the non-commutative KP and lifted mKP equations,
for example in manner achieved by Grudsky and Rybkin~\cite{GR}; and
(d) Connect the (quasi-) determinantal bundle associated with the underlying Fredholm Grassmannian with the quasi-determinant solution forms 
found by Gilson and Nimmo~\cite{GN} for the non-commutative KP equation, and Gilson \textit{et al.\/} \cite{GNS} for the modified KP equations~\eqref{eq:eqamKPassign}--\eqref{eq:amKP},
or those for the lifted mKP equations discussed in (1) above. 

(3) \emph{Non-commutative KP and lifted mKP hierarchies:}
The pre-P\"oppe algebra we have established is the appropriate context to study the whole of the non-commutative KP and lifted mKP hierarchies,
in particular, to determine certain classes of solutions to the whole hierarchies, as was achieved by Malham~\cite{MKdV} for the non-commutative KdV hierarchy
and Blower and Malham~\cite{BM} for the non-commutative NLS/mKdV hierarchy. Another immediate goal is to make this extension.

(4) \emph{Abstract pre-P\"oppe algebra:} In Malham~\cite{MKdV} and Blower and Malham~\cite{BM} the P\"oppe algebra was pulled back to a more abstract form,
more insightful for constructing the hierarchy. Such an abstraction might prove useful here.
The idea is to generate an abstract algebraic context analogous to the study of endomorphism algebras on the shuffle algebra in Malham and Wiese~\cite{MW}  
and Ebrahimi-Fard \textit{et al.\/} \cite{E-FLMM-KW}. See Malham~\cite{MKdV} and Blower and Malham~\cite{BM} for more details. Also see Manin~\cite{Manin},
as well as Magnot, Reyes and Rubtsov~\cite{MRR}.

(5) \emph{Pre-P\"oppe algebra, extensions:}  In Remark~\ref{rmk:algext} we alluded to extensions of the pre-P\"oppe algebra that may be useful to further investigations.
The pre-P\"oppe algebra we developed in Sections~\ref{sec:Poppealg}--\ref{sec:mKP} is generated by monomials of the form $\lb W\rb$ where
$W=\tilde VP_{a_1,b_1}\tilde V\cdots\tilde VP_{a_k,b_k}\tilde V$, and the $\tilde V$'s could be either $V$ or $V^\dag$, with the pre-P\"oppe product given in Lemma~\ref{lemma:Poppeprodmod}.
In particular, for example, in our computations for the lifted mKP equation in Section~\ref{sec:mKP} we considered the products of terms generated by $[V]$ and its derivatives
with respect to $x$ and $y$. However, from an algebra perspective, we could extend the generator set, from just $[V]$, to include $\{PV\}$ as well.  
Using Lemma~\ref{lemma:Poppeprodforms}, and noting Remark~\ref{rmk:Vcorrection}, it is straightforward to show that, for example with $\pa\coloneqq\pa_\ell+\pa_{r}$, we have,
\begin{align*}
\{PV\}[V]=&\;[VP_{\hat{1}}V]+[VP_{[1,0]}V^\dag]-2\,\pa_{\ell}[V],\\
[V]\{PV\}=&\;[VP_{\hat{1}}V]-[VP_{[1,0]}V^\dag]-2\,\pa_{r}[V],\\
\{PV\}\{PV\}=&\;\{VP_{\hat{1}}V\}-\{VP_{[1,0]}V^\dag\}-2\,\pa\{PV\}.
\end{align*}
Further, we can also prove results of the form,
\begin{equation*}
\mathrm{ad}_{\{PV\}}[VP_{\hat{1}}V]=\pa_y[V]+\pa_x[VP_{[1,0]}V^\dag]-2\,\pa^\sharp[VP_{\hat{1}}V],
\end{equation*}
where, $\pa^\sharp\coloneqq\pa_\ell-\pa_{r}$. And so forth.
For another example, for the quasi-trace formula developed above, in both Lemma~\ref{lemma:quasitraceid} and Corollary~\ref{cor:quasitrace},
we utilised similar computations to those used to establish these results.
Extending the algbera in this way, in principle, opens the door to considering different classes of solutions, as well as considering further integrable systems.  

(6) \emph{Linear systems:} The is also a linear systems approach to integrable systems with close connections to the P\"oppe approach we have employed herein.
See for example, Blower~\cite{Blower}, Blower, Brett and Doust~\cite{BBD}, Blower and Doust~\cite{BD} and Blower and Newsham~\cite{BN}.
Our goal here would be to extend this approach to the non-commutative KP and mKP equations and simultaneously crystalise its connection to the P\"oppe approach.

(7) \emph{Numerical simulation of the non-commutative KP and lifted mKP equations:}
There are many components/aspects of our numerical simulations in \ref{sec:numerics} that require further attention as follows:
(a) A detailed numerical analysis of the GLM, FTT-FD-exp and FFT2-exp algorithms employed in \ref{sec:numerics}, 
together with a more thorough span of numerical investigations of solutions to the non-commutative KP equation, are both required.
We would also like to add the algorithm involving the GLM equation and two-dimensional fast Fourier transform mentoined in Remark~\ref{rmk:GLMFFT2},
to this numerical analysis, as well as to implement it and use it to construct a much larger variety of solutions to the non-commutative KP equation.
Such GLM methods, by using different coordinate patches as mentioned in item (2) above, may also provide a tractible approach to solutions
with singularities, whether at individual points or as one-dimensional objects such as the infinite `trenches', in the $(x,y)$-plane;
(b) A similar numerical analysis and investigation is required for the lifted mKP equations, which are not included in \ref{sec:numerics}.
While we expect the investigation and implementation of the FTT-FD-exp and FFT2-exp algorithms for the lifted mKP equations to be relatively straightforward,
the solution approximation via the GLM equation requires more work. For example, in the GLM equation~\eqref{eq:GLMexplicit}, since in this case $Q\coloneqq P^2$,
the second $p$ factor in the integrand on the right is replaced by (suppressing the implicit $y$ and $t$ dependence in $q$ and $p$),
\begin{equation*}
q(z,\zeta;x)=\int_{-\infty}^0p(z+x,\nu+x)p(\nu+x,\zeta+x)\,\rd\nu.
\end{equation*}
Here $q$ represents the square-integrable kernel of $Q$. This integral is straightforwardly approximated by a Riemann sum in precisely the same
manner to how the integral term in the GLM equation is approximated in the non-commutative KP equation case, as outlined in detail in \ref{sec:numerics}.
This generates and approximation of the form $\widehat{Q}_{n,m}(z_\ell,\zeta_k;t)\coloneqq q(z_\ell+x_n,\zeta_k+x_n;y_m,t)$ on nodes in the truncated $(x,y)$-domain.
The integral term for the GLM equation in the lifted mKP equations case given by,
\begin{equation*}
\int_{-\infty}^0q(z,\xi;x)g(\xi,\zeta)\,\rd\xi,
\end{equation*}
can, as previously, then we also approximated by a Riemann sum utilising the approximation $\widehat{Q}_{n,m}=\widehat{Q}_{n,m}(z_\ell,\zeta_k;t)$.
This means that we can straightforwardly generate an approximation for $[G]=2\lb VPV^\dag\rb=2\lb P(\id-P^2)^{-1}\rb$ in a similar manner to that outlined
for the non-commutative KP equation in \ref{sec:numerics}. However for the lifted mKP equation we also need to generate $f\coloneqq-\{VP_{[1,0]}V^\dag\}$ and $h\coloneqq[VP_{[1,0]}V^\dag]$.
To achieve this, we recall that we can express $f$ and $h$ respectively in the form $f=-\lb VP_{[1,0]}V^\dag\rb+\lb V^\dag P_{[1,0]}V\rb$ and $h=\lb VP_{[1,0]}V^\dag\rb+\lb V^\dag P_{[1,0]}V\rb$.
In other words, it is sufficient to compute $\lb VP_{[1,0]}V^\dag\rb$ and $\lb V^\dag P_{[1,0]}V\rb$. If we set $F\coloneqq VP_{[1,0]}V^\dag$, then $F$ satisfies, 
\begin{equation*}
(\id-P)F(\id+P)=P_{[1,0]}.
\end{equation*}
We can consider this equation as a linear integral equation for the kernels of the operators $P$ and $F$.
Assuming we can suitably approximate the left, i.e.\/ $\pa_z$, and right, i.e.\/ $\pa_\zeta$, derivatives of the kernel of $P$,
we can compute the kernel of $F$ by successively solving the following equivalent system of linear integral equations,
\begin{equation*}
(\id-P)Y=P_{[1,0]}\quad\text{and}\quad F(\id+P)=Y.
\end{equation*}
Successively solving these linear integral equations can be achieved in a similar manner to that outlined for the GLM equation in \ref{sec:numerics}.
An approximation for $\lb V^\dag P_{[1,0]}V\rb$ can be generated analogously. Implementing this approach is one of our immediate goals;
(c) In both of the cases (a) and (b), we assume we are either given the kernel of the semi-additive operator $P$ representing the scattering data,
or, if we use the algorithm outlined in Remark~\ref{rmk:GLMFFT2}, the initial data $P_0$ for $P$ at time $t=0$.
Ideally, we would like to use the GLM method to compute the solution $g$, in either the non-commutative KP or lifted mKP cases, given initial $g_0$ for $g$ at time $t=0$. 
Specifically, given such initial data $g_0=g_0(x,y)$, we would need to solve the scattering problem to determine the scattering data $p_0$ to the linear system of equations~\eqref{eq:linearform},
which is semi-additive and satisfies the constraints~\eqref{eq:constraints}. In principle, the data can be carried forward in time, exactly, in Fourier space with respect to $x$ and $y$.
We could then determine $g=g(0,0;x,y,t)$ by numerically solving the GLM equation at that specific time $t$ of choice, without having to numerically time-step to that time.
We intend to explore all of these directions in due course.

(8) \emph{Lax formulation:} The non-commutative KP~\eqref{eq:KPsingform} and modified KP equations~\eqref{eq:eqamKPassign}--\eqref{eq:amKP} have well-known Lax formulations,
see Dimakis and M\"uller--Hoissen~\cite{DMH}, Gilson and Nimmo~\cite{GN}, Gilson \textit{et al.\/} \cite{GNS}, Kuperschmidt~\cite{K} and Wang and Wadati~\cite{WW}.
However any Lax formulation for the lifted mKP is yet to be determined.

(9) \emph{Novikov--Veselov and modified Novikov--Veselov equations:} The Novikov--Veselov (NV) equation is also considered a natural extension of the KdV equation to two dimensions,
see for example, Bogdanov~\cite{Bogdanov}, Croke~\cite{Croke}, Lassas, Mueller, Siltanen and Stahel~\cite{LMSS} and Nickel, Serov and Sch\"urmann~\cite{NSS}.
There is also an NV hierarchy, see for example Mironov~\cite{Mironov}.
Similarly, the modified NV equation is also considered as a two-dimensional extension of the mKdV equation, see Ferapentov~\cite{Ferapontov} and Ta\u{\i}manov~\cite{Taimanov}.
Pfaffian solution forms for the NV equation have been established by \"Unal~\cite{Metin}. This suggests the approach we consider herein for the non-commutative KP equation
could be extended to the NV equation and a non-commutative generalisation, and perhaps their modified versions.
There are also other associated KP systems of interest in this context, such as the elliptic KP-type system directly linearised by Fu and Nijhoff~\cite{FN3}.

(10) \emph{The KP and KPZ equations:} The connections established by Quastel and Remenik~\cite{QR} between the KP equation,
Tracy--Widom distributions in random matrix theory and the KPZ equation, which represents a model for random growth, direct polymers, percolation and so forth, are particularly intriguing.
Their work includes more general results for the matrix KP equation and highlights the need to extend the connection between the KP equation and Tracy--Widom distributions
to the non-commutative case. Results established in items (1) and (2) above might help in tackling such an extension.

(11) \emph{Supersymmetric extensions:} As we mentioned in the introduction, there are supersymmetric extensions of the KP and mKP equations.
In particular see Brunelli and Das~\cite{BrunelliDas}, Delduc, Gallot and Sorin~\cite{DGS}, Manin and Radul~\cite{ManinRadul} and Nishino~\cite{Nishino}.
It would be of interest to consider equivalent extensions for the non-commutative KP and lifted mKP equations herein, and use (perhaps a modified form of)
the pre-P\"oppe algebra to directly linearise the resulting supersymmetric extensions.

(12) \emph{Log-potential mKP form:} Ideally of course, we would also like to directly linearise the log-potential mKP form~\eqref{eq:lpKP} for $q$.
To tackle the log-potential mKP equation, we may need to adapt the semi-additive form we considered herein.
This is because the identifications $g\coloneqq q^{-1}q_x$ and $f\coloneqq q^{-1}q_y$ we considered in the introduction,
which generated the non-commutative mKP equation form \eqref{eq:mKPassign}--\eqref{eq:mKP}, do not distinguish $h$ and $f$. 
The significance of the log-potential mKP equation~\eqref{eq:lpKP} is that both the non-commutative mKP formulations,
\eqref{eq:eqamKPassign}--\eqref{eq:amKP} and \eqref{eq:mKPassign}--\eqref{eq:mKP}, can be directly generated from it. 
From this perspective and its form, its solution thus appears to represent the natural $\tau$-function or quasi-determinant solution form.
On the other hand, for the lifted mKP equations \eqref{eq:triplemKPh}--\eqref{eq:triplemKPf}, we have $f=\{V^\dag P_{[1,0]}V\}$ and $h=[V^\dag P_{[1,0]}V]$.
We thus now seek the formulation/theory that encompasses the whole of this non-commutative picture.

\section{Declarations}

\subsection{Funding and conflicts or competing interests}
SJAM received funding from the EPSRC for the Mathematical Sciences Small Grant EP/X018784/1.
There are no conflicts of interests or competing interests. 

\subsection{Data availability statement}
No data was used in this work.

\appendix

\section{Solitons (commutative KP)}\label{sec:solitons}
We derive the soliton solution form for the commutative form of the KP equation~\eqref{eq:KPsingform}.
We use interaction forms of several non-commutative versions of these solutions as initial data in our numerical simulations in~\ref{sec:numerics},
hence our inclusion of this standard derivation herein---see, for example, Drazin and Johnson~\cite{DJ}.
Suppose $A$, $B$, $\Lambda$ and $\Omega$ are constant matrices, and that,
\begin{equation*}
\Lambda\coloneqq A^2-B^2\quad\text{and}\quad \Omega\coloneqq(A+B)^3+(A+B)^{-1}\Lambda^2.
\end{equation*}
Assume that $A$ and $B$ commute, i.e.\/ $[A,B]=O$, the zero matrix. Thus all four quantities commute. Note that, also, we have $\Omega=4(A^3+B^3)$.
If the scattering data $p$ has the form \eqref{eq:semiadditiveform}, then the GLM equation \eqref{eq:GLM} with $Q\equiv P$ has the following explicit form
(suppressing the implicit $y$ and $t$ dependence in $p$ and $g$),
\begin{equation}
  p(z+x,\zeta+x)\!=\!g(z,\zeta;x)-\!\!\int_{-\infty}^0\!\!\!g(z,\xi;x)p(\xi+x,\zeta+x)\,\rd\xi. \label{eq:GLMexplicit}
\end{equation}
Note this applies in both the non-commutative, in which case the quantities $p$ and $g$ shown therein are matrix-valued, and the commutative, context.
Since our goal is to compute $\pa_x g(0,0;x,y,t)$, which represents the solution to the KP equation rather than its potential form, 
we can immediately set $z=0$ in \eqref{eq:GLMexplicit}.
Now suppose $p=p(z+x,\zeta+x;y,t)$ has the form,
\begin{equation}\label{eq:scatteringdata}
p=-(A+B)\exp\bigl(A(z+x)+B(\zeta+x)+\Lambda y+\Omega t\bigr).
\end{equation}
This exponential form for $p$ naturally satisfies \eqref{eq:constraints} and \eqref{eq:secondlinearform}, and thus in the latter case, equivalently \eqref{eq:linearform}. 
We look for a solution $g=g(0,\zeta;x,y,t)$ to~\eqref{eq:GLMexplicit} of the form,
\begin{equation*}
g=H(x,y,t)\exp\bigl(A(z+x)+B(\zeta+x)\bigr).
\end{equation*}
Substituting these two forms for $p$ and $g$ into~\eqref{eq:GLMexplicit} and solving for $H=H(x,y,t)$, reveals that,
\begin{equation}\label{eq:solitonsolution}
\pa_xg(0,0;x,y,t)=-\tfrac14(A+B)^2\,\mathrm{sech}^2\Theta,
\end{equation}
where, $\Theta\coloneqq\tfrac12\bigl((A+B)x+\Lambda y+\Omega t\bigr)$.
We observe that for this solution form, $[g_x,g_y]\equiv O$. Hence this is a solution to the commutative KP equation, which can be checked by direct substitution. 
In Figure~\ref{fig:soliton} we show the soliton solution corresponding to the matrices,
\begin{equation*}
A=\begin{pmatrix} 1.55 & 0 \\ 0 & 1.45 \end{pmatrix}  \quad\text{and}\quad B=\begin{pmatrix} 1.05 & 0 \\ 0  & 1.02 \end{pmatrix}.
\end{equation*}
The solutions shown in the middle and bottom panels in the Figure are computed at time $t=0.2$ using two independent numerical methods, by:
(i) Numerically solving the linear integral equation~\eqref{eq:GLMexplicit} using the scattering data~\eqref{eq:scatteringdata};
and (ii) using the FFT-FD-exp scheme. Both of these schemes are outlined in detail next, in \ref{sec:numerics}.
The respective solitons shown in Figure~\ref{fig:soliton}, match the soliton form~\eqref{eq:solitonsolution} to within numerical error.

\begin{figure*}
  \begin{center}
    \includegraphics[width=8cm,height=7cm]{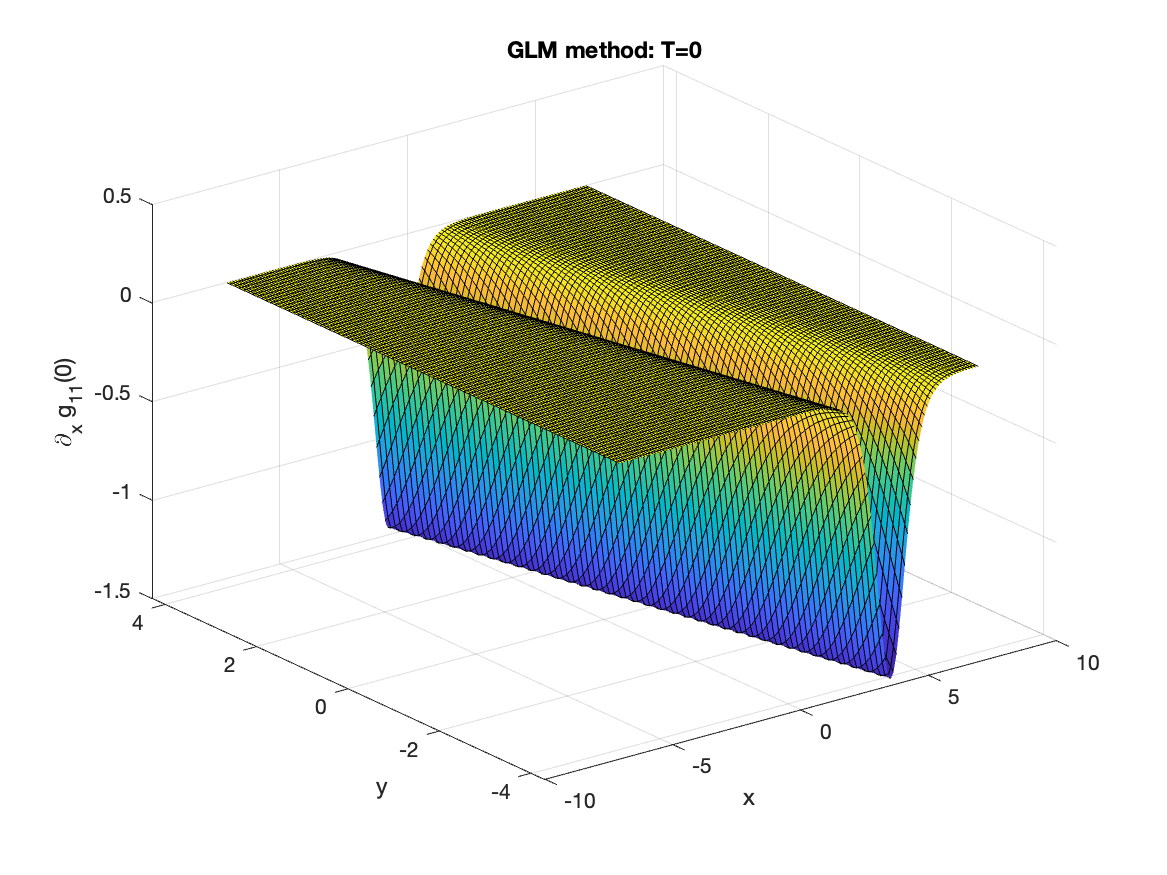}
    \includegraphics[width=8cm,height=7cm]{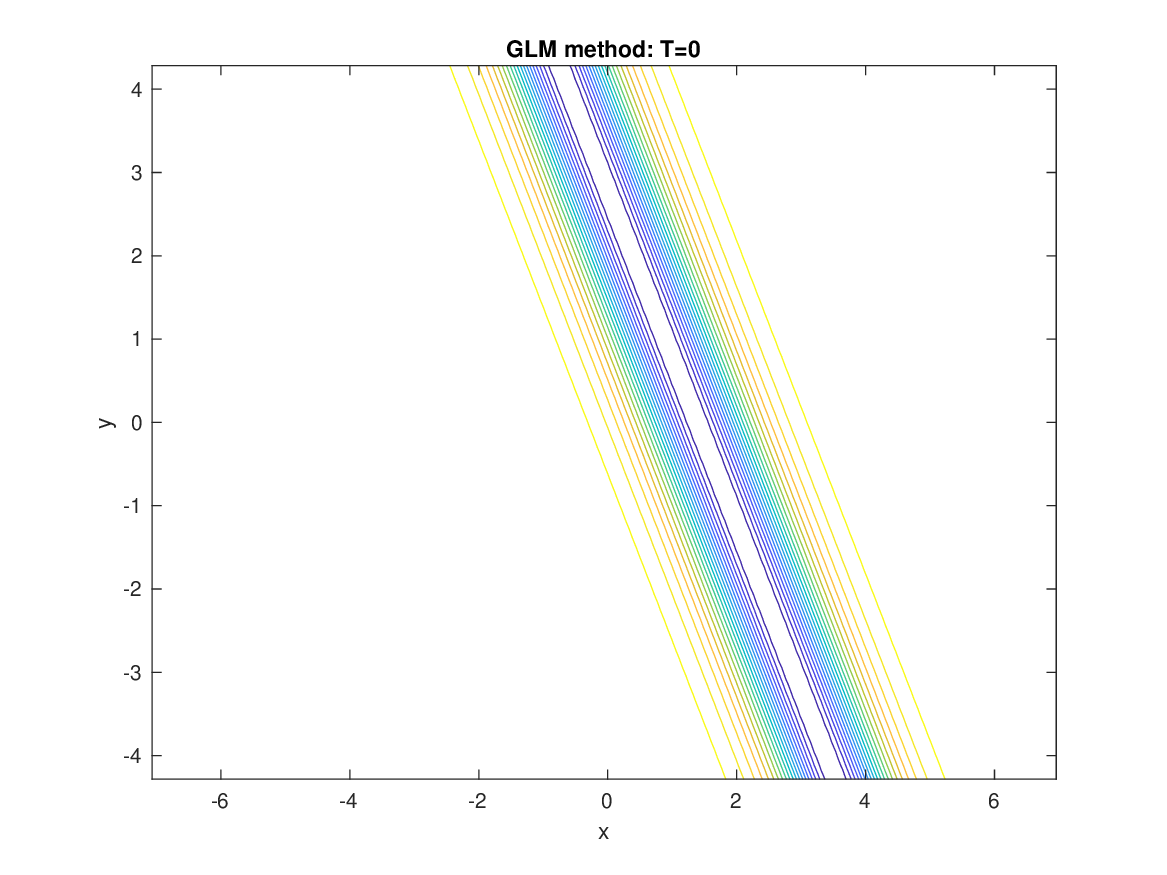}\\ 
    \includegraphics[width=8cm,height=7cm]{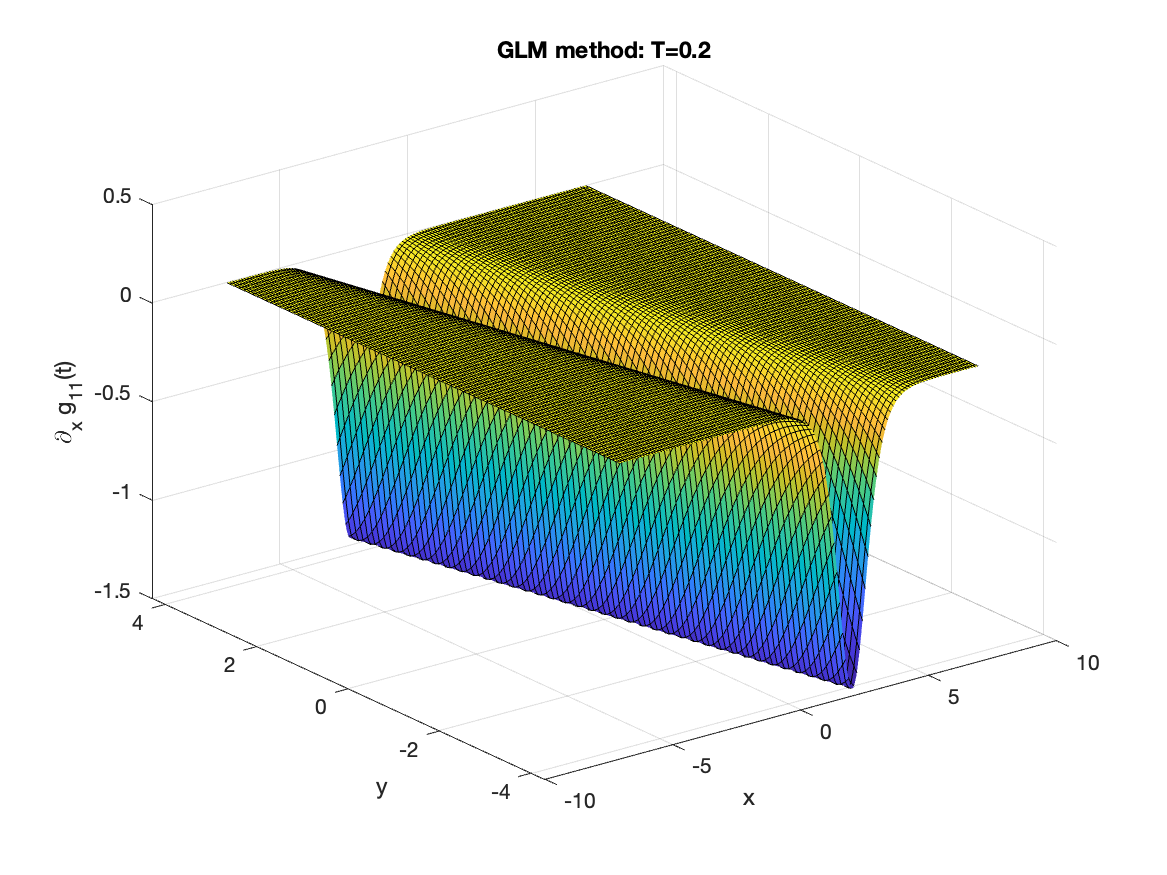}
    \includegraphics[width=8cm,height=7cm]{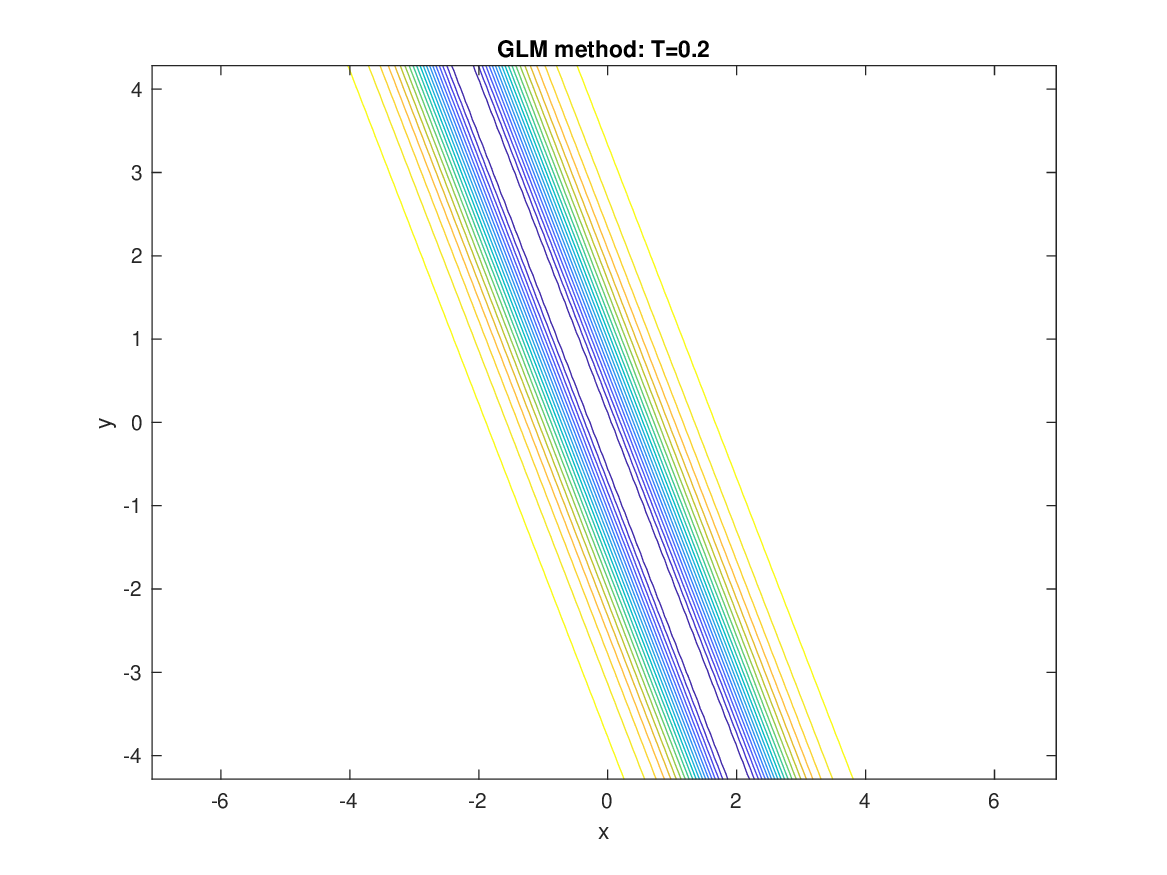}\\ 
    \includegraphics[width=8cm,height=7cm]{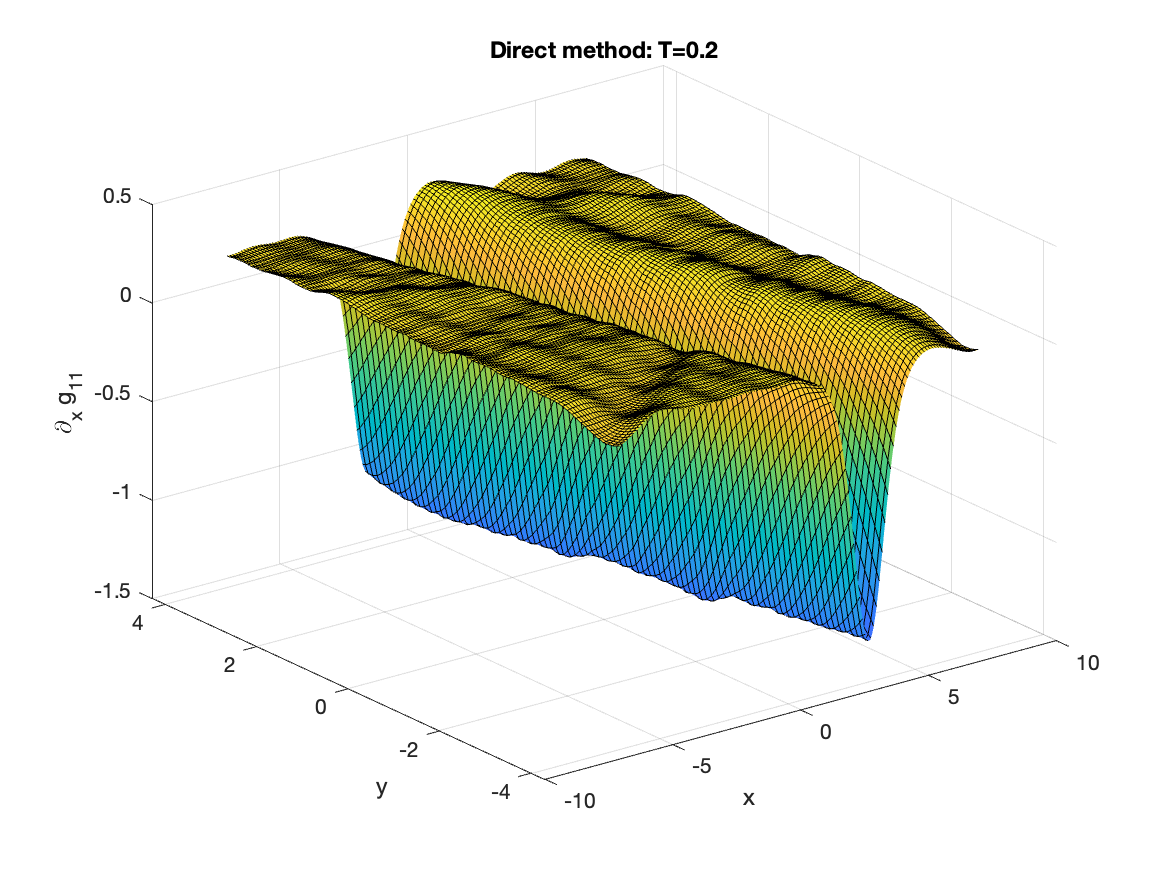}
    \includegraphics[width=8cm,height=7cm]{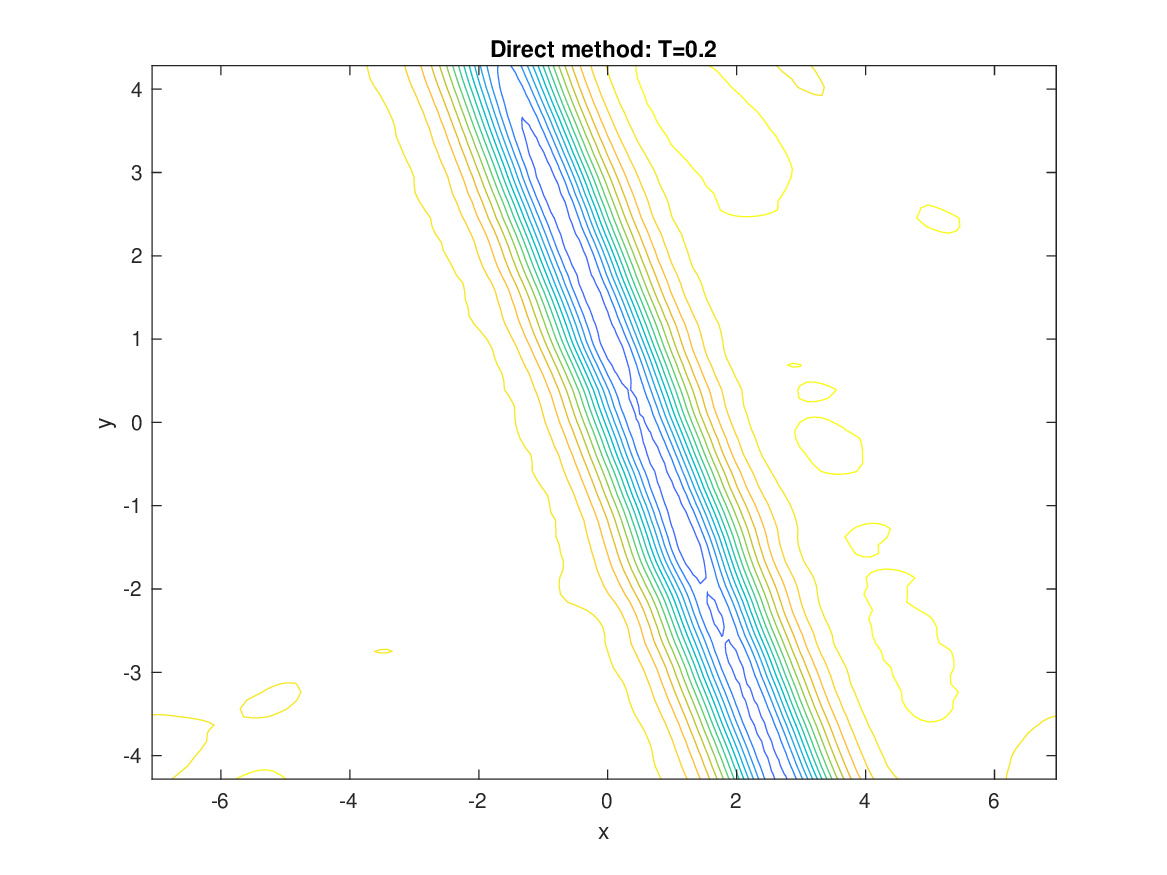}
  \end{center}
  \caption{We plot the soliton solution corresponding to~\eqref{eq:solitonsolution} at $t=0$ in the top panel set,
    computed by numerically solving the GLM equation~\eqref{eq:GLMexplicit}. In the middle and bottom panel sets,
    we plot the soliton solution computed at time $t=0.2$, respectively computed by numerically solving the GLM equation at that time,
    and then using the FFT-FD-exp scheme outlined in \ref{sec:numerics}.
    The right-hand panels give the corresponding contour plots.}
\label{fig:soliton}
\end{figure*}

\section{Simulations (non-commutative KP)}\label{sec:numerics}
Herein we demonstrate how we can numerically solve the GLM equation~\eqref{eq:GLM} to obtain approximate solutions to the non-commutative KP equation~\eqref{eq:KPsingform}.
For comparison, we also solve the non-commutative KP equation~\eqref{eq:KPsingform} directly using two pseudo-spectral algorithms based on 
an exponential split-step method in time. The first pseudo-spectral algorithm, which we nominate the FFT-FD-exp algorithm,
utilises the Fourier transform of the solution in the $x$-direction, and computes derivatives in the $y$-direction using central finite differences.
The second pseudo-spectral algorithm, which we nominate the FFT2-exp algorithm, utilises the Fourier transform in both the $x$- and $y$-directions,
and employs the window method outlined by Kao and Kodama~\cite{KaoKodama} to deal with the non-periodic boundary conditions. 

First, we outline the simple numerical method we used to solve the GLM equation~\eqref{eq:GLM}. This is given explicitly in~\eqref{eq:GLMexplicit},
though now we adhere to the context where both $p$ and $g$ are non-commutative. Our goal is to compute $g=g(0,0;x,y,t)$, so we set $z=0$ therein,
and assume $t\geqslant0$ is fixed, i.e.\/ the given single value time at which we wish to approximately compute the solution.
For a given $x$- and $y$-domain lengths $L_x>0$ and $L_y>0$, respectively, we assume $x\in[-L_x/2,L_x/2]$ and $y\in[-L_y/2,L_y/2]$.
However, we assume $\zeta,\xi\in[-L_x/2,0]$. We discretise $x$, $\zeta$ and $\xi$ to the nodal values $x_n$, $\zeta_n$ and $\xi_n$,
each with separation step $h_x>0$, and with $\zeta_n=\xi_n=x_n$ for $x_n\leqslant0$. We discretise $y$ to $y_m$ with separation step $h_y>0$.
For a given matrix-valued function form $p=p(\cdot,\cdot;t)$, for each nodal point $(x_n,y_m)$, we generate the matrix
$\widehat{P}_{n,m}(\zeta_k;t)\coloneqq p(x_n,\zeta_k+x_n;y_m,t)$, i.e.\/ for a given node $(x_n,y_m)$, we replace each matrix entry $p_{ij}$
of $p$ by the row vector $p_{ij}(x_n,\zeta_k+x_n;y_m,t)$ indexed by $k$.
Similarly, for each nodal point $(x_n,y_m)$, we generate the matrix $\widehat{Q}_{n,m}(\xi_\ell,\zeta_k;t)\coloneqq p(\xi_\ell+x_n,\zeta_k+x_n;y_m,t)$,
where now we replace each matrix entry $p_{ij}$ of $p$ by the matrix $p_{ij}(\xi_\ell+x_n,\zeta_k+x_n;y_m,t)$ indexed by $(\ell,k)$.
We also set $\widehat{G}(x_n,y_m;t)$ to be the matrix of unknown values $\widehat{G}_{n,m}(\zeta_k;t)\coloneqq g(0,\zeta_k;x_n,y_m,t)$, as follows.
As above, we imagine replacing the entry $g_{ij}$ in the matrix $g=g(0,\zeta;x,y,t)$ by the row vector $g_{ij}(0,\zeta_k;x_n,y_m,t)$ indexed by $k$.
In other words, the $(ik,j)$th position in the matrix $\widehat{G}_{n,m}$ is the $k$ component of the row vector $g_{ij}(0,\zeta_k;x_n,y_m,t)$.
We approximate the integral in the linear Fredholm equation~\eqref{eq:GLMexplicit} with $z=0$, by a left Riemann sum.
This means that an approximate solution to~\eqref{eq:GLMexplicit} with $z=0$, is given by the solution to the linear algebraic problem,
for any given $t\geqslant0$ and $(x_n,y_m)$:
\begin{equation*}
\widehat{P}=\widehat{G}-h_x\widehat{Q}\widehat{G}\quad\Leftrightarrow\quad \widehat{G}=\widehat{P}(\id-h_x\widehat{Q})^{-1}.
\end{equation*}
In the equation on the left, the Riemann sum is realised by the matrix multiplication and factor $h_x$. 
An approximation to the solution $g=g(0,0;x_n,y_m,t)$, is then generated by evaluating each of the row vectors in $\widehat{G}_{n,m}$ 
at the far right-hand entries corresponding to $\zeta_k=0$---recall $\zeta_k\in[-L_x/2,0]$. In other words, if the number of $x_n$ nodes is $M$,
with $x_0=-L_x/2$ and $x_{M-1}=L_x/2-h_x$, then we extract out the $\bigl(\bigl(i(M/2+1)\bigr),j)\bigr)$th entries of $\widehat{G}_{n,m}$, corresponding to $z_k=0$,
to generate the approximation $g_{ij}(0,0;x_n,y_m,t)$. We compute the solution to the non-commutative KP equation, namely $\pa_xg(0,0;x_n,y_m,t)$,
by computing the $x$-derivative via a finite central difference, assuming zero Neumann boundary conditions.

\begin{example}\label{ex:soliton}
We plot the time evolution of a basic soliton solution of the form~\eqref{eq:solitonsolution} for the commutative KP equation in Figure~\ref{fig:soliton}.
The top two panel sets in the Figure are generated using the GLM approach,
first setting the time $t=0$ for the top set, and then $t=0.2$ for the middle panel set.
We used $M=2^7$ nodes for both $x_n$ and $y_m$, and set $L_x=9\pi/2$ and $L_y=4\pi$.
The bottom panel set was computed using the FFT-FD-exp algorithm outlined below, up to time $t=0.2$ with $5000$ time steps,
and with $M=2^7$ modes in the $x$-direction and $M=2^7$ nodes in the $y$-direction.
The plots match the solution form~\eqref{eq:solitonsolution}, though we notice some boundary
interference in the case of the implementation of the FFT-FD-exp algorithm in the bottom set of plots.
\end{example}

Second, we now outline the two exponential split-step pseudo-spectral algorithms we have used.
For convenience, we set $A_{\mathrm{D}}$ to be the linear KP operator,
\begin{equation*}
A_{\mathrm{D}}\coloneqq\pa_x^3+3\pa_x^{-1}\pa_y^2.
\end{equation*}
The basic exponential split-step method we use to integrate the non-commutative KP equation is as follows:
\begin{align*}
v_n&=\exp\bigl(\Delta t\mathcal F(A_{\mathrm{D}})\bigr)u_n,\\
u_{n+1}&=v_n-\Delta t\mathcal F\bigl(\mathrm{N}(v_n)\bigr),
\end{align*}
where $\mathrm{N}(v)=6\pa_x({\mathcal F}^{-1}v)^2+[\mathrm{D}{\mathcal F}^{-1}v,{\mathcal F}^{-1}v]$.
Herein, $\mathcal F$ represents the Fourier transformation,
$u_n$ is the Fourier transform of the approximate solution to the non-commutative KP equation at time $t_n$, $n\in\mathbb Z$,
and $\mathcal F(A_{\mathrm{D}})$ represents the corresponding Fourier transformed version of $A_{\mathrm{D}}$. 
In terms of the spatial discretisation and Fourier transformation, we use two approaches:\smallskip

(1) \emph{FFT2-exp algorithm}: We consider the two-dimensional Fourier transform with respect to $x$ and $y$.
If $2\pi\mathrm{i}k_x/L_x$ and $2\pi\mathrm{i}k_y/L_y$ are the wavenumbers respectively in the $x$- and $y$-directions,
then in this case,
\begin{equation*}
\mathcal F\bigl(A_{\mathrm{D}}\bigr)(k_x,k_y)=(2\pi\mathrm{i}k_x/L_x)^3+3\frac{(2\pi\mathrm{i}k_y/L_y)^2}{(2\pi\mathrm{i}k_x/L_x)+2\pi\delta}.
\end{equation*}
Note, we have approximated the Fourier transform of $\pa_x^{-1}$ by $1/\bigl(2\pi\mathrm{i}k_x/L_x+2\pi\delta\bigr)$, where $\delta=2^{-52}$;
see Klein and Roidot~\cite[p.~3341]{KleinRoidot}. We augment this approach with the `window method' employed by Kao and Kodama~\cite{KaoKodama},
as outlined below.
\smallskip

(2) \emph{FFT-FD-exp algorithm}: We consider the Fourier transform with respect to $x$ only and approximate derivatives in the $y$-direction,
present in $A_{\mathrm{D}}$ and $\mathrm{D}$, by central differences. We assume zero Neumann boundary conditions at the $y=\pm L_y/2$ boundaries.
In this case, the form for $\mathcal F(A_{\mathrm{D}})$ shown above is modified so that a finite central difference replaces the derivative $\pa_y^2$ 
present in $A_{\mathrm{D}}$. We employed this algorithm to compute the soliton solution in Figure~\ref{fig:soliton} for the commutative KP equation,
as a comparison to the GLM method. 
\smallskip

Note that for both of these algorithms, the initial data is generated by numerically solving the GLM equation for the given scattering data $p$ at time $t=0$,
as outlined above. This is the case in all our example cases herein, i.e.\/ in Examples~\ref{ex:soliton}--\ref{ex:threesolitoncase}.
We use the FFT2-exp algorithm to compute approximate solutions to the non-commutative KP equation shown in Figure~\ref{fig:twosolitoncase}
and all subsequent figures. We use pseudo-spectral algorithms due to their simplicity and efficiency, see for example, Klein and Saut~\cite{KleinSaut}
and Grava, Klein and Pitton~\cite{GKP}. However the solutions we wish to compute are typically non-periodic. In the case of the soliton solution
to the commutative KP equation in Example~\ref{ex:soliton} and Figure~\ref{fig:soliton}, the solution is periodic in the $x$-direction in the sense that,
for finite times and domains, the far field limit of the solution in the $x$-direction is zero. This is not true in the $y$-direction, which motivated our
implementation of the FFT-FD-exp algorithm using finite differences in the $y$-direction and zero Neumann boundary conditions at $y=\pm L_y/2$.
However, this boundary assumption is a crude artificial approximation, and the (dispersive) finite-speed propapation of boundary disturbances
into the interior of the domain can be observed in the bottom panel set in Figure~\ref{fig:soliton}. In order to better counteract this problem,
we use the FFT2-exp algorithm combined with the `window method' employed by Kao and Kodama~\cite{KaoKodama}, which in particular attempts to address  
the problem of utlising pseudo-spectal algorithms in the context of non-periodic boundary conditions. Briefly, we implement this method here
as follows. We set $W=W(y)$ to be the window function with the generalised Gaussian form,
\begin{equation*}
W(y)\coloneqq\exp\bigl(-a_r|y/L_y|^r\bigr),
\end{equation*}
where $a_r\coloneqq(1.111)^r\log_{\mathrm{e}}10$ and $r\coloneqq38$.
Suppose $\overline{g}=\overline{g}(x,y;t)$ is the solution to the non-commutative KP equation, not in potential form,
in other words: $\overline{g}_t=\overline{g}_{xxx}+3\mathrm{D}\overline{g}_y-6\pa_x\overline{g}^2-6[\mathrm{D}\overline{g},\overline{g}]$,
with initial data $f=f(x,y)$ so $\overline{g}(x,y;0)=f(x,y)$.
Hence $\overline{g}=\pa_xg$ when $g$ is the solution to the potential form of the non-commutative KP equation~\eqref{eq:KPsingform}.
Then we implement the `window method' here in the following form. 
We decompose the initial data so that $f=Wf+(1-W)f$ and set $\overline{g}=\hat{g}+(1-W)f$. We simulate $\hat{g}$ with initial data $\hat{g}(0)=Wf$.
We assume $L_x$ is sufficiently large so that the boundary conditions in that direction are periodic in the sense that they are zero during the time of our computation.
The crucial assumption we make in the decomposition is that we can choose $L_y$ sufficiently large so that any minor disturbances made there, due to
our rendering of the boundary conditions we have made, do not have time to propagate to the main interior region during the time of our computation. 
As we see, we only partially manage to accomplish this, however the resulting simulations are sufficient for the purposes of our intended demonstration.
For completeness, we note that $\hat{g}$ satisfies the equation,
\begin{align*}
\hat{g}_t=&\;\hat{g}_{xxx}+3\mathrm{D}\hat{g}_y-6\pa_x\hat{g}^2-6[\mathrm{D}\hat{g},\hat{g}]+6W^\prime[\pa_x^{-1}f,\hat{g}]\\
&\;+6(1-W)\bigl(W\pa_x f^2+W[\mathrm{D}f,f]-\pa_x(\hat{g}f+f\hat{g})\bigr)\\
&\;+6(1-W)\bigl(W^\prime[\pa_x^{-1}f,f]-[\mathrm{D}f,\hat{g}]-[\mathrm{D}\hat{g},f]\bigr)\\
&\;-3\pa_x^{-1}\bigl(W^{\prime\prime}f+2W^\prime f_y\bigr).
\end{align*}
Kao and Kodama~\cite{KaoKodama} additionally rescale the $x$ and $y$ domains, which we do not invoke here.
We apply the FFT2-exp algorithm to the evolution equation for $\hat{g}$ above, noting that the nonlinear term $\mathrm{N}(v)$
now consists of all of the nonlinear and non-homogeneous terms on the right shown above---indeed all the terms on the right apart from the first two.
For the reasons outlined in Klein and Roidot~\cite{KleinRoidot}, and Klein and Saut~\cite{KleinSaut}, the solutions $\overline{g}$ we consider herein 
satisfy the constraint $\int_{\mathbb R} \overline{g}_{yy}\,\rd x=0$---or in practice over the $x$-domain $[-L_x/2,L_x/2]$---for computation times including $t=0$.
This is because for the data, $\overline{g}=\pa_x g$, and we assume zero boundary boundary data in the $x$ far field uniformly in $y$. 
When this constraint is not satisfied, we observe the infinite ``trenches'' mentioned in Klein and Roidot~\cite{KleinRoidot}. 

\begin{figure*}
  \begin{center}
    \includegraphics[width=8cm,height=7cm]{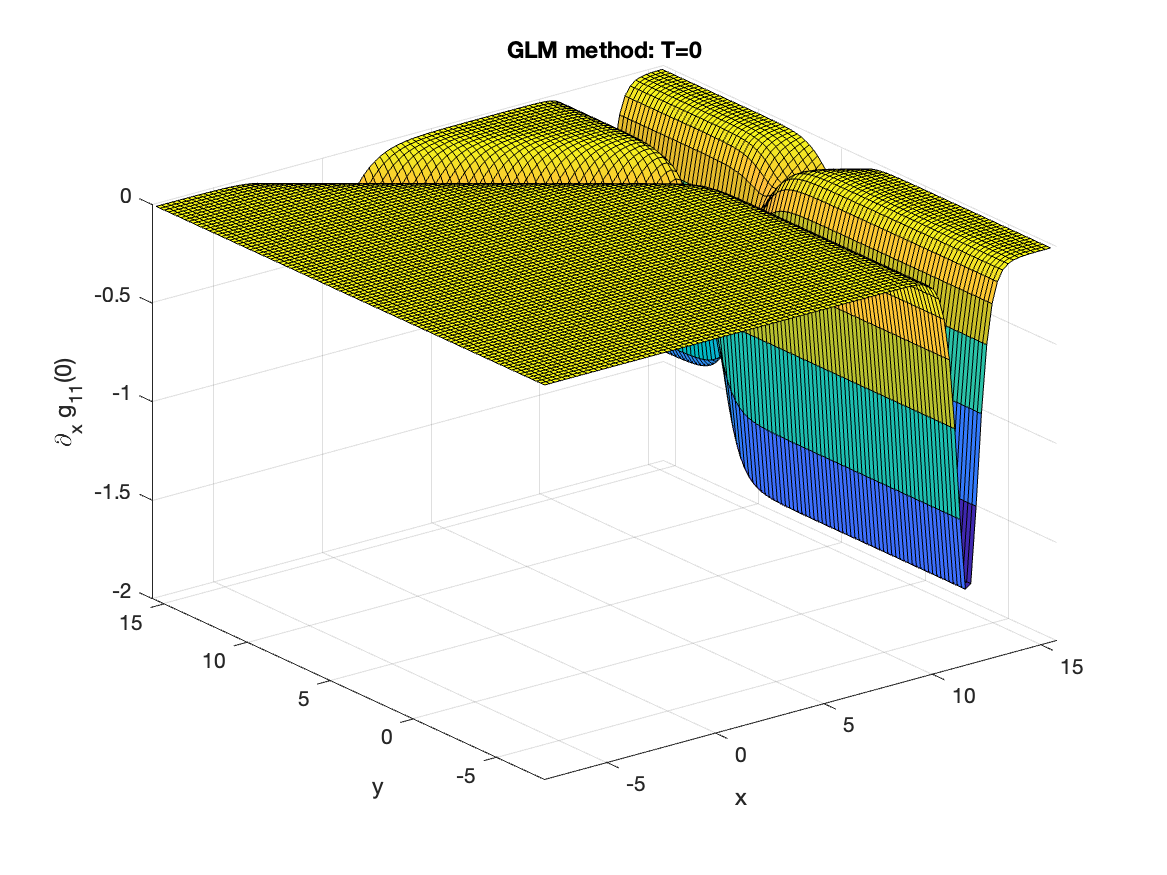}
    \includegraphics[width=8cm,height=7cm]{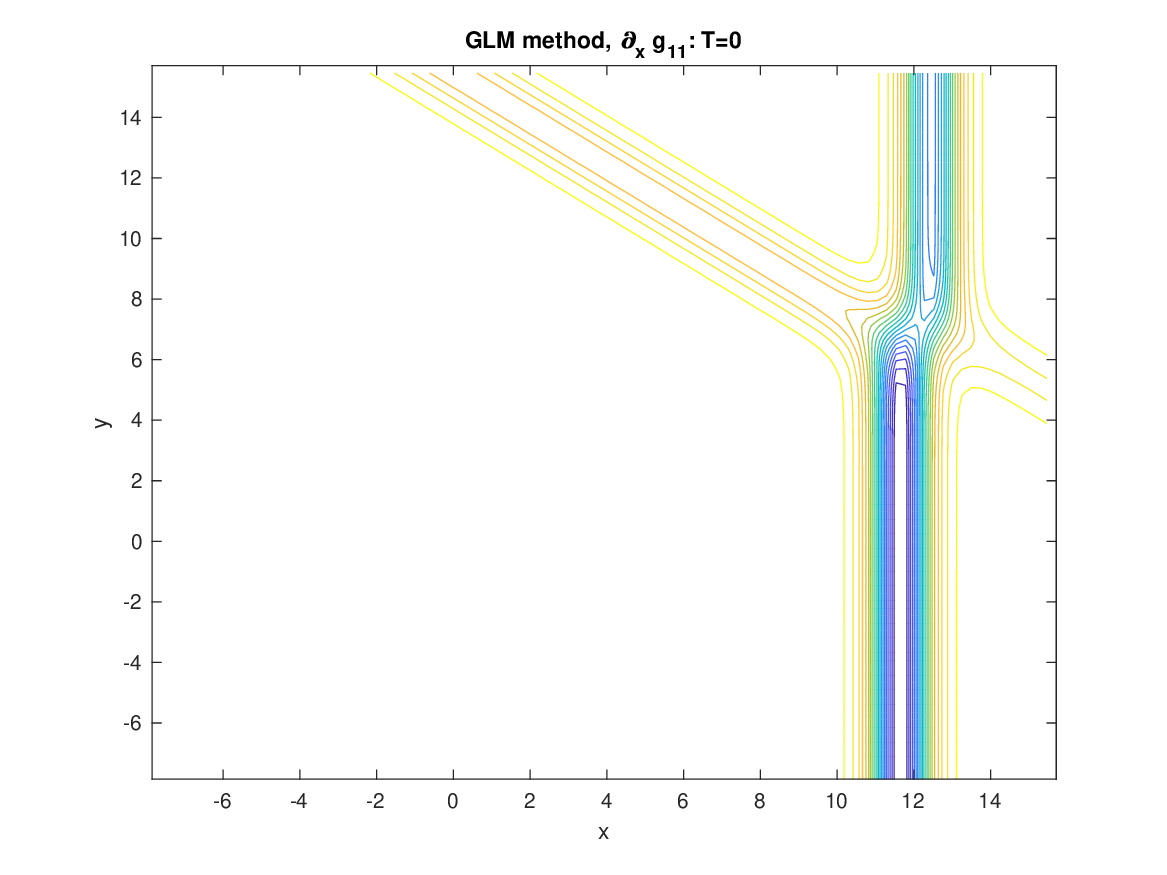}\\
    \includegraphics[width=8cm,height=7cm]{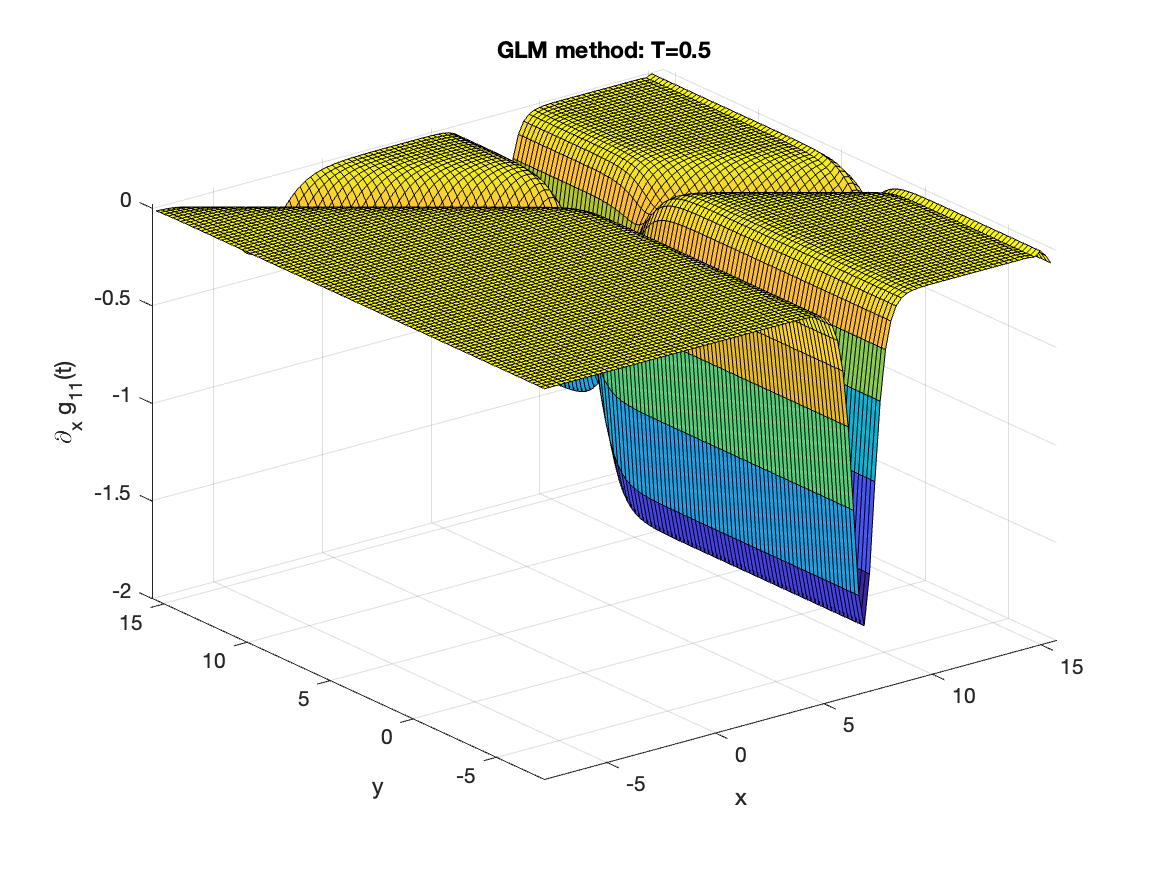}
    \includegraphics[width=8cm,height=7cm]{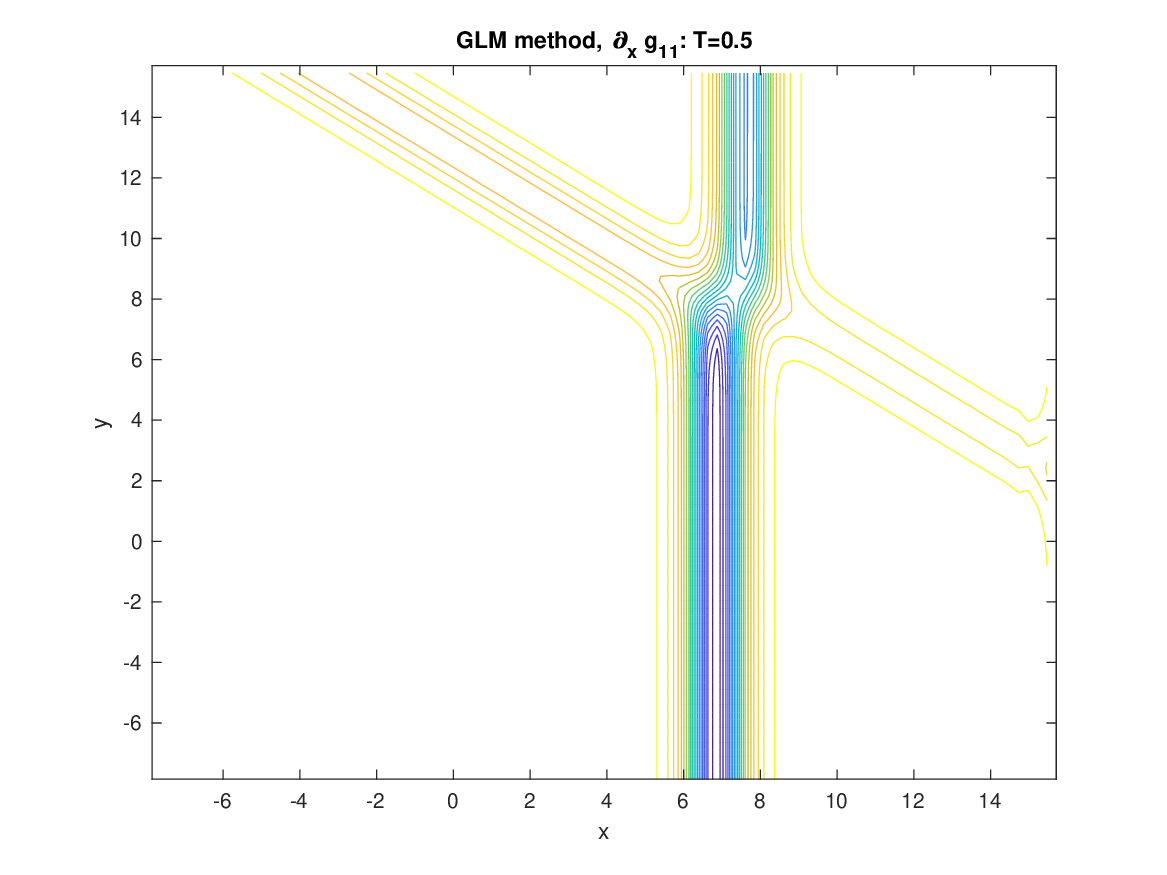}\\
    \includegraphics[width=8cm,height=7cm]{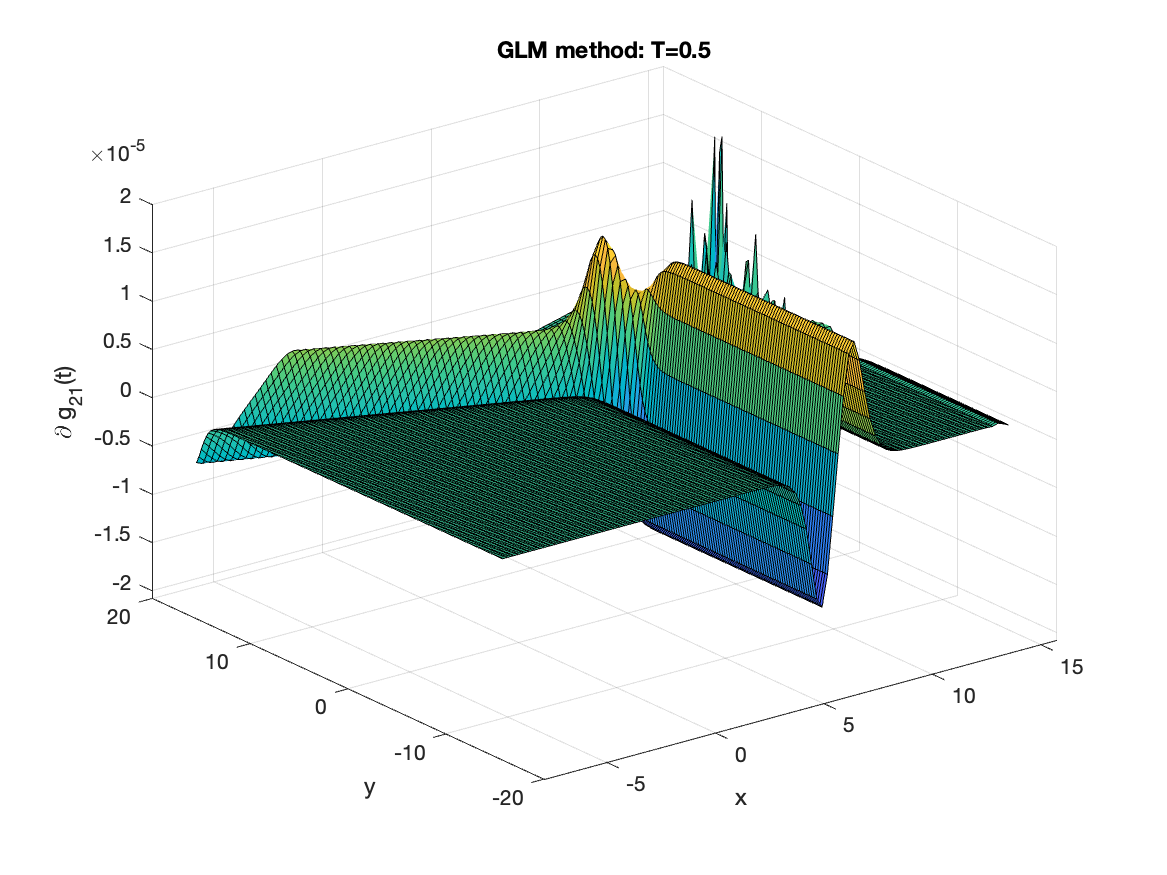}
    \includegraphics[width=8cm,height=7cm]{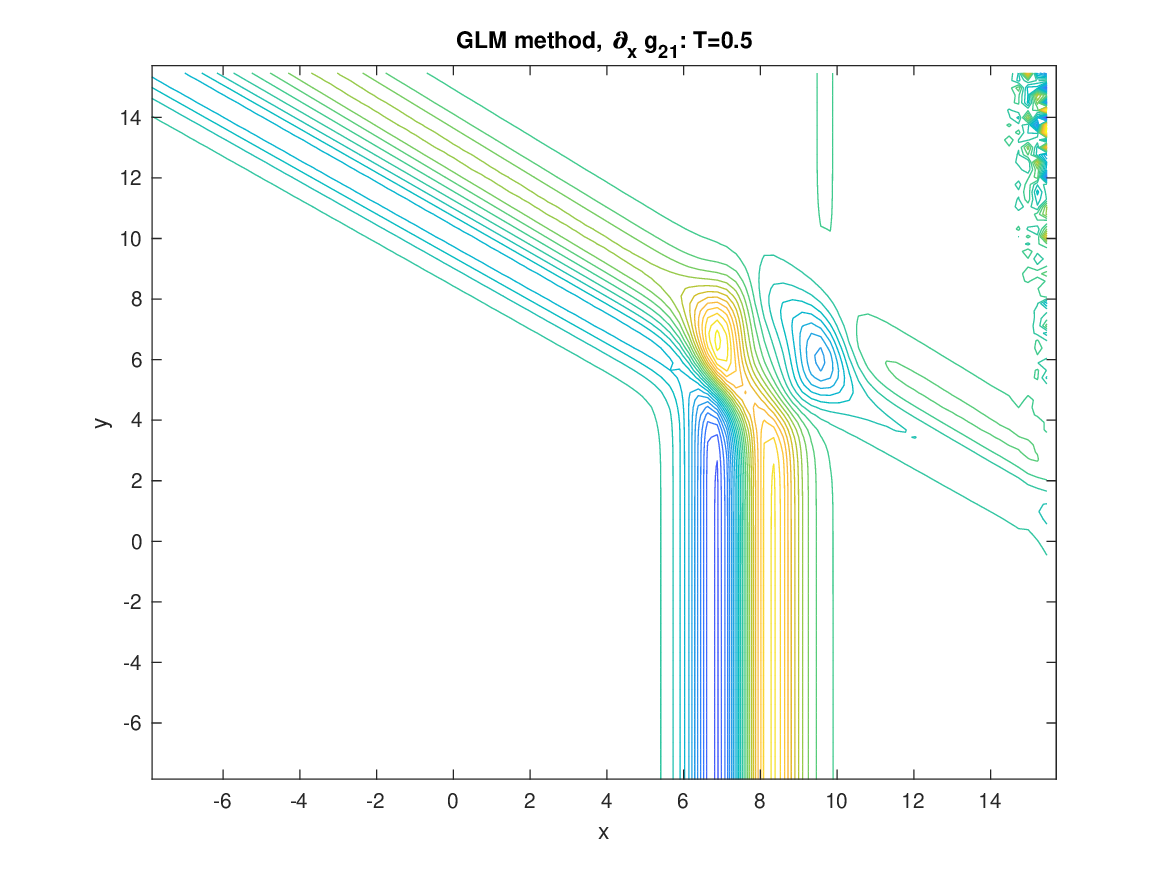}
  \end{center}
  \caption{We plot the component $\pa_xg_{11}$ for the two soliton-like interaction corresponding to Example~\ref{ex:twosolitoncase} at $t=0$ in the top panel set.
    In the middle and bottom panel sets, we plot the solution components $\pa_xg_{11}$ and $\pa_xg_{21}$,
    to the non-commutative KP equation computed, at time $t=0.5$. The right-hand panels give the corresponding contour plots.
    The solutions at $t=0$ and $t=0.5$ are computed by numerically solving the GLM equation~\eqref{eq:GLMexplicit}.
    Some numerical error can be observed in the top right of the $\pa_xg_{21}$ plots.}
\label{fig:twosolitoncase}
\end{figure*}

\begin{example}[Two soliton-like case]\label{ex:twosolitoncase}
In Figure~\ref{fig:twosolitoncase} we consider the evolution of $\overline{g}$ from the data $f$, resembling the interaction of two solitary waves.
The domain size is $L_x=L_y=10\pi$ and the number of modes in both the $x$- and $y$-directions is $M=2^7$. With $\epsilon\coloneqq10^{-6}$, we suppose,
\begin{equation*}
  A_1\coloneqq\begin{pmatrix} 1.55 & 1.0\,\epsilon \\ 0.9\,\epsilon & 1.45\end{pmatrix}
  \quad\text{and}\quad
  A_2\coloneqq\begin{pmatrix} 1.3 & 0.9\,\epsilon \\ 1.2\,\epsilon & 1.4\end{pmatrix}
\end{equation*}
with $B_1=A_1$ and $B_2=O$, the zero $2\times 2$ matrix. We also set $\Lambda_j\coloneqq A_j^2-B_j^2$ and $\Omega_j\coloneqq4(A_j^3+B_j^3)$ for $j=1,2$.
If we substiute $(A_1,B_1,\Lambda_1,\Omega_1)$ for $(A,B,\Lambda,\Omega)$ in the scattering data \eqref{eq:scatteringdata} to generate $p_1$,
and numerically solve the GLM equation when $t=0$, as described above in~\ref{sec:solitons} and Example~\ref{ex:soliton},
we generate the soliton solution impinging the bottom boundary in the top panels in Figure~\ref{fig:twosolitoncase}. 
Similarly, suppose substituting $(A_2,B_2,\Lambda_2,\Omega_2)$ instead generates $p_2$ in \eqref{eq:scatteringdata}.
Solving the GLM equation numerically generates the soliton solution impinging the left section of the top boundary in the top panels in the Figure.
Note, the matrices $A_1$ and $B_1$ trivially commute, as do $A_2$ and $B_2$. However $A_1$ and $A_2$ do not commute and thus neither do $p_1$ and $p_2$.
If we then solve the GLM equation numerically assuming the scattering data is $p=p_1+p_2$, we generate the full solution interaction shown in the top panel set in the Figure. 
Note we only display the $x$- and $y$-domains in the range $[-L_x/4,L_x/2]$ and $[-L_y/4,L_y/2]$, and the $x$-coordinate in $p$ is shifted so the interaction occurs in this region.
This initial solution represents a two soliton interaction, in this non-commutative context.
Subsequently we numerically solve the GLM equation at the single time $t=0.5$.
This generates the solution forms for $\pa_x g_{11}$ and $\pa_x g_{21}$ shown in Figure~\ref{fig:twosolitoncase}---we only show these two components for brevity.
\end{example}

\begin{example}[Three soliton-like case]\label{ex:threesolitoncase}
  In Figures~\ref{fig:g11} and \ref{fig:allgs} we numerically compute the interaction of three solitary-like waves in the non-commutative setting.
  The domain size is $L_x=L_y=10\pi$ and the number of nodes/modes in both the $x$- and $y$-directions is $M=2^8$.
  Suppose $A_1$ and $A_2$ are the same non-commuting matrices as in Example~\ref{ex:twosolitoncase},
  and that $A_3=O$, $B_1=A_1$, $B_2=O$ and $B_3=A_2$. As in Example~\ref{ex:twosolitoncase}, we set $\Lambda_j$ and $\Omega_j$
  to be the appropriate quantities shown therein, though now for $j=1,2,3$. We then generate the scattering data $p_j$, for $j=1,2,3$,
  by respectively substituting the coefficient sets $(A_j,B_j,\Lambda_j,\Omega_j)$ for $(A,B,\Lambda,\Omega)$ in the scattering data \eqref{eq:scatteringdata}.
  For each $j=1,2,3$, the pairs $A_j$ and $B_j$ trivially commute. However $A_1$ and $A_2$ do not commute so that while $p_2$ and $p_3$ commute,
  $p_1$ does not commute with either $p_2$ or $p_3$. The signature of the individual solitons generated by $p_1$, $p_2$ and $p_3$ can
  be observed in the top panel set in Figure~\ref{fig:g11}. The solution soliton-like component impinging middle to right section of the bottom boundary is generated by $p_1$,
  the soliton-like component impinging the left section of the top boundary corresponds to $p_2$, while the component impinging the left section
  of the bottom boundary corresponds to $p_3$.
  The top set of figures in Figure~\ref{fig:g11} is generated by solving the GLM equation numerically assuming the scattering data is $p=p_1+p_2+p_3$.
  The time therein is set to be $t=0$.
  We shift the $x$-coordinate in $p$ to ensure the full non-commutative solution interaction, shown in the top panel set in the Figure, occurs 
  in the displayed region $(x,y)\in[-L_x/4,L_x/2]\times[-L_y/2,L_y/2]$. For brevity we only show the $\pa_x g_{11}$ component. 
  The middle panel set in Figure~\ref{fig:g11} show the solution $\pa_x g_{11}$, corresponding to the mixed scattering data $p$,
  computed using the GLM method at time $t=0.5$. The bottom panel set in the figure shows the solution $\overline{g}_{11}=\pa_x g_{11}$
  computed using the FFT2-exp method up to the time $t=0.5$ using $5000$ time steps. The right panels in the Figure show the contour
  plots corresponding to the left panel plots. Figure~\ref{fig:allgs} shows, respectively in the top, middle and bottom panel sets,
  the other three components $\pa_x g_{12}$, $\pa_x g_{21}$ and $\pa_x g_{22}$ computed at $t=0.5$. The left panels show these 
  components in the solution computed using the GLM method, while the right panels show the same three components of the solution computed using the FFT2-exp method.   
  In the bottom panel set in Figure~\ref{fig:g11} and all three right panels in Figure~\ref{fig:allgs}, the limits of the FFT2-exp method combined with
  the window method can be observed. Perturbations in the numerical solution due to the non-periodic boundary conditions eventually begin to impact
  the interior of the computational domain. Indeed, the numerical solutions are smoother for smaller times.
  However the numerical simulations demonstrate that solutions to the non-commutative KP equation~\eqref{eq:KPsingform} can be relatively easily
  computed by numerically solving the corresponding GLM equation for suitable scattering data. Also see Remark~\ref{rmk:GLMFFT2}.
\end{example}

\begin{figure*}
  \begin{center}
    \includegraphics[width=8cm,height=7cm]{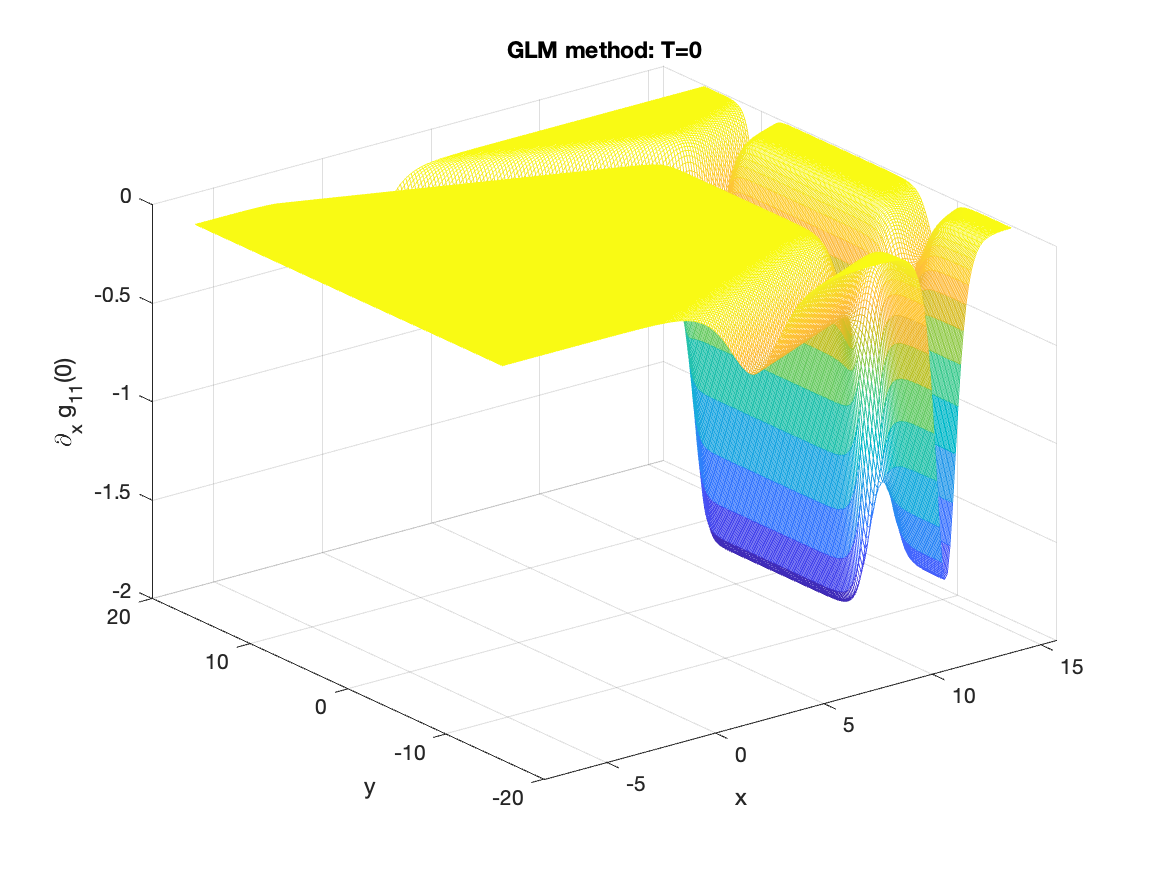}
    \includegraphics[width=8cm,height=7cm]{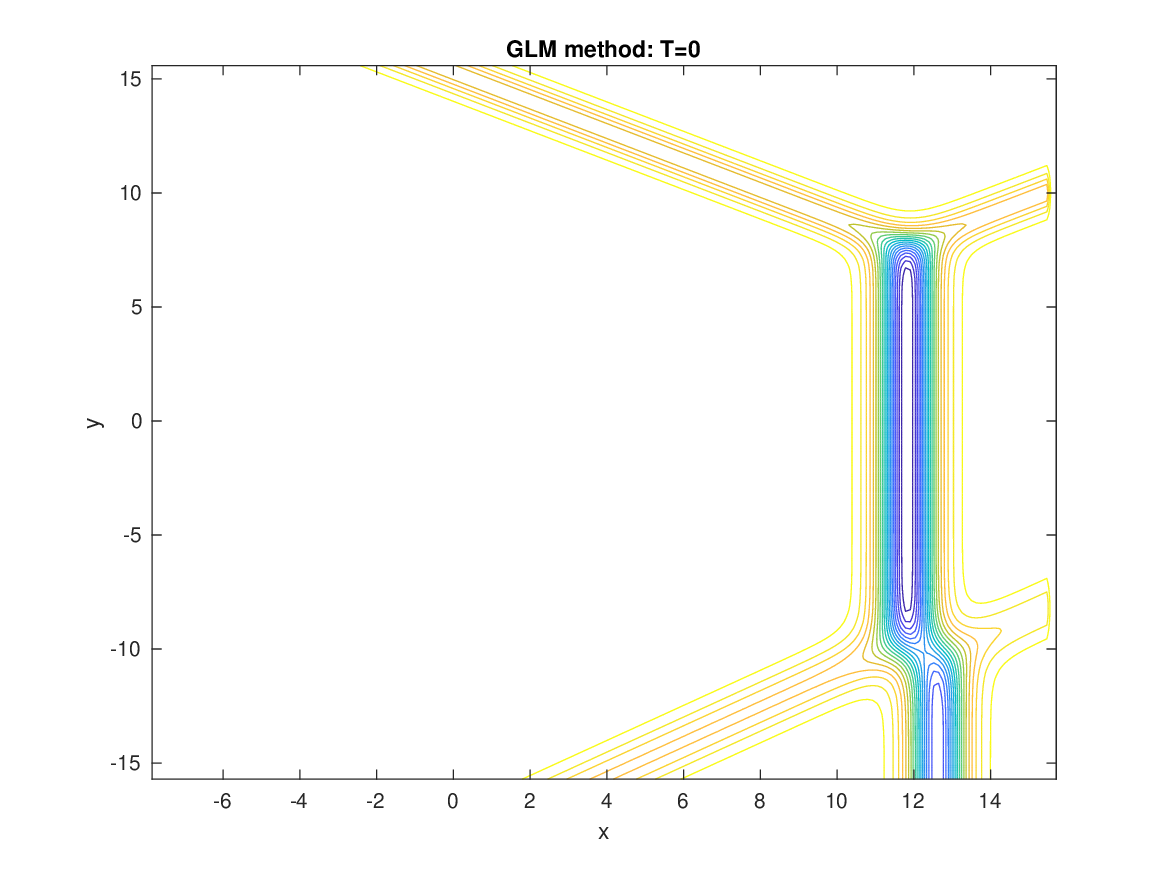}\\
    \includegraphics[width=8cm,height=7cm]{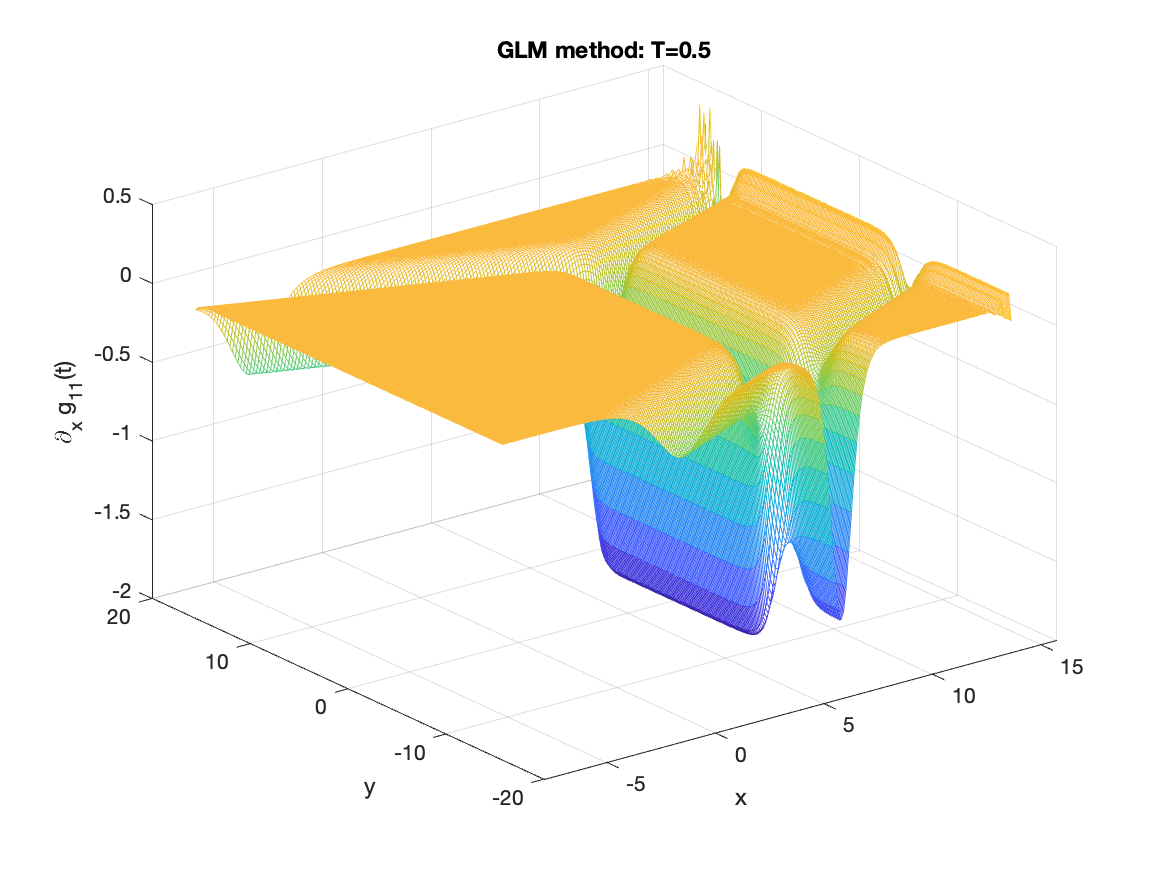}
    \includegraphics[width=8cm,height=7cm]{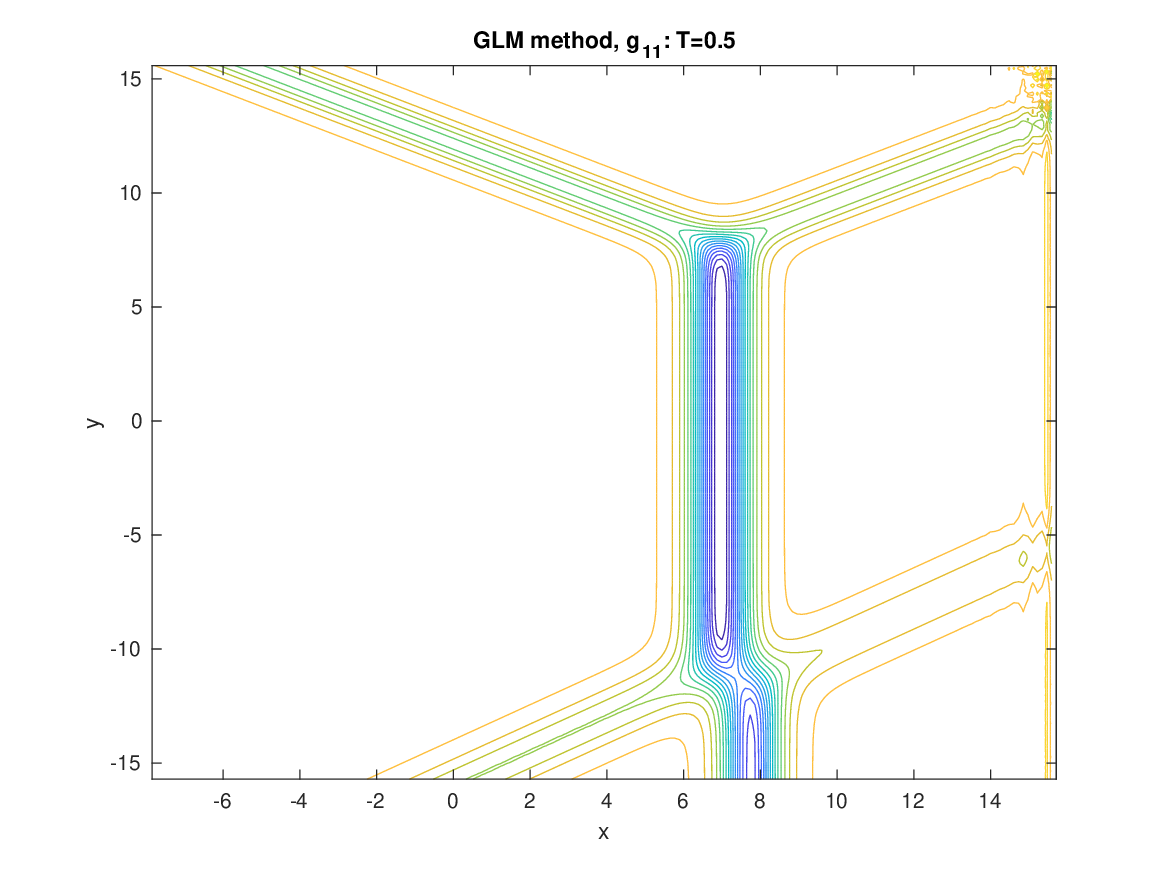}\\
    \includegraphics[width=8cm,height=7cm]{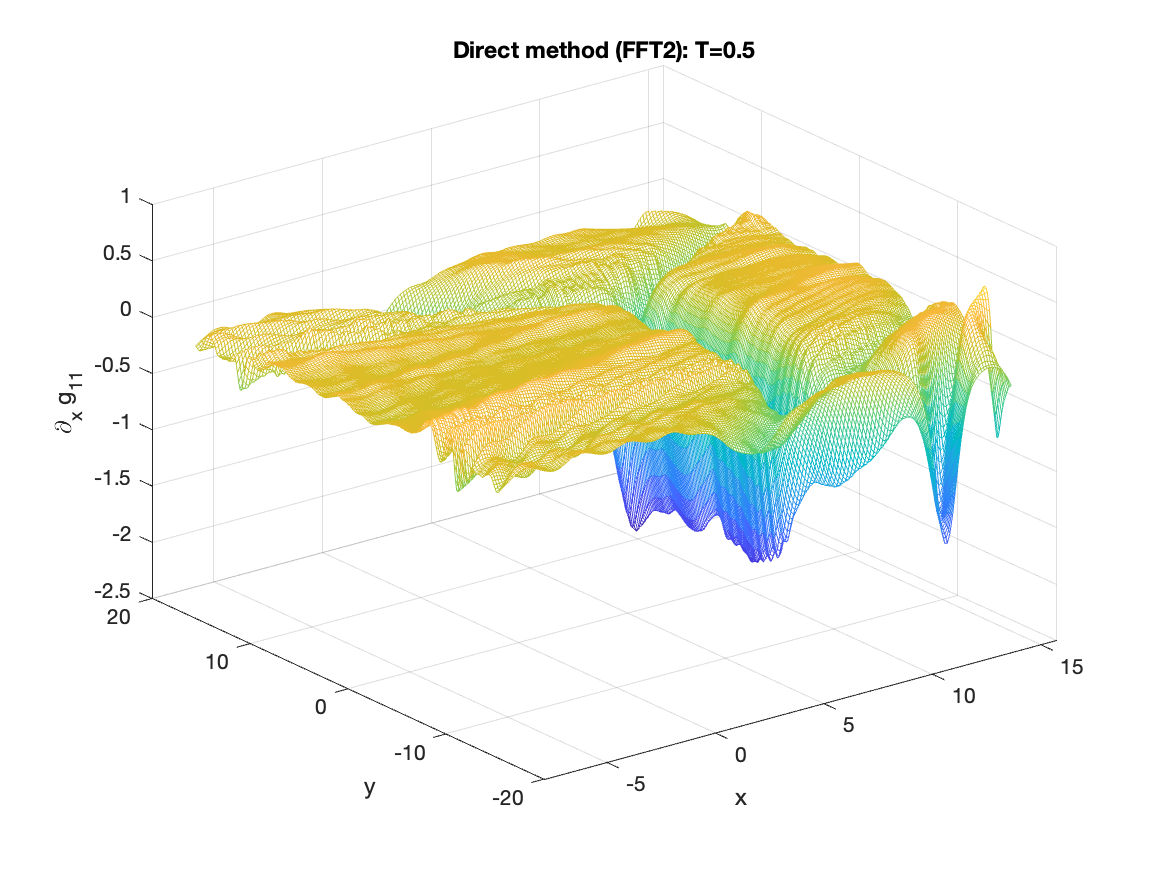}
    \includegraphics[width=8cm,height=7cm]{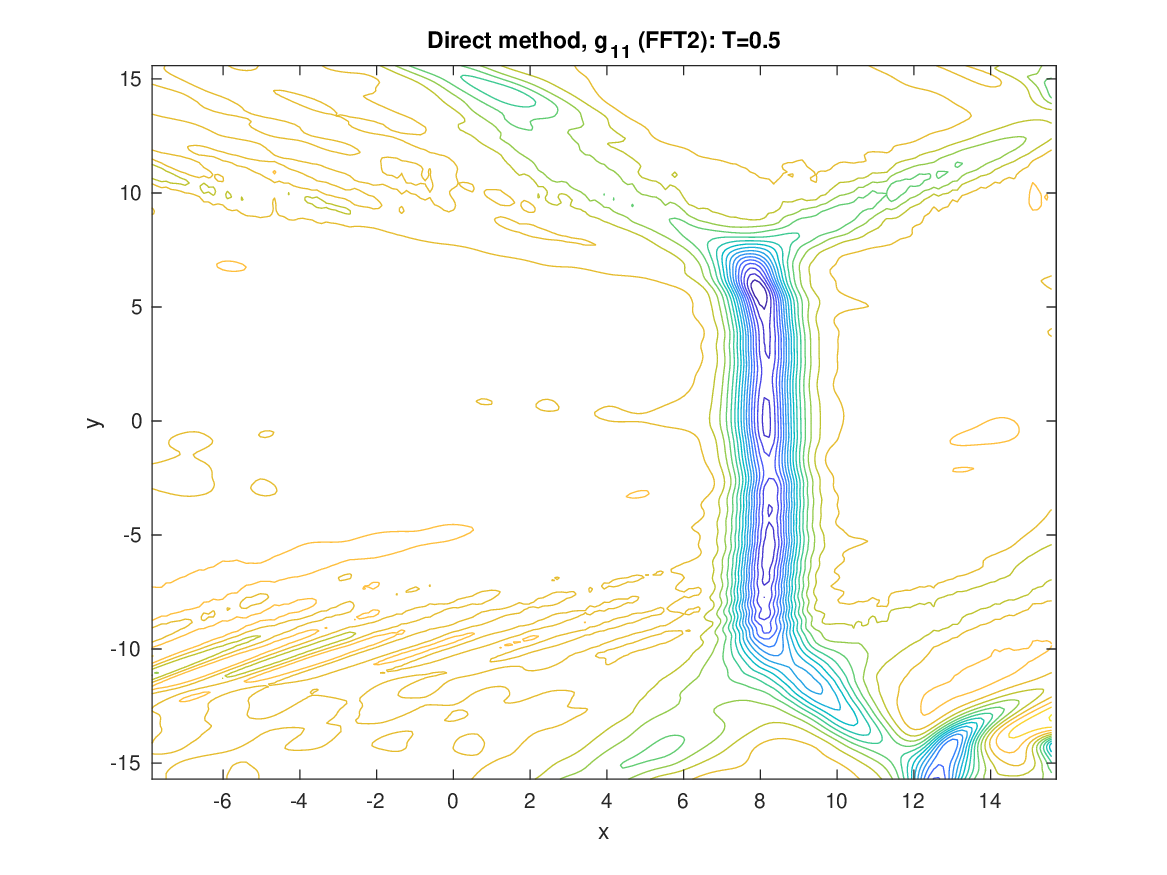}
  \end{center}
  \caption{The top panel set shows the component $\pa_xg_{11}$ for the three soliton-like interaction corresponding to Example~\ref{ex:threesolitoncase} at $t=0$.
    In the middle and bottom panel sets, we plot the solution component $\pa_xg_{11}$ to the non-commutative KP equation at time $t=0.5$, respectively
    computed using the GLM method and the FTT2-exp method. The right-hand panels give the corresponding contour plots.}
\label{fig:g11}
\end{figure*}

\begin{figure*}
  \begin{center}
  \includegraphics[width=8cm,height=7cm]{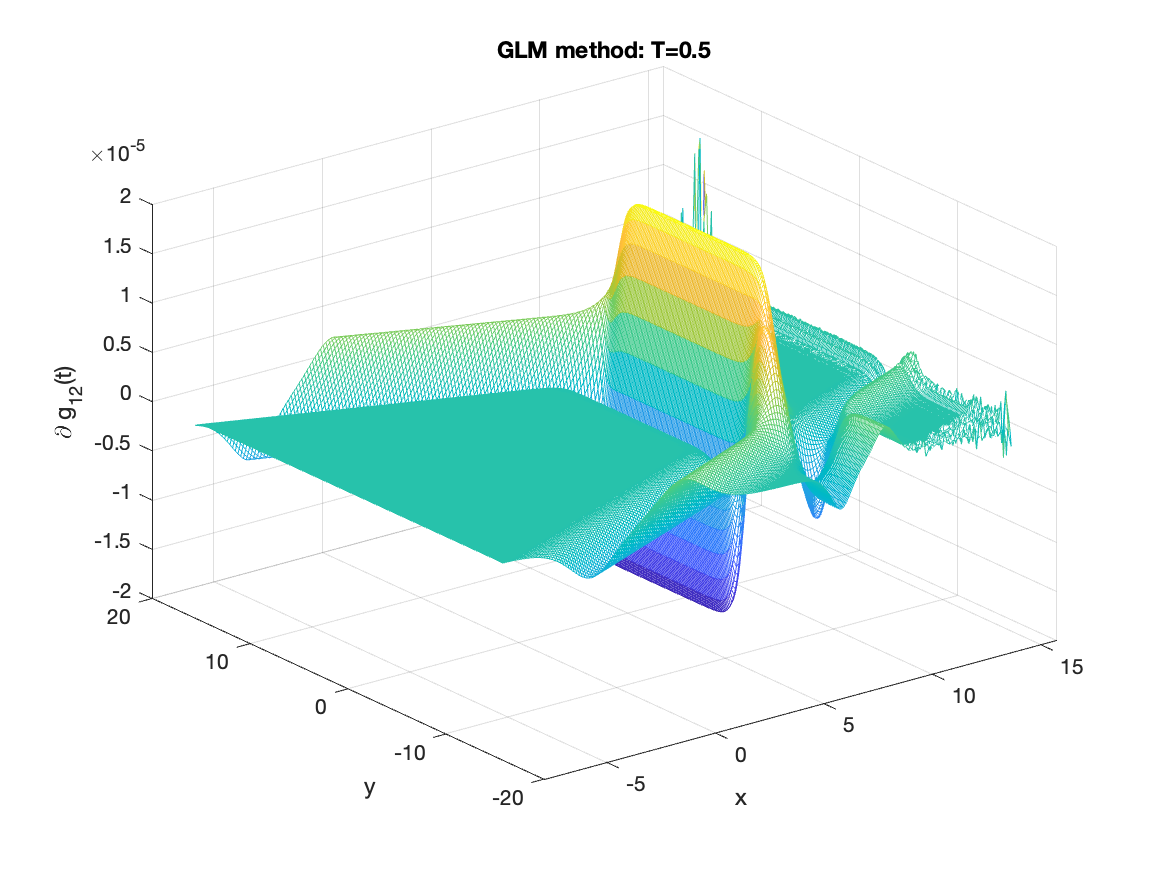}
  \includegraphics[width=8cm,height=7cm]{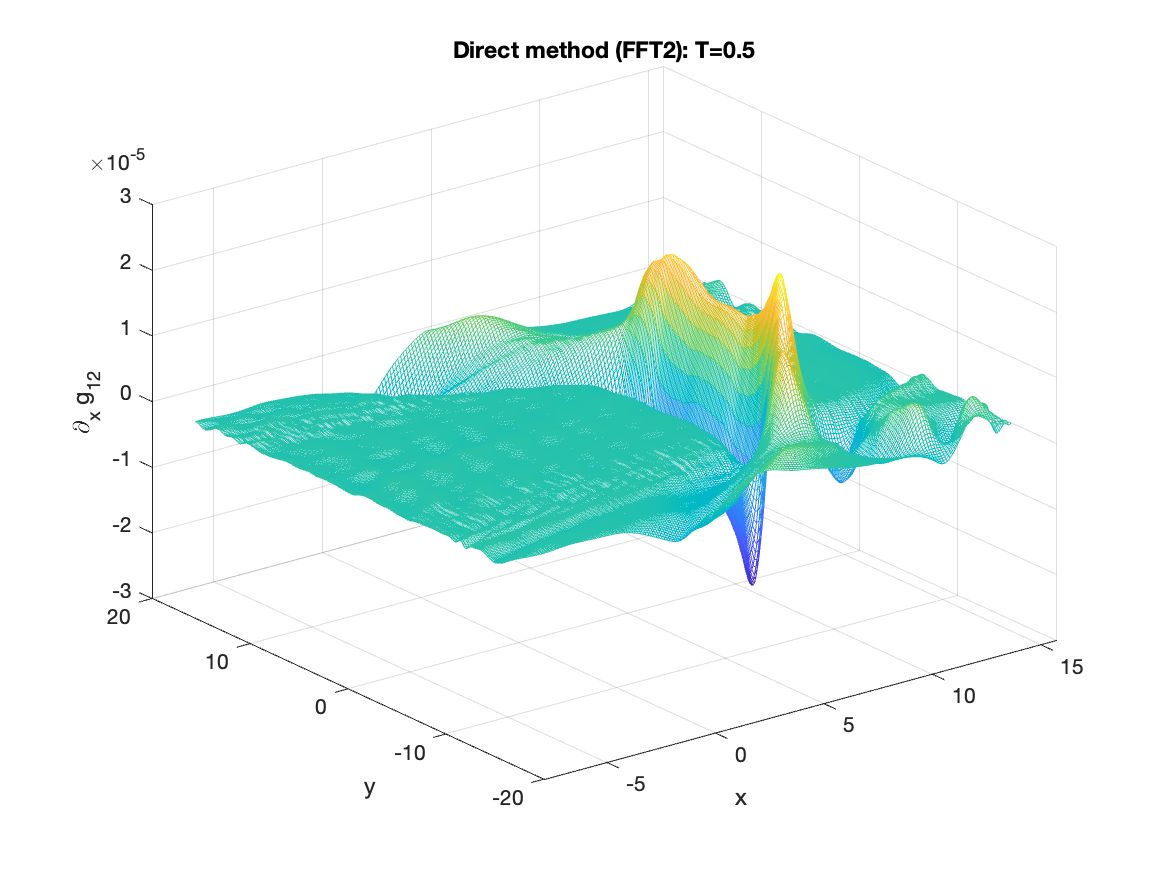}\\
  \includegraphics[width=8cm,height=7cm]{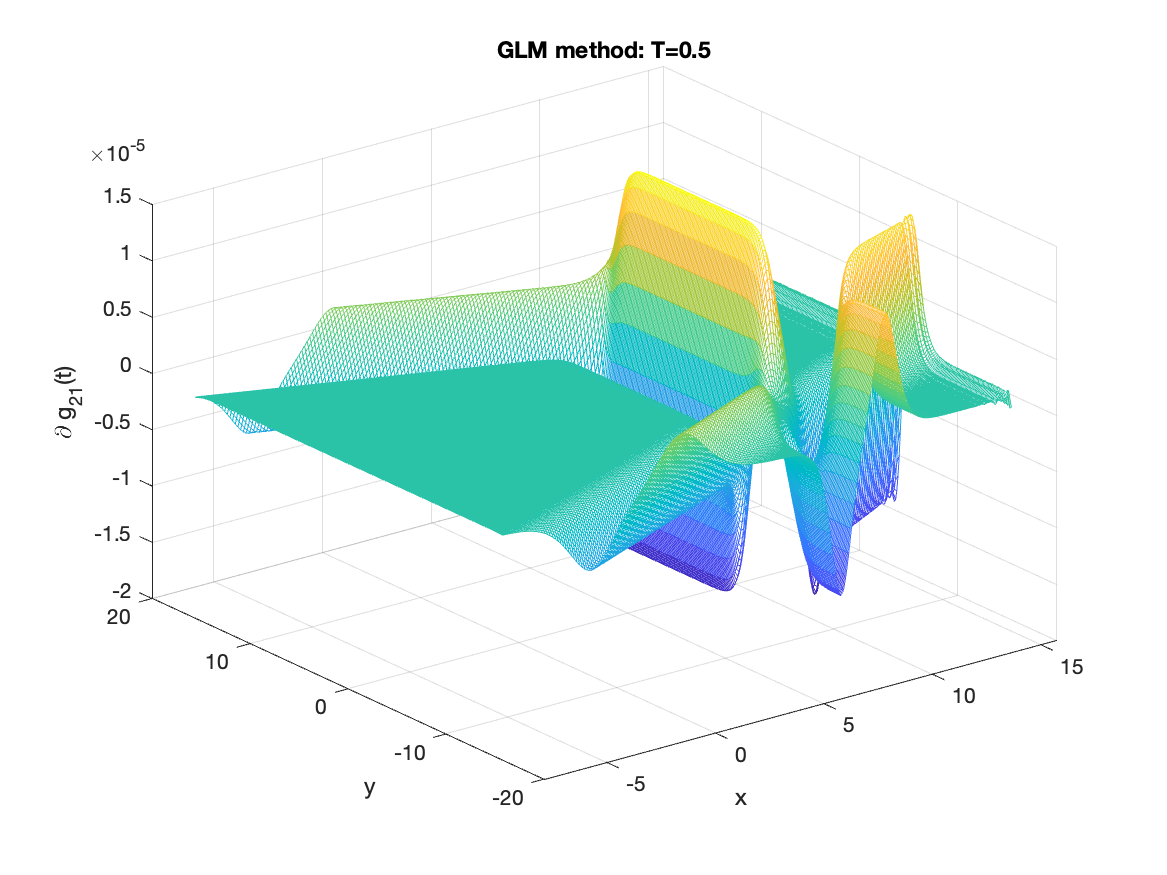}
  \includegraphics[width=8cm,height=7cm]{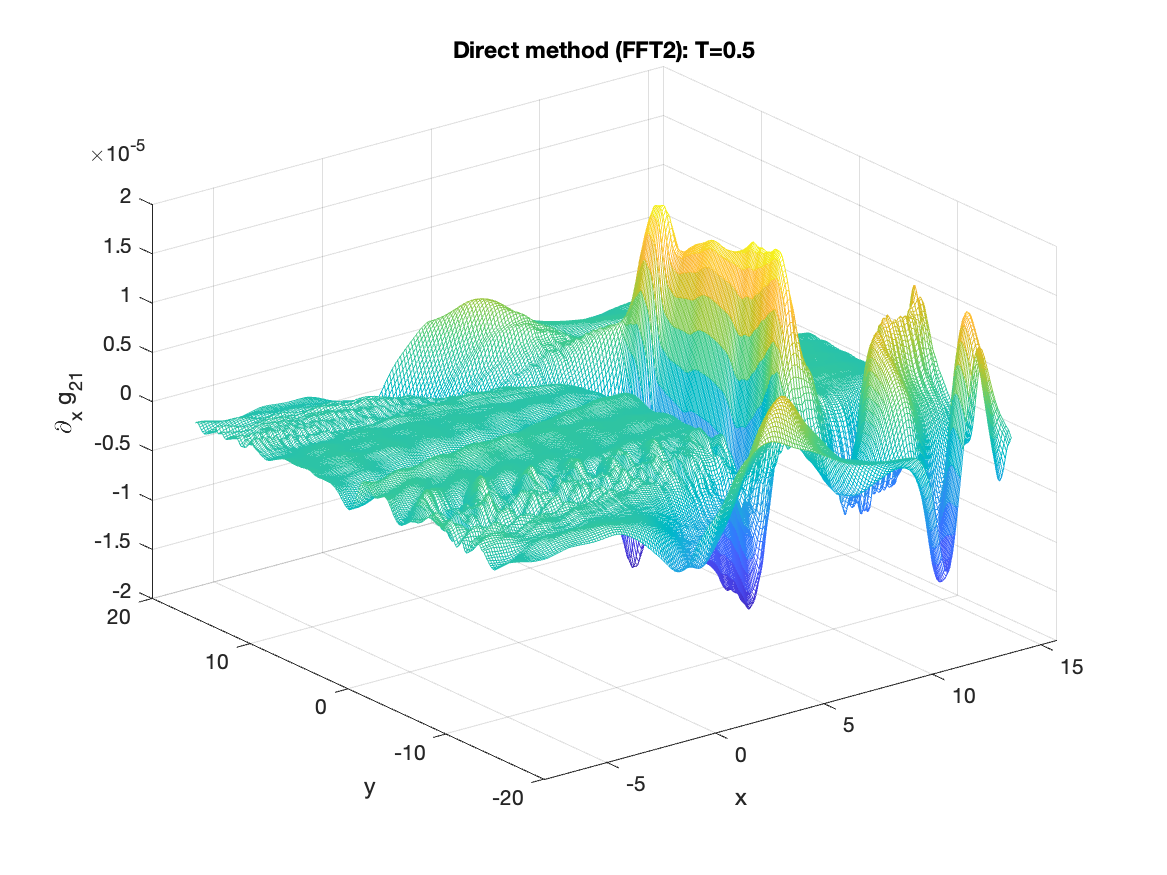}\\
  \includegraphics[width=8cm,height=7cm]{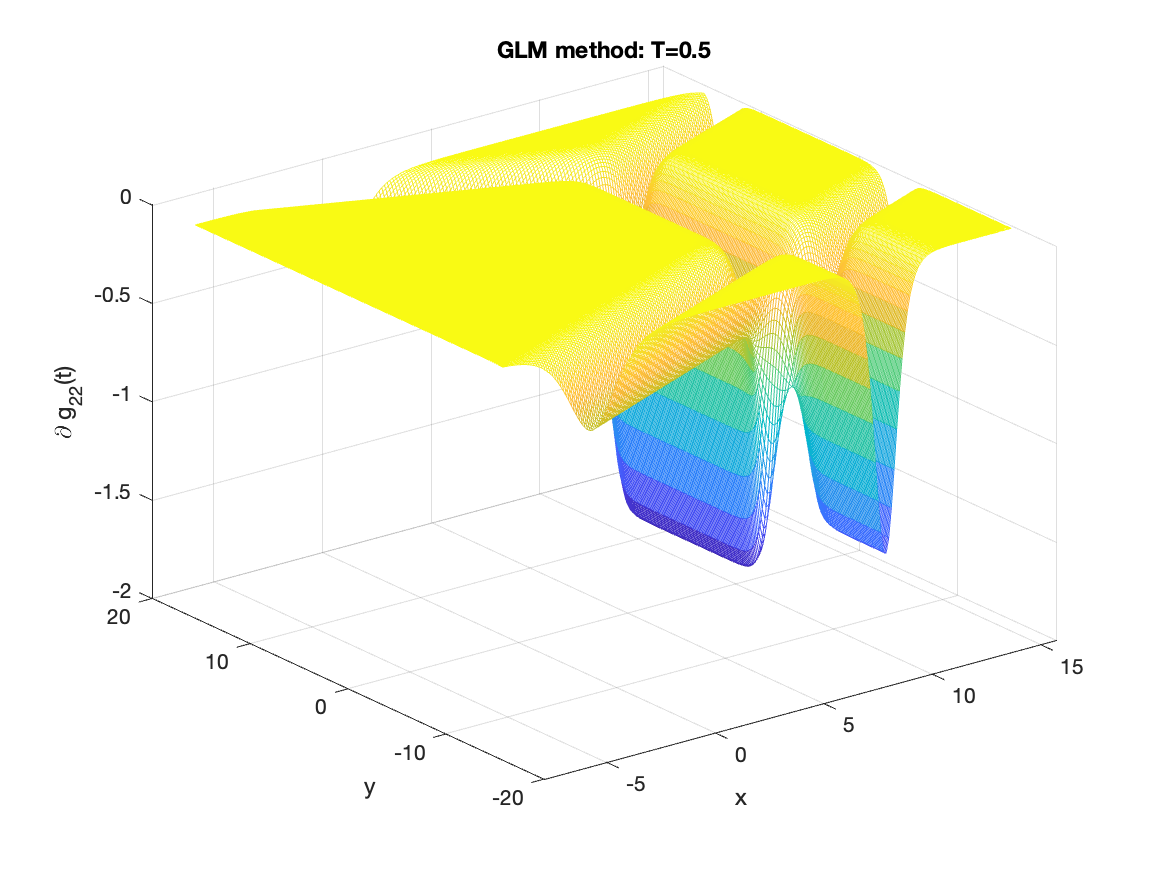}
  \includegraphics[width=8cm,height=7cm]{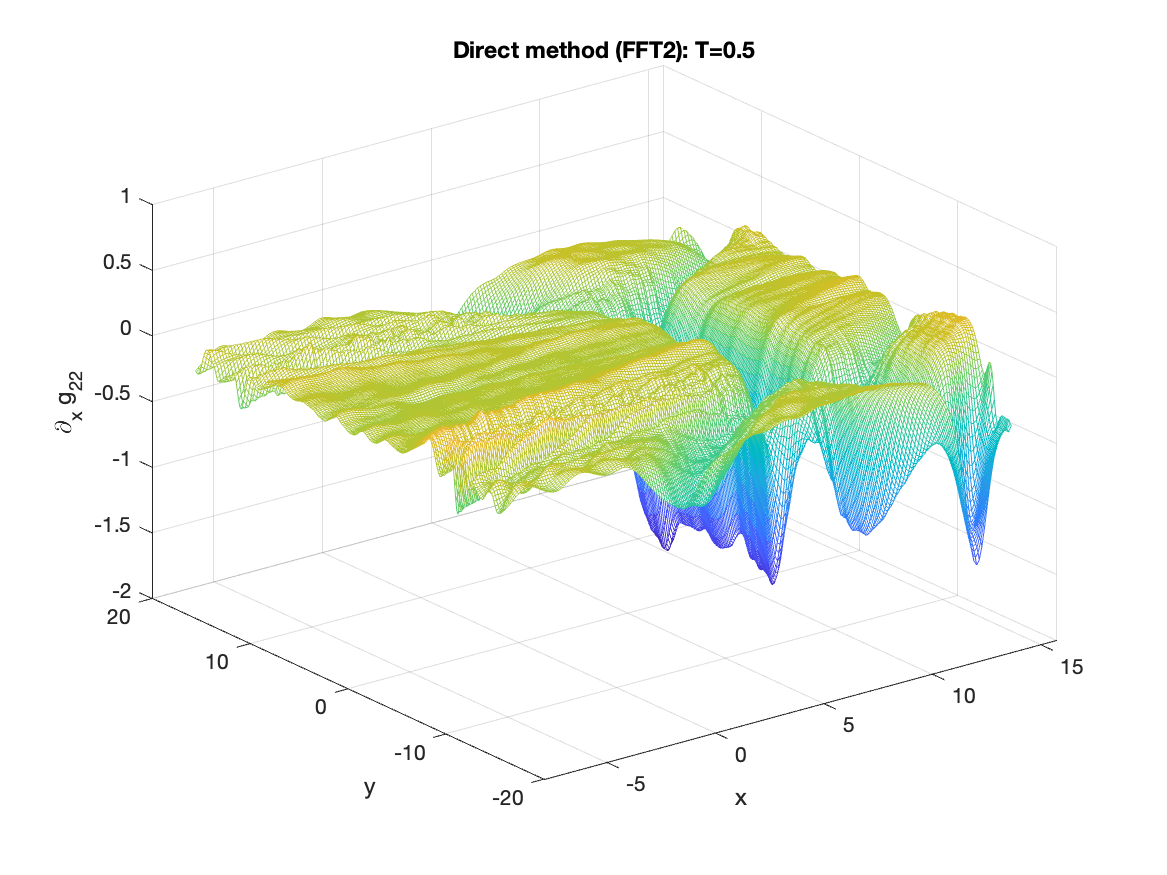}
  \end{center}
  \caption{The top, middle and bottom panel sets respectively show the components $\pa_xg_{12}$, $\pa_xg_{21}$ and $\pa_xg_{22}$
    for the three soliton-like interaction corresponding to Example~\ref{ex:threesolitoncase} at time $t=0.5$.
    The left panels show the components generated when the solution to the non-commutative KP equation is computed using the GLM method,
    while the right panels show the corresponding components generated via the FFT2-exp method.}
\label{fig:allgs}
\end{figure*}

\begin{remark}\label{rmk:GLMFFT2}
  In our computations of approximate solutions to the non-commutative KP equation above, we always assumed the scattering
  data to have the exponential form \eqref{eq:scatteringdata} or to be a linear combination of such forms,
  where the individual scattering data exponential functions do not commute. More generally and in practice,
  we would like to choose an arbitrary matrix-valued scattering data function $p_0=p_0(z,\zeta)$, and generate the scattering data $p=p(z+x,\zeta+x;y,t)$ from $p_0$.
  If $\mathcal F(p_0)=\bigl(\mathcal F(p_0)\bigr)(k_z,k_\zeta)$ represents the two-dimensional Fourier transform of $p_0$, then,
  with $\nu_z\coloneqq 2\pi\mathrm{i}k_z$ and $\nu_\zeta\coloneqq 2\pi\mathrm{i}k_\zeta$,
  \begin{equation*}
      {\mathcal F}^{-1}\Bigl(\mathrm{e}^{t\bigl((\nu_z+\nu_\zeta)x+(\nu_z^2-\nu_\zeta^2)y+4(\nu_z^3+\nu_\zeta^3)\bigr)}\bigl(\mathcal F(p_0)\bigr)(k_z,k_\zeta)\Bigr),
  \end{equation*}
  generates the time-evolved scattering data function $p=p(z+x,\zeta+x;y,t)$ satisfying the constraints \eqref{eq:constraints} and \eqref{eq:secondlinearform}. 
  In numerical computations we take $\nu_z\coloneqq 2\pi\mathrm{i}k_z/L_x$ and $\nu_\zeta\coloneqq 2\pi\mathrm{i}k_\zeta/L_y$.
  However, in practice, since the wavenumbers $\nu_z$ and $\nu_\zeta$ are purely imaginary, the exponential factor above involving `$(\nu_z^2-\nu_\zeta^2)y$'
  can grow super-exponentially, even though $y\in[-L_y/2, L_y/2]$. An assumption about the suitable decay of $\mathcal F(p_0)$ for large $|k_z|$ and $|k_\zeta|$ is required.
\end{remark}

\section{Addendum to proof of Theorem~\ref{thm:mKP}}\label{sec:finishproof2ndThm}
Herein we establish the additional step required to complete the proof of Theorem~\ref{thm:mKP},
namely to show that $h\coloneqq[VP_{[1,0]}V^\dag]$, $g\coloneqq[V]$ and $f\coloneqq-\{VP_{[1,0]}V^\dag\}$ satisfy the evolution equation for $f$
in the lifted mKP equations~\eqref{eq:triplemKPh}--\eqref{eq:triplemKPf}.
We simply enumerate all the terms in~\eqref{eq:triplemKPf}.
It is then straightforward to check they combine in precisely the correct manner to satisfy~\eqref{eq:triplemKPf}.
First, using~\eqref{eq:gsquared}, we observe,
\begin{align*}
\mathrm{D}\pa_y(f+g^2)=&\;\pa_y\{VP_{[2,0]}V\}\\
=&\;\{V(P_{4,0}-2P_{2,2}+P_{0,4})V\}\\
&\;+2\{VP_{[2,0]}VP_{[2,0]}V\}.
\end{align*}
Second, we observe that,
\begin{align*}
  \tfrac13(f_t-f_{xxx})=&\;-\tfrac43\{V(P_{3,0}+P_{0,3})VP_{[1,0]}V^\dag\}\\
                       &\;+\tfrac43\{VP_{[1,0]}V^\dag(P_{3,0}+P_{0,3})V^\dag\}\\
                       &\;-\{V(P_{4,0}-2P_{3,1}+2P_{1,3}-P_{0,4})V^\dag\}\\
                       &\;+2\{VP_{\hat 1}VP_{\hat 1}VP_{\hat 1}VP_{[1,0]}V^\dag\}\\
                       &\;+\{VP_{\hat 2}VP_{\hat 1}VP_{[1,0]}V^\dag\}\\
                       &\;+\{VP_{\hat 1}VP_{\hat 2}VP_{[1,0]}V^\dag\}\\
                       &\;+2\{VP_{\hat 1}VP_{\hat 1}VP_{[2,0]}V^\dag\}\\
                       &\;-2\{VP_{\hat 1}VP_{\hat 1}VP_{[1,0]}V^\dag P_{\hat 1}V^\dag\}\\
                       &\;+\tfrac13\{VP_{\hat 3}VP_{[1,0]}V^\dag\}\\
                       &\;+\{VP_{\hat 2}VP_{[2,0]}V^\dag\}\\
                       &\;-\{VP_{\hat 2}VP_{[1,0]}V^\dag P_{\hat 1}V^\dag\}\\
                       &\;+\{VP_{\hat 1}V(P_{3,0}+P_{2,1}-P_{1,2}-P_{0,3})V^\dag\}\\
                       &\;-2\{VP_{\hat 1}VP_{[2,0]}V^\dag P_{\hat 1}V^\dag\}\\
                       &\;+2\{VP_{\hat 1}VP_{[1,0]}V^\dag P_{\hat 1}V^\dag P_{\hat 1}V^\dag\}\\ 
                       &\;-\{V(P_{3,0}+P_{2,1}-P_{1,2}-P_{0,3})V^\dag P_{\hat 1}V^\dag\}\\
                       &\;-\{VP_{[2,0]}V^\dag P_{\hat 2}V^\dag\}\\
                       &\;+2\{VP_{[2,0]}V^\dag P_{\hat 1}V^\dag P_{\hat 1}V^\dag\}\\
                       &\;-2\{VP_{[1,0]}V^\dag P_{\hat 1}V^\dag P_{\hat 1}V^\dag P_{\hat 1}V^\dag\}\\
                       &\;+\{VP_{[1,0]}V^\dag P_{\hat 2}V^\dag P_{\hat 1}V^\dag\}\\
                       &\;+\{VP_{[1,0]}V^\dag P_{\hat 1}V^\dag P_{\hat 2}V^\dag\}\\
                       &\;-\{VP_{\hat 1}VP_{[1,0]}V^\dag P_{\hat 2}V^\dag\}\\
                       &\;-\tfrac13\{VP_{[1,0]}V^\dag P_{\hat 3}V^\dag\}.
\end{align*}
Third, using~\eqref{eq:gsquared} and Lemma~\eqref{lemma:Poppeprodforms}, we observe that,
\begin{align*}
  \{f_x,f+g^2\}=&\;\{VP_{\hat 1}V(P_{1,1}-P_{0,2})V^\dag P_{\hat 1}V^\dag\}\\
                &\;-\{VP_{\hat 1}V(P_{1,1}-P_{0,2})VP_{\hat 1}V\}\\
                &\;+\{VP_{\hat 1}VP_{[1,0]}V^\dag(P_{2,0}+P_{1,1})V^\dag\}\\
                &\;-\{VP_{\hat 1}VP_{[1,0]}V^\dag(P_{2,0}+P_{1,1})V\}\\
                &\;-\{VP_{\hat 1}VP_{[1,0]}V^\dag P_{\hat 1}V^\dag P_{\hat 1}V^\dag\}\\
                &\;-\{VP_{\hat 1}VP_{[1,0]}V^\dag P_{[1,0]}VP_{\hat 1}V\}\\
                &\;+\{V(P_{2,1}-P_{0,3})V^\dag P_{\hat 1}V^\dag\}\\
                &\;-\{V(P_{2,1}-P_{0,3})VP_{\hat 1}V\}\\
                &\;+\{VP_{[2,0]}V^\dag(P_{2,0}+P_{1,1})V^\dag\}\\
                &\;-\{VP_{[2,0]}V^\dag(P_{2,0}+P_{1,1})V\}\\
                &\;-\{VP_{[2,0]}V^\dag P_{\hat 1}V^\dag P_{\hat 1}V^\dag\}\\
                &\;-\{VP_{[2,0]}V^\dag P_{[1,0]}VP_{\hat 1}V\}\\
                &\;-\{VP_{[1,0]}V^\dag(P_{1,1}+P_{0,2})V^\dag P_{\hat 1}V^\dag\}\\
                &\;+\{VP_{[1,0]}V^\dag(P_{1,1}+P_{0,2})VP_{\hat 1}V\}\\
                &\;-\{VP_{[1,0]}V^\dag P_{\hat 1}V^\dag(P_{2,0}+P_{1,1})V^\dag\}\\
                &\;+\{VP_{[1,0]}V^\dag P_{\hat 1}V^\dag(P_{2,0}+P_{1,1})V\}\\
                &\;+\{VP_{[1,0]}V^\dag P_{\hat 1}V^\dag P_{\hat 1}V^\dag P_{\hat 1}V^\dag\}\\
                &\;+\{VP_{[1,0]}V^\dag P_{\hat 1}V^\dag P_{[1,0]}VP_{\hat 1}V\}\\
                &\;-\{V(P_{1,1}+P_{0,2})VP_{\hat 1}V P_{[1,0]}V^\dag\}\\
                &\;-\{V(P_{1,1}+P_{0,2})V^\dag P_{\hat 1}V^\dag P_{[1,0]}V\}\\
                &\;-\{VP_{\hat 1}V(P_{2,0}+P_{1,1})V P_{[1,0]}V^\dag\}\\
                &\;-\{VP_{\hat 1}V(P_{2,0}+P_{1,1})V^\dag P_{[1,0]}V\}\\
                &\;-\{VP_{\hat 1}VP_{\hat 1}VP_{\hat 1}VP_{[1,0]}V^\dag\}\\
                &\;+\{VP_{\hat 1}VP_{[1,0]}V^\dag P_{\hat 1}V^\dag P_{[1,0]}V\}\\
                &\;-\{V(P_{1,1}+P_{0,2})VP_{[2,0]}V^\dag\}\\
                &\;+\{V(P_{1,1}+P_{0,2})V^\dag P_{[2,0]}V\}\\
                &\;-\{VP_{\hat 1}V(P_{3,0}-P_{1,2})V^\dag\}\\
                &\;+\{VP_{\hat 1}V(P_{3,0}-P_{1,2})V\}\\
                &\;-\{VP_{\hat 1}VP_{\hat 1}VP_{[2,0]}V^\dag\}\\
                &\;-\{VP_{\hat 1}VP_{\hat 1}V^\dag P_{[2,0]}V\}\\
                &\;+\{V(P_{1,1}+P_{0,2})VP_{[1,0]}V^\dag P_{\hat 1}V^\dag\}\\
                &\;+\{V(P_{1,1}+P_{0,2})V^\dag P_{[1,0]}VP_{\hat 1}V\}\\
                &\;+\{VP_{\hat 1}V(P_{2,0}-P_{1,1})V^\dag P_{\hat 1}V^\dag\}\\
                &\;+\{VP_{\hat 1}V(P_{2,0}-P_{1,1})VP_{\hat 1}V\}\\
                &\;+\{VP_{\hat 1}VP_{\hat 1}VP_{[1,0]}V^\dag P_{\hat 1}V^\dag\}\\
                &\;-\{VP_{\hat 1}VP_{[1,0]}V^\dag P_{[1,0]}VP_{\hat 1}V\}.
\end{align*}
Fourth, using Lemma~\eqref{lemma:Poppeprodforms}, we observe that,
\begin{align*}
  \{g,h_y\}=&\;\{V(P_{3,0}-P_{1,2})VP_{[1,0]}V^\dag\}\\
            &\;-\{V(P_{3,0}-P_{1,2})V^\dag P_{[1,0]}V\}\\
            &\;+\{VP_{\hat 1}VP_{[2,0]}VP_{[1,0]}V^\dag\}\\
            &\;+\{VP_{[1,0]}V^\dag P_{[2,0]}V^\dag P_{[1,0]}V\}\\
            &\;+\{V(P_{4,0}-2P_{3,1}+2P_{1,3}-P_{0,4})V^\dag\}\\
            &\;+\{V(P_{4,0}-2P_{2,2}+P_{0,4})V\}\\
            &\;+\{VP_{\hat 1}V(P_{3,0}-P_{2,1}-P_{1,2}+P_{0,3})V^\dag\}\\
            &\;-\{VP_{[1,0]}V^\dag(P_{3,0}-P_{2,1}-P_{1,2}+P_{0,3})V\}\\
            &\;-\{V(P_{2,0}-P_{1,1})V^\dag P_{[2,0]}V^\dag\}\\
            &\;+\{V(P_{2,0}-P_{1,1})VP_{[2,0]}V\}\\
            &\;-\{VP_{\hat 1}VP_{[1,0]}V^\dag P_{[2,0]}V^\dag\}\\
            &\;-\{VP_{[1,0]}V^\dag P_{[1,0]}VP_{[2,0]}V\}\\
            &\;-\{VP_{[2,0]}V(P_{1,1}-P_{0,2})V^\dag\}\\
            &\;+\{VP_{[2,0]}V(P_{1,1}-P_{0,2})V\}\\
            &\;+\{VP_{[2,0]}VP_{[1,0]}V^\dag P_{\hat 1}V^\dag\}\\
            &\;+\{VP_{[2,0]}VP_{[1,0]}V^\dag P_{[1,0]}V\}\\
            &\;+\{V(P_{3,0}-P_{2,1}-P_{1,2}+P_{0,3})V^\dag P_{\hat 1}V^\dag\}\\
            &\;+\{V(P_{3,0}-P_{2,1}-P_{1,2}+P_{0,3})V^\dag P_{[1,0]}V\}\\
            &\;+\{VP_{[1,0]}V^\dag(P_{2,1}-P_{0,3})V^\dag\}\\
            &\;-\{VP_{[1,0]}V^\dag(P_{2,1}-P_{0,3})V\}\\
            &\;-\{VP_{[1,0]}V^\dag P_{[2,0]}V^\dag P_{\hat 1}V^\dag\}\\
            &\;-\{VP_{[1,0]}V^\dag P_{[2,0]}V^\dag P_{[1,0]}V\}.
\end{align*}
Fifth, using Lemma~\eqref{lemma:Poppeprodforms}, we observe that,
\begin{align*}
-[g_x,h_x]=&\;-\{V(P_{1,1}+P_{0,2})VP_{\hat 1}VP_{[1,0]}V^\dag\}\\
&\;+\{V(P_{1,1}+P_{0,2})V^\dag P_{\hat 1}V^\dag P_{[1,0]}V\}\\
&\;-\{VP_{\hat 1}V(P_{2,0}+P_{1,1})VP_{[1,0]}V^\dag\}\\
&\;+\{VP_{\hat 1}V(P_{2,0}+P_{1,1})V^\dag P_{[1,0]}V\}\\
&\;-\{V P_{\hat 1}V P_{\hat 1}V P_{\hat 1}VP_{[1,0]}V^\dag\}\\
&\;-\{V P_{\hat 1}V P_{[1,0]}V^\dag P_{\hat 1}V^\dag P_{[1,0]}V\}\\
&\;-\{V(P_{1,1}+P_{0,2})VP_{[2,0]}V^\dag\}\\
&\;-\{V(P_{1,1}+P_{0,2})V^\dag P_{[2,0]}V\}\\
&\;-\{VP_{\hat 1}V(P_{3,0}-P_{1,2})V^\dag\}\\
&\;-\{VP_{\hat 1}V(P_{3,0}-P_{1,2})V\}\\
&\;-\{V P_{\hat 1}V P_{\hat 1}VP_{[2,0]}V^\dag\}\\
&\;+\{V P_{\hat 1}V P_{[1,0]}V^\dag P_{[2,0]}V\}\\
&\;+\{V(P_{1,1}+P_{0,2})VP_{[1,0]}V^\dag P_{\hat 1}V^\dag\}\\
&\;-\{V(P_{1,1}+P_{0,2})V^\dag P_{[1,0]}VP_{\hat 1}V\}\\
&\;+\{VP_{\hat 1}V(P_{2,0}-P_{1,1})V^\dag P_{\hat 1}V^\dag\}\\
&\;-\{VP_{\hat 1}V(P_{2,0}-P_{1,1})VP_{\hat 1}V\}\\
&\;+\{V P_{\hat 1}V P_{\hat 1}V P_{[1,0]}V^\dag P_{\hat 1}V^\dag\}\\
&\;+\{V P_{\hat 1}V P_{[1,0]}V^\dag P_{[1,0]}VP_{\hat 1}V\}\\
&\;+\{VP_{\hat 1}V(P_{1,1}-P_{0,2})V^\dag P_{\hat 1}V^\dag\}\\
&\;+\{VP_{\hat 1}V(P_{1,1}-P_{0,2})VP_{\hat 1}V\}\\
&\;+\{VP_{\hat 1}VP_{[1,0]}V^\dag(P_{2,0}+P_{1,1})V^\dag\}\\
&\;+\{VP_{\hat 1}VP_{[1,0]}V^\dag(P_{2,0}+P_{1,1})V\}\\
&\;-\{V P_{\hat 1}V P_{[1,0]}V^\dag P_{\hat 1}V^\dag P_{\hat 1}V^\dag\}\\
&\;+\{V P_{\hat 1}V P_{[1,0]}V^\dag P_{[1,0]}VP_{\hat 1}V\}\\
&\;+\{V(P_{2,1}-P_{0,3})V^\dag P_{\hat 1}V^\dag\}\\
&\;+\{V(P_{2,1}-P_{0,3})VP_{\hat 1}V\}\\
&\;+\{VP_{[2,0]}V^\dag(P_{2,0}+P_{1,1})V^\dag\}\\
&\;+\{VP_{[2,0]}V^\dag(P_{2,0}+P_{1,1})V\}\\
&\;-\{V P_{[2,0]}V^\dag P_{\hat 1}V^\dag P_{\hat 1}V^\dag\}\\
&\;+\{V P_{[2,0]}V^\dag P_{[1,0]}VP_{\hat 1}V\}\\
&\;-\{VP_{[1,0]}V^\dag(P_{1,1}+P_{0,2})V^\dag P_{\hat 1}V^\dag\}\\
&\;-\{VP_{[1,0]}V^\dag(P_{1,1}+P_{0,2})VP_{\hat 1}V\}\\
&\;-\{VP_{[1,0]}V^\dag P_{\hat 1}V^\dag(P_{2,0}+P_{1,1})V^\dag\}\\
&\;-\{VP_{[1,0]}V^\dag P_{\hat 1}V^\dag(P_{2,0}+P_{1,1})V\}\\
&\;+\{V P_{[1,0]}V^\dag P_{\hat 1}V^\dag P_{\hat 1}V^\dag P_{\hat 1}V^\dag\}\\
&\;-\{V P_{[1,0]}V^\dag P_{\hat 1}V^\dag P_{[1,0]}VP_{\hat 1}V\}.
\end{align*}
Sixth, again using Lemma~\eqref{lemma:Poppeprodforms}, we observe that,
\begin{align*}
-\bigl[f,\mathrm{D}(f+g^2)\bigr]=&\;-\{V(P_{1,1}-P_{0,2})V^\dag P_{[2,0]}V^\dag\}\\
&\;+\{V(P_{1,1}-P_{0,2})VP_{[2,0]}V\}\\
&\;-\{VP_{[1,0]}V^\dag(P_{3,0}-P_{1,2})V^\dag\}\\
&\;+\{VP_{[1,0]}V^\dag(P_{3,0}-P_{1,2})V\}\\
&\;+\{VP_{[1,0]}V^\dag P_{\hat 1}V^\dag P_{[2,0]}V^\dag\}\\
&\;+\{VP_{[1,0]}V^\dag P_{[1,0]}VP_{[2,0]}V\}\\
&\;-\{V(P_{2,1}-P_{0,3})VP_{[1,0]}V^\dag\}\\
&\;+\{V(P_{2,1}-P_{0,3})V^\dag P_{[1,0]}V\}\\
&\;-\{VP_{[2,0]}V(P_{2,0}-P_{1,1})V^\dag\}\\
&\;+\{VP_{[2,0]}V(P_{2,0}-P_{1,1})V\}\\
&\;-\{VP_{[2,0]}VP_{\hat 1}VP_{[1,0]}V^\dag\}\\
&\;-\{VP_{[2,0]}VP_{[1,0]}V^\dag P_{[1,0]}V\}.
\end{align*}
Putting all the terms above together shows that $h\coloneqq[VP_{[1,0]}V^\dag]$, $g\coloneqq[V]$ and $f\coloneqq-\{VP_{[1,0]}V^\dag\}$ satisfy
the lifted mKP equation~\eqref{eq:triplemKPf}, thus establishing Theorem~\ref{thm:mKP}.

\end{document}